%% file: main.tex
\documentclass[a4paper,USenglish,thm-restate,cleveref,numberwithinsect]{lipics-v2021}

\input{defs}

\pdfoutput=1
\nolinenumbers

\newif\iffull

\fulltrue

\author{Armando Casta\~neda}
{Instituto de Matemáticas, Universidad Nacional Autónoma de México}
{armando.castaneda@im.unam.mx}
{https://orcid.org/0000-0002-8017-8639}
{Supported by DGAPA PAPIIT project IN108723 and Royal Society grant  IES$\backslash$R1$\backslash$221226.}

\author{Gregory Chockler}
{Department of Computer Science, University of Surrey}
{g.chockler@surrey.ac.uk}
{https://orcid.org/0000-0001-6700-9235}
{Supported by Royal Society grant: IES$\backslash$R1$\backslash$221226; CHIST-ERA project REDONDA EP/Y036425/1;
and EPSRC grants: EP/X037142/1 and EP/X015149/1.}

\author{Brijesh Dongol}
{Department of Computer Science, University of Surrey}
{b.dongol@surrey.ac.uk}
{https://orcid.org/0000-0003-0446-3507}
{
Supported by VeTSS; Royal Society grant: IES$\backslash$R1$\backslash$221226; CHIST-ERA project REDONDA EP/Y036425/1;
and EPSRC grants: EP/X037142/1, EP/X015149/1, EP/V038915/1, and 
EP/R025134/2.
}

\author{Ori Lahav}
{School of Computer Science, Tel Aviv University}
{orilahav@tau.ac.il}
{https://orcid.org/0000-0003-4305-6998}
{Supported by the European Research Council (ERC) under the European Union's Horizon 2020
research and innovation programme (grant agreement no.~851811)
and the Israel Science Foundation (grant number~814/22).}

\authorrunning{A. Casta\~neda, G. Chockler, B. Dongol, and O. Lahav}

\Copyright{Armando Casta\~neda, Gregory Chockler, Brijesh Dongol, and Ori Lahav}

\iffull
\else
\relatedversion{}\relatedversiondetails{Full Version}{https://arxiv.org/abs/2405.16611}
\fi

\begin{document}



\title{What Cannot Be Implemented on Weak Memory?}

\maketitle
\allowdisplaybreaks

\begin{abstract}
  We present a general methodology for establishing the impossibility
  of implementing certain concurrent objects on different (weak)
  memory models.  The key idea behind our approach lies in
  characterizing memory models by their \emph{mergeability
    properties}, identifying restrictions under which independent
  memory traces 
  can be merged into a single valid memory trace.  In turn, we show
  that the mergeability properties of the underlying memory model
  entail similar mergeability requirements on the specifications of
  objects that can be implemented on that memory model.  We
  demonstrate the applicability of our approach to establish the
  impossibility of implementing 
  standard distributed objects
  with different restrictions on memory traces on three memory models:
  strictly consistent
  memory, total store order, and release-acquire.  
  These impossibility results allow us to identify tight and almost tight bounds for some objects,
  as well as new separation results between weak memory models,
  and between well-studied 
  objects based on their implementability on weak memory models.
\end{abstract}

\input{intro-disc}

\input{memory}
\input{merge}
\input{impos-obj}
\input{apps-summary}
\input{related}

\bibliography{biblio}

\clearpage
\appendix

\iffull
 \input{summary_app.tex}

 \input{merge_app.tex}
 \input{impos_app.tex}
 \input{apps-strong}
 \input{simulation}
\fi
 \input{max-RA.tex}
 \input{apps-weak}
 \input{snap-counter-app}
 \iffull
 \input{cons-num-two-sided.tex}
 \fi

\end{document}

%% file: defs.tex
\usepackage{mathpartir}
\usepackage[inline]{enumitem}
\usepackage{xspace}
\usepackage{array}
\usepackage{multicol}
\usepackage{makecell}

\usepackage{algorithm}
\usepackage[noend]{algpseudocode}

\usepackage{float}

\floatstyle{ruled}
\newfloat{Example}{t}{example}

\hypersetup{
  colorlinks=true,breaklinks,
  linkcolor=green!50!black,
  citecolor=blue
}

\AtEndPreamble{%

\newtheorem*{notation*}{Notation}
}

\newcommand{\citet}[1]{{\bf \color{red}Fix citet }\cite{#1}}

\newcommand{\mergethm}{Merge Theorem\xspace}

\DeclareSymbolFont{Shuffle}{U}{shuffle}{m}{n}
\DeclareFontFamily{U}{shuffle}{}
\DeclareFontShape{U}{shuffle}{m}{n}{%
  <-8>shuffle7%
  <8->shuffle10%
}{}
\DeclareMathSymbol\shuffle{\mathbin}{Shuffle}{"001}
\DeclareMathSymbol\cshuffle{\mathbin}{Shuffle}{"002}

\Crefname{Example}{Example}{Examples}
\crefname{Example}{Example}{Examples}

\crefformat{section}{#2\S{}#1#3}
\Crefname{section}{Section}{Section}
\Crefformat{section}{Section #2#1#3}
\crefname{corollary}{\text{Corollary}}{\text{corollaries}}
\Crefname{corollary}{\text{Corollary}}{\text{Corollaries}}
\crefname{lemma}{\text{Lem.}}{\text{Lemmas}}
\Crefname{lemma}{\text{Lem.}}{\text{Lemmas}}
\crefname{proposition}{\text{Prop.}}{\text{Propositions}}
\Crefname{proposition}{\text{Proposition}}{\text{Propositions}}
\crefname{definition2}{\text{Def.}}{\text{Definitions}}
\Crefname{definition2}{\text{Def.}}{\text{Definitions}}
\crefname{notation}{\text{Notation}}{\text{Notations}}
\Crefname{notation}{\text{Notation}}{\text{Notations}}
\crefname{theorem}{\text{Thm.}}{\text{Theorems}}
\Crefname{theorem}{\text{Thm.}}{\text{Theorems}}
\crefname{figure}{\text{Fig.}}{\text{Figures}}
\Crefname{figure}{\text{Fig.}}{\text{Figures}}
\crefname{example}{\text{Example}}{\text{Examples}}
\Crefname{example}{\text{Example}}{\text{Examples}}
\crefformat{appendix}{#2\S{}#1#3}

\theoremstyle{definition}
\newtheorem{definition2}[theorem]{Definition}

\renewcommand{\paragraph}[1]{\smallskip\noindent{\bfseries #1}}

\newcommand{\inarr}[1]{
  \begin{array}[t]{@{}l@{}}
    #1
  \end{array}
}
\newcommand{\inarrC}[1]{
  \begin{array}[t]{@{}c@{}}
    #1
  \end{array}
}

\newcommand{\ie}{i.e.,\@\xspace}
\newcommand{\eg}{e.g.,\@\xspace}

\newcommand{\wrt}{w.r.t.\@\xspace}
\newcommand{\aka}{a.k.a.\@\xspace}

\newcommand{\suq}{\subseteq}
\newcommand{\N}{\mathbb{N}}
\newcommand{\emptyseq}{\varepsilon}
\newcommand{\ite}[3]{\textbf{if } #1 \textbf{ then } #2 \textbf{ else } #3}

\newcommand{\nin}{\not\in}
\newcommand{\set}[1]{\{{#1}\}}

\newcommand{\size}[1]{|{#1}|}
\renewcommand{\st}{\; | \;}

\newcommand{\tup}[1]{{\langle{#1}\rangle}}

\newcommand{\defeq}{\triangleq}

\newcommand{\restrict}[2]{{#1}|_{#2}}
\newcommand{\True}{{\it true}}
\newcommand{\False}{{\it false}}
\newcommand{\maketil}[1]{{#1}\ldots{#1}}
\newcommand{\til}{\maketil{,}}

\newcommand{\init}{{\mathsf{init}}}

\newcommand{\reorder}[2]{\mathsf{reorder}_{#1}(#2)}

\newcommand{\lts}{L}
\newcommand{\states}[1]{{\mathsf{states}({#1})}}
\newcommand{\initial}[1]{{\mathsf{init}({#1})}}
\newcommand{\trans}[1]{{\mathsf{trans}({#1})}}
\newcommand{\labels}[1]{{\mathsf{labels}({#1})}}
\newcommand{\traces}[1]{{\mathsf{traces}({#1})}}
\newcommand{\otraces}[1]{{\mathsf{otraces}({#1})}}

\newcommand{\sysp}[2]{{\mathsf{S}_#1^#2}}
\newcommand{\sys}[1]{{\mathsf{S}_#1}}
\newcommand{\implo}{I}
\newcommand{\impl}{\mathcal{I}}
\newcommand{\obj}{O}
\newcommand{\proc}{p}

\newcommand{\op}{o}
\newcommand{\act}{a}
\newcommand{\evt}{e}
\newcommand{\var}{x}

\newcommand{\cproc}{\mathtt{p}}

\newcommand{\Loc}{{\mathsf{Loc}}}
\newcommand{\Val}{{\mathsf{Val}}}
\newcommand{\Var}{{\mathsf{Var}}}
\newcommand{\Proc}{{\mathsf{P}}}

\newcommand{\memacts}{\mathsf{MemActs}}

\newcommand{\Reads}{\mathsf{R}}

\newcommand{\Fences}{\mathsf{F}}
\newcommand{\MemEvs}{\mathsf{MemEvs}}
\newcommand{\ObjEvs}{\mathsf{Evs}}

\newcommand{\Read}{\texttt{R}}
\newcommand{\Write}{\texttt{W}}
\newcommand{\Fence}{\texttt{F}}
\newcommand{\RMW}{\texttt{RMW}}

\newcommand{\ack}{{\normalfont\texttt{ack}}}
\newcommand{\vold}{v_\textsf{old}}
\newcommand{\vexp}{v_\textsf{exp}}
\newcommand{\vnew}{v_\textsf{new}}

\newcommand{\mm}{M}
\newcommand{\SCM}{{\text{{\normalfont SCM}}}\xspace}
\newcommand{\TSO}{{\text{{\normalfont TSO}}}\xspace}
\newcommand{\RA}{{\text{{\normalfont RA}}}\xspace}
\newcommand{\Atomic}{\SCM}

\newcommand{\invi}[1]{{\mathtt{inv}(#1)}}
\newcommand{\resi}[1]{{\mathtt{res}(#1)}}
\newcommand{\opsf}[1]{{\mathsf{ops}(#1)}}
\newcommand{\retsf}[1]{{\mathsf{rets}(#1)}}
\newcommand{\ops}{{\mathit{ops}}}
\newcommand{\rets}{{\mathit{rets}}}

\newcommand{\procsf}[1]{\mathsf{proc}\mbox{-}\mathsf{set}({#1})}
\newcommand{\actsf}[1]{{\mathsf{acts}}({#1})}
\newcommand{\tracef}[1]{{\mathsf{trace}({#1})}}
\newcommand{\histories}[1]{{\mathsf{H}({#1})}}
\newcommand{\chistories}[1]{{\mathsf{ComH}({#1})}}
\newcommand{\cshistories}[1]{{\mathsf{ComSeqH}({#1})}}
\newcommand{\procf}[1]{{\mathsf{proc}({#1})}}
\newcommand{\actionf}[1]{{\mathsf{act}({#1})}}
\newcommand{\varf}[1]{{\mathsf{var}({#1})}}
\newcommand{\typf}[1]{{\mathsf{typ}({#1})}}
\newcommand{\valwf}[1]{{\mathsf{val}({#1})}}
\newcommand{\typsf}[1]{\mathsf{typ}\mbox{-}\mathsf{set}({#1})}

\newcommand{\spec}[1]{{\mathsf{spec}({#1})}}

\newcommand{\mra}{m}
\newcommand{\View}{\mathsf{View}}
\newcommand{\pv}{{\mathcal{V}_\mathsf{p}}}

\newcommand{\fv}{{{V}_\mathsf{f}}}
\newcommand{\view}{V}
\newcommand{\msg}{\mu}

\newcommand{\valmf}[1]{{\mathsf{val}({#1})}}

\newcommand{\timef}[1]{{\mathsf{ts}({#1})}}
\newcommand{\viewf}[1]{{\mathsf{view}({#1})}}

\newcommand{\h}{h}
\newcommand{\seqh}{h_\mathsf{seq}}
\newcommand{\sproc}{\mathsf{sproc}}
\newcommand{\lin}{\mathsf{lin}}

\newcommand{\restrictmem}[1]{\restrict{#1}{\mathsf{M}}}
\newcommand{\restrictobj}[2][\obj]{\restrict{#2}{#1}}

\newcommand{\ainv}[1]{{{\stackrel{#1}{}}}}
\newcommand{\winv}[1]{{\stackrel{\updateopp{#1}}{}}}
\newcommand{\rinv}{{{\stackrel{\scanop}{}}}}
\newcommand{\rval}[1]{{{\stackrel{ {#1}}{}}}}
\newcommand{\wint}[2]{\;|\!$\parbox[t][][t]{#2}{$\,\winv{#1}\hfill\rval{\ack\,}$\vspace*{-1pt}{\color{black}\hrule height 1pt}}$\!|\;}
\newcommand{\rint}[2]{\;|\!$\parbox[t][][t]{#2}{$\,\rinv\hfill\rval{#1}\,$\vspace*{-1pt}{\color{black}\hrule height 1pt}}$\!|\;}
\newcommand{\hint}[3]{\;|\!$\parbox[t][][t]{#3}{$\,\ainv{#1}\hfill\rval{#2}\,$\vspace*{-1pt}{\color{black}\hrule height 1pt}}$\!|\;}

\newcommand{\invres}[3]{({#1}\text{:}\hint{#2\vphantom{)}}{#3\vphantom{)}}{20pt}\!)}
\newcommand{\invress}[4]{({#1}\text{:}\hint{#2\vphantom{)}}{#3\vphantom{)}}{#4}\!)}

\newcommand{\ev}[2]{{#1}\text{:}{#2}}

\newcommand{\locko}{\mathsf{Lock}}
\newcommand{\acquireo}{\mathtt{acquire}}

\newcommand{\snapshoto}{\mathsf{Snapshot}}
\newcommand{\updateop}{{\tt update}}
\newcommand{\updateopp}[1]{\updateop(#1)}
\newcommand{\scanop}{{\tt scan}}
\newcommand{\snapop}{\scanop}
\newcommand{\writesno}[1]{\updateop(#1)}

\newcommand{\rego}{\mathsf{Reg}}
\newcommand{\writeo}[1]{\mathtt{write}(#1)}
\newcommand{\writeoo}{\mathtt{write}}
\newcommand{\reado}{\mathtt{read}}

\newcommand{\mrego}{\mathsf{MaxReg}}

\newcommand{\mcounto}{\mathsf{MC}}
\newcommand{\inco}{\mathtt{inc}}

\newcommand{\counto}{\mathsf{Counter}}
\newcommand{\deco}{\mathtt{dec}}

\newcommand{\seto}{\mathsf{Set}}
\newcommand{\putop}{{\tt add}}
\newcommand{\takeop}{{\tt remove}}

\newcommand{\conso}{\mathsf{Consensus}}
\newcommand{\proposeop}[1]{{\tt propose}(#1)}
\newcommand{\proposeopp}{{\tt propose}}


%% file: intro-disc.tex
\section{Introduction}

Weak memory models have become standard in modern hardware
architectures and programming languages. Unlike traditional strictly
consistent memory (SCM), which provides \emph{atomic} read/write instructions, 
memories achieve efficiency by multiple optimizations, which, in
particular, delay propagation of writes instead of making them
immediately visible to subsequent reads in other threads. Two well-studied 
models, which we consider in this paper, are \emph{total store
  order} model (TSO), as implemented in
SPARC~\cite{sparc-tso,guide2011intel} and x86
multiprocessors~\cite{x86-tso}, and the weaker \emph{release-acquire} model (RA), 
a fragment of C/C++11~\cite{Batty:2011,sra}, which guarantees 
causal consistency together with per-location strict consistency (\aka  coherence).

The standard memory model for the design and analysis of asynchronous shared memory algorithms is SCM.
These algorithms however, are not guaranteed to work correctly on weaker memory models
(such as TSO and RA) due to the lack of atomicity of reads and writes. 
To ensure atomicity, one can use \emph{fence} or atomic \emph{read-modify-write (RMW)} instructions 
provided by the weak memory models. However, since fences and RMWs disable hardware optimizations and 
enforce synchronization between threads, they incur substantial performance overheads.
Thus, one would like to understand when fences and RMWs are necessary 
and when they can be avoided, in order to correctly and efficiently implement the large 
body of existing shared memory algorithms on weak memory architectures.

In this paper, we set out to tackle this important and challenging
question. The crux of our approach 
is
based on \emph{mergeability} of traces and object
histories.  Roughly
speaking, two 
memory
traces (sequences of memory accesses)
of some memory model $M$ are strongly (resp., weakly) mergeable
if every (resp., some) interleaving of these traces forms a valid
trace of $M$. Likewise, two object histories (sequences of invocations
and responses) of some object $O$ are strongly (resp., weakly
mergeable) if every (resp., some) interleaving of these histories
forms a valid history of $O$.  Then, our key result is the \emph{\mergethm}, 
which, roughly speaking, states that strongly (resp., weakly)
mergeable memory traces can only be used to implement strongly
(resp., weakly) mergeable object histories. 
Contrapositively, when operations of a certain concurrent
object are not strongly (resp., weakly) mergeable, then the
memory traces implementing these operations on a memory model $M$
cannot be strongly (resp., weakly) mergeable in $M$.
The correctness and progress conditions in the \mergethm are weaker versions of
 \emph{linearizability}~\cite{HW90} and \emph{obstruction-freedom}~\cite{obstruction-free}.

A prerequisite for applying our \mergethm for a particular memory
model is to identify useful mergeability properties of the model.  For
\SCM, \TSO, and \RA, we develop a set of properties (see
\cref{tab:results}) that describe conditions under which traces
of the models can be (weakly/strongly) merged.  These results provide
key insights into the synchronization power of these
memory models, and together with the \mergethm allow us to derive multiple impossibility results, and
identify optimal implementations.

\begin{Example}
\newcommand{\pinv}[2]{{\stackrel{\putop({#2})}{}}}
\newcommand{\tinv}{{{\stackrel{\takeop{({1}})}{}}}}
\newcommand{\pint}[3]{\;|\!$\parbox[t][][t]{#3}{$\,\pinv{#1}{#2}\hfill\rval{\ack\,}$\vspace*{-1pt}{\color{black}\hrule height 1pt}}$\!|\;}
\newcommand{\tint}[2]{\;|\!$\parbox[t][][t]{#2}{$\,\tinv\hfill\rval{#1}\,$\vspace*{-1pt}{\color{black}\hrule height 1pt}}$\!|\;}

  \small
 Consider a set object that provides the high-level operations
 $\putop(v)$ and $\takeop(v)$, where $\takeop$ returns $\True$ iff the
 element $v$ is in the abstract set and in this case removes $v$ from the set.
Consider the following histories assuming two processes:
$$
\inarr{\begin{array}{ll}
\cproc_1: & \pint{\cproc_1}{{{1}}}{50pt}   \tint{{{\True}}}{50pt} \\ 
\cproc_2: &  \\
\end{array}
\qquad\qquad
\begin{array}{ll}
\cproc_1: & \pint{\cproc_1}{{{1}}}{50pt}   \\ 
\cproc_2: & \qquad\qquad\qquad \quad \tint{{{\True}}}{50pt} \\ 
\end{array} \hfill \\[1ex]
\qquad\qquad\qquad \text{history $\h_1$} \qquad\qquad\qquad\qquad\qquad\qquad\qquad\qquad \text{history $\h_2$}} \qquad
$$ 
Let $\sigma_0$ be the trace of a set implementation $\implo$ generated
by $\cproc_1$ executing $\putop(1)$ until completion from the initial
state, and for $i\in\set{1,2}$, let $\sigma_i$ be the trace generated
by $\cproc_i$ after $\sigma_0$ to induce history $\h_i$.
Such traces must exist assuming $\implo$ is obstruction-free.
If $\sigma_1$ and $\sigma_2$ can be merged into a trace
$\sigma$ such that $\sigma_0\cdot \sigma$ is a valid trace of
a memory model $M$, then we reach a contradiction because $\cproc_1$ (resp.,
$\cproc_2$) cannot distinguish between $\sigma_1$ (resp.,
$\sigma_2$) and $\sigma$, and thus both $\takeop$ operations of
$\cproc_1$ and $\cproc_2$ in $\sigma$ return $\True$, contradicting
linearizability of $\implo$.  In other words, since the two
${\takeop}$ invocations cannot be merged into a single linearizable
object history, it must be that the corresponding memory traces cannot be merged.
In particular, if $\sigma_1$ and $\sigma_2$ have neither RAW nor RMW,
then they can always be merged on \SCM, which gives us the impossibility
result of~\cite{AGHKMV11} for this object.
    \caption{Linearizable Obstruction-Free Set}
  \label{ex:LOO-intro}
\end{Example}

For instance, consider the read-after-write pattern (RAW),
which is often used by shared memory algorithms under SCM (such as classical
mutual exclusion~\cite{ewd123,bakery}) as a 
synchronization mechanism. In RAW, 
a process first writes to a shared variable and then reads from a different shared variable,
and under SCM, this ensures that at least one
of the two processes writing to two different variables has
to observe the value written by the other process 
(see the SB program in \cref{sec:memory}).
This means that solo traces that use RAW are not mergeable into a single trace.
In turn, it is straightforward to establish that any two RAW-free 
read-write traces (by distinct processes) are weakly mergeable under SCM (\cref{sec:merge}).

With this observation, we easily re-establish (and generalize) 
the ``Laws of Order'' results from~\cite{AGHKMV11},
showing that mutual exclusion protocols, as well as concurrent objects with \emph{strongly non-commutative} methods,
cannot be implemented on SCM with neither RAW nor RMW.
We do so by simple mergeability-based arguments (see, \eg~\cref{ex:LOO-intro}),
instead of rather complex and ad-hoc application of the covering technique used in~\cite{AGHKMV11}.
Intuitively, two methods are strongly non-commutative if executing one of them first 
affects the response of the other, and vice versa.
Moreover, by using mergeability properties for TSO and RA we directly obtain
similar impossibility results for these models, whereas the argument in~\cite{AGHKMV11}
for weak memory models is only implicit, based on the fact that 
enforcing a write to be executed before a read (i.e., implementing RAW) on a weak model requires a fence.

A benefit of our generic methodology is that
we can also reason about implementability of methods that 
are not strongly non-commutative, hence not covered by~\cite{AGHKMV11}:

\paragraph{One-Sided Non-Commutative Operations.}
Some objects such as register, max-register, snapshot and monotone counter
have pairs of methods 
that do not strongly non-commute.
To support them, we consider \emph{one-sided non-commutativity} of pairs of methods,
which, roughly speaking, means that executing one of them first affects the response of the other,
but not necessarily vice versa. 
We then apply the \mergethm to show that any linearizable obstruction-free
implementations of these objects must use fences or RMWs in TSO and RA.

Then, for max-register, a useful building block in several implementations, \eg~\cite{AAC12, BaigHMT23, CP21},
we obtain \emph{fence-optimal} implementations in TSO and RA. 
The TSO implementation is obtained through a more general \emph{fence-insertion strategy}:
a transformation that takes any read/write linearizable implementation in SCM 
and adds fences between every write followed by a read or a return of an operation, 
provably resulting in a linearizable implementation in TSO.
Combined with a wait-free read/write max-register implementation in SCM (with uses neither RAW nor RMW),
the transformation gives a fence-optimal wait-free read/write max-register implementation
in TSO.
For RA, we develop a similar linearizable implementation 
by placing a fence in the beginning and the end of every operation, 
which leads to a fence-optimal implementation of max-register in RA.

\paragraph{Snapshot and Counter.}
We also reason about snapshot and (non-monotone) counter,
which fall beyond the scope of non-commutativity.
These two objects are of particular interests:
snapshot is \emph{universal} for a family of objects whose pairs of
operations either commute or one overwrites the other~\cite{AH90}, 
and counter is a useful building block for randomized consensus~\cite{aspness-random2, AspnesH1990consensus}.
For TSO, the fence-insertion transformation above 
once again provides a wait-free fence-optimal snapshot (resp., counter) implementation 
where every update operation ends with a fence.
However, we use our \mergethm to show that, in sharp contrast to max-register, 
there is no obstruction-free read/write snapshot (resp., counter) implementation in RA,
whose operations start with a fence and end with a fence (see outline in \cref{ex:snapshot-intro}).
To the best of our knowledge, this is the first sharp separation between max-register
on the one hand and snapshot and counter on the other in terms of their 
implementability under $\RA$ using only reads, writes and fences.

\begin{Example}
  \small
  Mergeability can justify a novel impossibility result for \RA,
  showing that a shared (single-writer multi-reader) snapshot object cannot be implemented with
  only reads, writes and fences under the restriction that all fences
  are only placed at the beginning and end of a method invocation.
Consider the following histories assuming three processes:
$$
\inarr{\begin{array}{ll}
\cproc_1: & \wint{1}{80pt} \\ 
\cproc_2: &  \\ 
\cproc_3: & \quad \rint{\tup{{{1}},\bot,\bot}}{80pt}
\end{array}
\qquad\qquad\qquad
\begin{array}{ll}
\cproc_1: & \\ 
\cproc_2: &  \wint{1}{70pt} \rint{\tup{\bot,{{1}},\bot}}{70pt} \\ 
\cproc_3: & 
\end{array} \hfill \\[1ex]
\qquad\qquad\qquad \text{history $\h_1$} \qquad\qquad\qquad\qquad\qquad\qquad\qquad \text{history $\h_2$}} \qquad
$$
An obstruction-free implementation should generate both histories.
A merge-based argument implies that the memory traces $\sigma_1$ and $\sigma_2$ induced by the implementation 
when it generates $\h_1$ and $\h_2$ must not be mergable in the underlying memory model.
Otherwise, the same algorithm will also allow some interleaving $\h$ of $\h_1$ and $\h_2$,
but it is easy to observe that no such interleaving is linearizable:
no valid single history $\h$ with only two updates, 
$\updateopp{1}$ by $\cproc_1$ and $\updateopp{1}$ by $\cproc_2$,
can have \emph{both} scan results $\tup{{1},\bot,\bot}$ and $\tup{\bot,{1},\bot}$.
The $\RA$ memory model allows any two RMW-free traces $\sigma_1$ and $\sigma_2$ by disjoint sets of processes
to be merged, provided that fences are not used in the middle of these
traces.  Roughly speaking, following~\cite{sra,promising}, the
semantics of $\RA$ is based on point-to-point communication, making it
is possible for $\cproc_1$ and $\cproc_3$ to communicate directly,
without affecting $\cproc_2$.
Thus, every implementation of snapshot on $\RA$ 
uses RMWs or fences in the middle of operations.
    \caption{Linearizable Obstruction-Free Snapshot}
  \label{ex:snapshot-intro}
\end{Example}

\bigskip
\paragraph{Outline.}
The rest of this paper is structured as follows.
In \cref{sec:memory} we define the notion of a memory model.
In \cref{sec:merge} we establish multiple mergeability properties for these memory models.
In \cref{sec:impos-obj} we present the general impossibility result.
In \cref{sec-first-applications} we discuss applications of the theorem for well known objects,
and tightness of the obtained lower bounds.
We conclude and discuss related work in \cref{sec:related}.


%% file: memory.tex
\section{Weak Memory Models}
\label{sec:memory}

In this paper, we consider three memory models: 

\begin{description}[leftmargin=0pt]
\item [Strictly Consistent Memory ($\SCM$):]
In this model every write is propagated to all
threads immediately after being executed. 
In the weak memory literature, this
  memory model is often referred to as sequential
  consistency, but it essentially corresponds to a
  collection of \emph{linearizable} (\aka atomic) register objects~\cite{HW90}.
\item [Total Store Order ($\TSO$):]
Each process has a local FIFO store buffer.  Writes are first enqueued
in the buffer of the writing process, and later propagate from
the buffer to main memory in an \emph{internal} step that occurs
non-deterministically as part of the system's execution. 
A read of a variable returns the latest write to the variable in the
reading process' buffer or the value in
main memory if there is no pending write to that variable in the buffer.
\item [Release/Acquire ($\RA$):] This model
 employs a notion of synchronization between
processes through acquiring instructions (read or RMW) which
synchronize with previously executed releasing instructions (write or RMW) 
when the acquiring instruction reads its value from the releasing
instruction. Such synchronization transfers
``happens-before'' knowledge from the releasing instruction to the
acquiring instruction. Following a release-acquire synchronization,
instructions that follow (in ``happens-before'' order) the acquire
instruction must be consistent with the happens-before knowledge
received through the synchronization.
\end{description}

The classic examples used to explain these memory models 
are the \emph{store buffering} (SB), 
\emph{independent reads of independent writes} (IRIW), 
and \emph{message passing} (MP) programs,
given below.  We assume shared variables $x$ and
$y$ initialized with the value $0$ and process-local variables $a,b,\ldots$. 
The possible final values of $a,b,\ldots$ depend on the memory model.

\medskip
\noindent\hfill\begin{tabular}[t]{@{}c@{\qquad}c@{\qquad}c@{}} \small
  $\inarrC{\begin{array}[c]{@{}l||l@{}}
    \inarr{\text{Proc $\cproc_1$} \\ x := 1; \\ a := y;}
    & \inarr{\text{Proc $\cproc_2$} \\ y := 1; \\ b := x;}
   \end{array}\\  \text{(SB)}}$
  & \small
  $\inarrC{\begin{array}[c]{@{}l||l||l||l@{}}
    \inarr{\text{Proc $\cproc_1$} \\ x := 1;}
    & \inarr{\text{Proc $\cproc_2$} \\ a := x; \\ b := y;}
    & \inarr{\text{Proc $\cproc_3$} \\ c := y; \\ d := x;}
    & \inarr{\text{Proc $\cproc_4$} \\ y := 1; }
   \end{array}\\  \text{(IRIW)}}$
  & \small
  $\inarrC{\begin{array}[c]{@{}l||l@{}}
    \inarr{\text{Proc $\cproc_1$} \\ x := 1; \\ y := 1; }
    & \inarr{\text{Proc $\cproc_2$} \\ a := y; \\ b := x;}
   \end{array}\\  \text{(MP)}}$
\end{tabular}\hfill {} 
      
\medskip
Under \Atomic, no execution of SB ends with $a = b = 0$, while this
outcome is possible under both \TSO and \RA. Under both \Atomic and
\TSO, no execution of IRIW ends with $a = c = 1$ and $b = d = 0$,
while this outcome is possible under \RA, indicating that under \RA,
processes $\cproc_2$ and $\cproc_3$ observe the writes to $x$ and $y$
in a different order. In particular, under \RA, suppose that both
$\cproc_1$ and $\cproc_4$ execute their writes. It is possible for
$\cproc_2$ (resp., $\cproc_3$) to read the new value for $x$ (resp.,
$y$) then read the old value for $y$ (resp., $x$). Although \RA is
weaker than both \SCM and \TSO, like \TSO, \RA 
maintains causal consistency as demonstrated MP. Under all three
memory models, when MP terminates, if $a=1$, then $b=1$, indicating
that if $\cproc_2$ is aware of the write to $y$ by $\cproc_1$, then it
must also be aware of the prior write to $x$.

Non-\Atomic-outcomes (\aka \emph{weak behaviors}) can be avoided in weak memory models by using
\emph{fence} instructions. In \TSO
fences drain the store buffer of the process that executes the fence.
In \RA fences synchronize \emph{in pairs}, transferring happens-before
knowledge from one process to another. We formally include fences
also in \SCM (with ``no-op'' semantics).

\subsection{Formalizing Weak Memory Models}

For the formal definitions of the models,
we find it most convenient to follow an operational
presentation, where memory models are specified by labeled transition systems.

\begin{description}[leftmargin=0pt,itemsep=2pt]
\item[Sequences.] For a sequence $s=\tup{x_1 \til x_n}$, $s[i]$
  denotes the $i$th element of $s$ (\ie $x_i$), and $\size{s}$
  denotes the length of $s$ (\ie $n$).  We write $x\in s$ when $s[i]=x$ for some
  $1 \leq i \leq n$.  We denote by $\emptyseq$ the empty sequence,
  write $s_1 \cdot s_2$ for concatenation of $s_1$ and $s_2$ and
  denote by $X^*$ the set of all sequences over elements of a set
  $X$.  The restriction of a sequence $s$ \wrt a set $Y$, denoted
  $\restrict{s}{Y}$, is the longest subsequence of $s$ that consists
  only of elements in $Y$. These notations are lifted to sets in the obvious way
  (\eg $S \cdot s' \defeq \set{s \cdot s' \st s \in S}$ and $\restrict{S}{Y} \defeq \set{\restrict{s}{Y} \st s \in S}$).
We use the suffix
`$\mbox{-}\mathsf{set}$' to lift a function $f$ from some set $X$ to a
function form sequences over $X$, formally defined by:
$f\mbox{-}\mathsf{set}(s)\defeq \set{f(s[i]) \st 1 \leq i \leq \size{s}}$.
  
\item[Labeled Transition Systems (LTSs).]  An
  LTS 
  $\lts$ consists of a set of states, $\states{\lts}$; an initial
  state, $\initial{\lts}\in\states{\lts}$; a set of transition labels,
  $\labels{\lts}$; and a set of transitions, 
  $\trans{\lts} \subseteq  \states{\lts} \times \labels{\lts} \times \states{\lts}$.  
  We write
  $q \xrightarrow{l}_\lts q'$ for 
  $\tup{q, l, q'}\in\trans{\lts}$, and given
  $\pi=\tup{l_1 \til l_n}\in \labels{\lts}^*$, we write
  $q \xrightarrow{\pi}_\lts q'$ for
  $\exists q_2 \til q_n \ldotp q \xrightarrow{l_1}_\lts q_2
  \xrightarrow{l_2}_\lts \ldots q_n
  \xrightarrow{l_n}_\lts q'$.  
  An \emph{execution fragment} of $\lts$ is a sequence
$\alpha=\tup{q_0, l_1, q_1, l_2 \til l_n, q_n}$ of alternating states and transition labels 
such that $q_i \xrightarrow{l_{i+1}}_\lts q_{i+1}$ for every $0 \leq i \leq n-1$.
The \emph{trace} of $\alpha$, denoted
$\tracef{\alpha}$, is the restriction of $\alpha$ \wrt  $\labels{\lts}$.
 We denote by $\traces{\lts,q}$ the
set of all sequences that are traces of some execution fragment
$\alpha$ of $\lts$ that starts from $q\in\states{\lts}$.  An execution fragment
$\alpha$ of $\lts$ is an \emph{execution} of $\lts$ if it starts from $\initial{\lts}$.
A sequence $\pi$ of transition labels is
a \emph{trace} of $\lts$ if it is a trace of some execution of
$\lts$.  We denote by $\traces{\lts}$ the set of all traces of $\lts$
(so we have $\traces{\lts}=\traces{\lts,\initial{\lts}}$).
  
\item[Domains.]
  We assume sets $\Var$ of \emph{shared variables}
  and $\Val$ of \emph{values} with a distinguished \emph{initial} value $0\in \Val$. 
   We let $\Proc \defeq \set{\cproc_1 \til \cproc_N}$ be the set of process identifiers.

\item[Memory Actions.]
  Memory operations execute atomically
  using \emph{memory actions}, which include both argument and return values.
  Formally, a \emph{memory action} $\act \in \memacts$ is one the following
  (where $\var\in\Var$ and $v,\vold,\vnew\in \Val$):
  \begin{enumerate*}[label=(\roman*)]
  \item write action of the form $\Write(\var,v)$;
  \item read action of the form $\Read(\var,v)$;
  \item RMW action of the form $\RMW(\var,\vold,\vnew)$;
    and \item fence action of the form $\Fence$.
  \end{enumerate*}
  We denote by $\typf{\act}$ the type of the memory action $\act$ ($\Write$, $\Read$, $\RMW$, or $\Fence$)
  and by $\varf{\act}$ the variable accessed by action $\act$ (when applicable).

 \item[Memory Events.]
  A \emph{memory event}
  $\evt \in \MemEvs$ is a pair $\evt=\ev{\proc}{\act}$ where
  $\proc\in\Proc$ and $\act \in\memacts$.  We use $\procf{\evt}$ and
  $\actionf{\evt}$ to retrieve the components of $\evt$
  ($\proc$ and $\act$, respectively).
  The functions $\typf{\cdot}$ and $\varf{\cdot}$ are lifted to events in the obvious way.

 \item [Memory Models.]   
    The semantics of the memory
  operations is given by an LTS, called a \emph{memory model}.  The
  transition labels of a memory model $\mm$,
  $\labels{\mm} \defeq \MemEvs \cup \set{\tau}$, consist of memory
  events, as well as $\tau$, which represents a silent memory internal
  step.  
\end{description}

\noindent
  We demonstrate the formulation of TSO as an LTS. 
\iffull
Appendix \ref{sec:memory-models} formally presents \SCM and \RA. 
\fi

\begin{definition2}
\label{sec:tso}
$\TSO$'s states are pairs $\tup{m,b}$, where
$m \in\Var \to \Val$ is the main memory and
$b \in \Proc \to (\Var \times \Val)^*$ assigns a store buffer to every
process; the initial state is
$\init(\TSO) \defeq\tup{\lambda \var \ldotp 0, \lambda \proc \ldotp \emptyseq}$
(\ie all variables in memory are zeroed and all store buffers are empty);
and the transitions are as follows, where
$\restrict{\beta}{\var}$ denotes the restriction of a store buffer $\beta$ to pairs
of the form $\tup{\var,\_}$:
\begin{mathpar} 
  \inferrule[write]
  {\evt=\ev{\proc}{\Write(\var,v)} \\\\
  b' = b[\proc \mapsto b(\proc) \cdot \tup{\var, v}]}
  {\tup{m, b} \xrightarrow{\evt} \tup{m, b'}}
  \and
    \inferrule[read-from-buffer]
  {\evt=\ev{\proc}{\Read(\var,v)} \\\\ 
      \restrict{b(\proc)}{\var} = \_ \cdot \tup{\tup{\var,v}}}
  {\tup{m, b} \xrightarrow{\evt} \tup{m,b}}
  \and
    \inferrule[read-from-memory]
  {\evt=\ev{\proc}{\Read(\var,v)} \\\\ 
    \restrict{b(\proc)}{\var} = \emptyseq \\  m(\var)=v}
  {\tup{m, b} \xrightarrow{\evt} \tup{m,b}}
   \\
  \inferrule[rmw]{\evt=\ev{\proc}{\RMW(\var,\vold,\vnew)} \\\\ b(\proc) = \emptyseq \\ m(\var)=\vexp}
{\tup{m,b} \xrightarrow{\evt} \tup{m[\var\mapsto \vnew],b}}
\and  
  \inferrule[fence]
  {\evt = \ev{\proc}{\Fence}
    \\\\ b(\proc) = \emptyseq }
  {\tup{m, b} \xrightarrow{\evt} \tup{m, b}} 
    \and
  \inferrule[propagate]
  {b(\proc) = \tup{\tup{\var,v}} \cdot \beta \\\\
  m'= m[\var \mapsto v] \\ b' = b[\proc \mapsto \beta]}
  {\tup{m, b} \xrightarrow{\tau} \tup{m', b'}}   
\end{mathpar}
\end{definition2}

\begin{description}[leftmargin=0pt,itemsep=2pt]
\item [Memory Sequences.] We refer to sequences
  $\rho \in (\MemEvs \cup \set{\tau})^*$ as \emph{memory sequences}
  and to sequences $\sigma \in \MemEvs^*$ as \emph{observable memory
    sequences}.  We use the following notations:
\begin{itemize}
\item $\restrict{\sigma}{\proc}$ denotes the restriction of $\sigma$
\wrt $\set{\evt \in \MemEvs \st \procf{\evt}=\proc}$.
 \item $\otraces{\mm, q}$ denotes the set of all
  observable memory sequences obtained by restricting traces of $\mm$ from a state $q$
  to non-$\tau$ steps, \ie
  $\otraces{\mm, q} \defeq \restrict{\traces{\mm, q}}{\MemEvs}$. 
  \item $\otraces{\mm} \defeq \restrict{\traces{\mm}}{\MemEvs}$ is the
  set of all observable memory sequences of $\mm$.
\end{itemize}    

\item [Stable States.] 
A state $q \in \states{\mm}$ is \emph{stable} if
$q \not\xrightarrow{\tau}_\mm q'$ for any $q'\in \states{\mm}$.  
Every state of $\Atomic$ is stable, a state of $\TSO$ is
stable iff all store buffers are empty, and a state of $\RA$
is stable iff all processes are aware of all writes.

\item [Well-Behaved Memory Models.] 
$\TSO$ is strictly weaker
than $\Atomic$ and $\RA$ is strictly weaker than $\TSO$, which
formally means that
$\otraces{\Atomic} \subsetneq \otraces{\TSO} \subsetneq \otraces{\RA}$.  
In the sequel we will need the following assumption on memory models:

\begin{definition2}
\label{def:well-behaved}
A memory model $\mm$ is \emph{well-behaved} if there exists a simulation $R$
from $\SCM$ to $\mm$ whose codomain consists solely of stable states.
That is, there should exist a relation
$R \suq \states{\SCM} \times \set{q \in \states{\mm} \st q \text{ is stable}}$
such that $(i)$ $\tup{\initial{\SCM},\initial{\mm}} \in R$;
and $(ii)$ if $\tup{m,q}\in R$ and $m \xrightarrow{l}_\SCM m'$,
then $q \xrightarrow{l}_\mm \, \xrightarrow{\tau}_\mm^* q'$ and $\tup{m',q'}\in R$ for some stable $q' \in \states{\mm}$.
\end{definition2}

Note that if $\mm$ is well-behaved, then 
$\sigma_0 \cdot \sigma \in \otraces{\SCM}$ implies that 
there exist a stable state $q\in \states{\mm}$
and a memory trace $\rho_0$ such that 
$\initial{\mm} \xrightarrow{\rho_0}_\mm q$, $\restrict{\rho_0}{\MemEvs}=\sigma_0$, and $\sigma \in \otraces{\mm,q}$.
\iffull
The following lemma is proven in~\cref{sec:memory-models}.
\fi

\begin{restatable}{lemma}{wellbehaved}
\label{lem:at_traces}
Each $\mm\in\set{\Atomic, \TSO,\RA}$ is well-behaved.
\end{restatable}

\end{description}


%% file: merge.tex
\newcommand{\cmark}{{\color{green!85!black}{\checkmark}}}
\newcommand{\purple}[1]{\color{purple}#1}
\newcommand{\teal}[1]{\color{teal}#1}
\newcommand{\blue}[1]{\color{blue}#1}
\newcommand{\red}[1]{\color{brown}#1}

\section{Mergeability Results for Memory Models}
\label{sec:merge}

We consider two notions of mergeability of observable memory traces,
\emph{weak} mergeability, which means that \emph{some} interleaving of
the given traces is admitted, and \emph{strong} mergeability, which
requires that \emph{all} interleavings are admitted. 
We denote by $s_1 \shuffle s_2$ the 
the set of all interleavings of $s_1$ and $s_2$.

For our impossibility result to handle a non-empty base object history (as in \cref{ex:LOO-intro}),
it does not suffice to merge memory traces from the initial state.
Instead, we require the traces to be mergeable from every \emph{stable} state:

\begin{definition2}
\label{def:mergeable}
Two observable memory traces $\sigma_1, \sigma_2$ with
$\procsf{\sigma_1} \cap \procsf{\sigma_2} = \emptyset$ are
\emph{weakly (resp., strongly) mergeable} in a memory model $\mm$ if
for every stable state $q_0 \in \states{\mm}$ such that
$\sigma_1,\sigma_2 \in \otraces{\mm,q_0}$, we have
$\sigma \in \otraces{\mm,q_0}$ for some (resp., every)
$\sigma \in \sigma_1 \shuffle \sigma_2$.
\end{definition2}

\iffull
\begin{table}[t] 
  \centering 
  \begin{center}
    \scalebox{0.9}{
    \begin{tabular}[t]{l@{}l|l|l|l@{~~}l@{~~}l|l@{~~}l@{~~}l|l}
      & & & Memory  & \multicolumn{3}{c|}{Restrictions on $\sigma_1$} & \multicolumn{3}{c|}{Restrictions on $\sigma_2$} & Merge   \\
      \cline{5-10}
      \# & Name &Theorem & model   & process & events & pattern &  process & events & pattern &  property \\
      \hline
      1& $\TSO^\textsf{s}$ &\cref{thm:tso-merge2} &\TSO    & solo    & RW    & --- & solo & RW & --- &  Strong \\
      2& $\RA^\textsf{s}_1$&\cref{thm:ra-merge1}~\eqref{thm:RAM-Strong1} &\RA     & ---         & RW     & --- & --- & --- & --- &  Strong \\
      3& $\RA^\textsf{s}_2$ &\cref{thm:ra-merge3}~\eqref{thm:RAM-Strong2} &\RA      & ---     & RWF   & PPTF & --- & ---   & PPTF &  Strong \\
      4& $\RA^\textsf{s}_3$ &\cref{thm:ra-merge4}~\eqref{thm:RAM-Strong3} &\RA      & ---     & RWF   & PPLF & --- & --- & PPLF &  Strong \\
            \hline
      5& $\SCM^\textsf{w}$ &\cref{thm:atomic-merge} & \Atomic & ---    & RW  & RBW        & --- & --- & --- & Weak \\
      6& $\TSO^\textsf{w}$ &\cref{thm:tso-merge1} &\TSO    & solo      & RWF & LTF   & --- & --- & --- &  Weak \\
      7& $\RA^\textsf{w}$ &\cref{thm:ra-merge2} &\RA      & ---     & RWF   & LTF & --- & --- & --- &  Weak 
    \end{tabular}}
  \end{center}
  \caption{Merging observable memory sequences $\sigma_1$ and
    $\sigma_2$ such that
    $\procsf{\sigma_1} \cap \procsf{\sigma_2} = \emptyset$}
  \label{tab:results}
  \vspace{-20pt}
\end{table}
\fi

\Cref{tab:results} presents the merge properties established
for the memory models we consider  
\iffull
(see~\cref{app:proofs_merge} for the proofs).
\fi
To specify restrictions on the mergeable traces,
we say that an observable memory sequence $\sigma$ is:
\begin{description}[leftmargin=5pt]
\item [\emph{solo}] if $\size{\procsf{\sigma}}=1$; 
\item [\emph{read-write} (\emph{RW})] if $\typsf{\sigma}\suq\set{\Read,\Write}$;
\item [\emph{read-write-fence} (\emph{RWF})] if $\typsf{\sigma}\suq\set{\Read,\Write,\Fence}$; 
\item [\emph{read-before-write} (\emph{RBW})] if for every $k_1 < k_2$,
if $\typf{\sigma[k_1]}=\Write$, $\typf{\sigma[k_2]}=\Read$, and $\varf{\sigma[k_1]}\neq \varf{\sigma[k_2]}$,
then $\typf{\sigma[k]}=\Write$ and $\varf{\sigma[k]}= \varf{\sigma[k_2]}$ for some $k_1 < k < k_2$;\footnote{RBW is equivalent to the absence of the read-after-write (RAW) pattern as defined in~\cite{AGHKMV11}.}
\item [\emph{trailing-fence} (\emph{TF})] 
if there is no $k$ such that 
$\typf{\sigma[k]}=\Fence$ but $\typf{\sigma[k+1]}\neq\Fence$;
\item [\emph{leading-fence} (\emph{LF})] 
if there is no $k$ such that 
$\typf{\sigma[k]}=\Fence$ but $\typf{\sigma[k-1]}\neq\Fence$;
\item [\emph{per-process trailing fence} (\emph{PPTF})]
  if $\restrict{\sigma}{\proc}$ is TF for all processes $\proc$;
\item [\emph{per-process leading fence} (\emph{PPLF})]
  if $\restrict{\sigma}{\proc}$ is LF for all processes $\proc$; and 
\item [\emph{leading-and-trailing-fence} (\emph{LTF})] if 
  $\sigma = \sigma_1 \cdot \sigma_2$ for some LF $\sigma_1$ and TF $\sigma_2$.
\end{description}
We have three types of restrictions, namely:
\begin{enumerate*}[label=(\roman*)]
\item a restriction on the processes (solo);
\item restrictions on the types of events (RW and RWF); and 
\item restrictions on the access pattern (all others).
\end{enumerate*}
The restrictions on types and access patterns correspond to synchronization mechanisms that are expensive performance wise.
RMWs and non-RBW were identified as such in~\cite{AGHKMV11},
and since we explicitly deal with weak memory models, we add fences to this list.
To motivate our focus on leading/trailing fence placement,
we note that the trivial linearizable implementation of an atomic register using a write/read instruction
requires fences:
at the end of every write operation on \TSO,
and at the beginning and the end of every (write/read) operation on \RA.
We aim to investigate whether other objects admit similar implementations.

Next, we briefly discuss the results in the table:
\begin{description}[leftmargin=5pt]
\item [SCM.] In \SCM, if $\sigma_1$ is RW-RBW, then it can be weakly merged
with any other observable memory trace. 
Indeed, being RW-RBW, $\sigma_1$ must be of the form
  $\sigma_1^{\sf r}  \cdot \sigma_1^{\sf w}$ where $\sigma_1^{\sf r}$ is a
  sequence of reads and $ \sigma_1^{\sf w}$ is a sequence of writes and
  reads, starting with a write, where the reads in $\sigma_1^{\sf w}$
  read from the writes in $ \sigma_1^{\sf w}$.  
Then, it is
  straightforward to see that $\sigma_1$ and any observable memory sequence $\sigma_2$ can be merged
  to form the trace
  $\sigma = \sigma_1^{\sf r} \cdot \sigma_2 \cdot \sigma_1^{\sf w}$, which is valid trace under
  $\Atomic$. 
We note that the RBW restriction is necessary here, as 
$\tup{\purple{\ev{\cproc_1}{\Write(x,1)}}, \blue{\ev{\cproc_1}{\Read(y,0)}}}$
and
$\tup{\teal{\ev{\cproc_2}{\Write(y,1)}}, \red{\ev{\cproc_2}{\Read(x,0)}}}$
(which may arise from the SB example)
are not weakly mergeable.
Also note that there is no useful \emph{strong} merge property for \SCM.
Even $\tup{\purple{\ev{\cproc_1}{\Write(x,1)}}}$
and $\tup{\teal{\ev{\cproc_2}{\Read(x,0)}}}$ cannot be strongly merged.
\item [TSO.] 
In \TSO, $\sigma_1$ and $\sigma_2$ can be \emph{strongly} merged when they
are both solo-RW traces. This holds
because with only writes and reads, there is always an observable trace
where {\em all} the writes of both $\sigma_1$ and $\sigma_2$ remain in
the local store buffers, allowing the events of $\sigma_1$ and
$\sigma_2$ to be arbitrarily interleaved. 
\TSO also satisfies a weak merge property if $\sigma_1$ is solo-RWF-LTF
and $\sigma_2$ is arbitrary. 
To do so, we let $\sigma_1 = \sigma_1^{\sf lf} \cdot \sigma_1' \cdot \sigma_1^{\sf tf}$ 
where $\typsf{\sigma_1^{\sf lf}} \cup \typsf{\sigma_1^{\sf tf}} \subseteq \{\Fence\}$ and
$\sigma_1'$ is RW. Then,
$\sigma_1^{\sf lf} \cdot \sigma_1' \cdot \sigma_2 \cdot \sigma_1^{\sf tf}$ 
is a valid \TSO observable trace since no instruction in
$\sigma_1'$ forces writes to propagate.
We note that the solo restriction is essential. 
For example, 
$\tup{\purple{\ev{\cproc_1}{\Write(x,1)}}, \blue{\ev{\cproc_2}{\Read(x,1)}},\blue{\ev{\cproc_2}{\Read(y,0)}}}$
and
$\tup{\teal{\ev{\cproc_4}{\Write(y,1)}}, \red{\ev{\cproc_3}{\Read(y,1)}},\red{\ev{\cproc_3}{\Read(x,0)}}}$
(which may arise from the IRIW example)
are not weakly mergeable.
\item [RA.] 
We prove three strong merge properties for \RA:
\begin{enumerate*}
\item[($\RA^\textsf{s}_1$)]
  If $\sigma_1$ is RW, then it can be strongly merged with
  $\sigma_2$ even when $\sigma_1$ is non-solo. Indeed, in the
  absence of RMWs and fences in $\sigma_1$, the writes in $\sigma_1$ can be
  propagated to other processes of $\sigma_1$, but never propagate to
  the processes of $\sigma_2$, and vice-versa. 
\item[($\RA^\textsf{s}_2$)] If $\sigma_1$ is RWF-PPTF
  and $\sigma_2$ is PPTF, the strong merge argument is as
  follows. First, we remove all the fences in $\sigma_1$, which results
  in an RW trace. From $\RA^\textsf{s}_1$, this trace can be strongly merged
  with $\sigma_2$. In the resulting trace, we reintroduce the fences
  removed from $\sigma_1$ arbitrarily after the last read or write of
  the corresponding process. Regardless of whether this fence is before
  or after a fence of $\sigma_2$, the resulting fence synchronization
  has no effect since $\sigma_2$ is also PPTF. 
\item[($\RA^\textsf{s}_3$)] If $\sigma_1$ is RWF-PPLF and $\sigma_2$ is PPLF the argument is
  symmetric to $\RA^\textsf{s}_2$.
\end{enumerate*}
Finally, \RA satisfies a weak merge property if $\sigma_1$ is
RWF-LTF. As in the \TSO weak merge property, we split
$\sigma_1 = \sigma_1^{\sf lf} \cdot \sigma_1' \cdot \sigma_1^{\sf tf}$.
By $\RA^\textsf{s}_1$, $\sigma_1' \cdot \sigma_2$ is
an \RA observable trace. Then,
$\sigma_1^{\sf lf} \cdot \sigma_1' \cdot \sigma_2 \cdot \sigma_1^{\sf
  tf}$ is an \RA observable trace since the leading/trailing fences have no
bearing on the execution.
\end{description}


%% file: impos-obj.tex
\section{A Recipe for Merge-Based Impossibility Results}
\label{sec:impos-obj}

We introduce objects, implementations, and histories (\cref{sec:impl}),
and our main theorem (\cref{sec:theorem}).

\subsection{Objects and Their Implementations}
\label{sec:impl}

We consider systems implementing of a high-level object $\obj$
using the low-level atomic shared-memory operations provided by the memory model $\mm$.

\begin{description}[leftmargin=0pt,itemsep=2pt,listparindent=\parindent]

\item [Objects.]
An \emph{object} $\obj$ is a pair $\obj=\tup{\ops,\rets}$,
where $\ops$ is a set of \emph{operation names} (each of which may include argument values)
and $\rets$ is a set of \emph{response values}.
We use $\opsf{\obj}$ and $\retsf{\obj}$ to retrieve the components of an object $\obj$ 
($\ops$ and $\rets$, respectively).
We use $\ack$ for a default response value for operations that do not return any value.
\item [Object Actions.]
To delimit executions of operations of $\obj$, we use 
\emph{object actions} that can be either \emph{invocation actions} of the form $\invi{\op}$ with $\op\in \opsf{\obj}$,
or \emph{response actions} of the form $\resi{u}$ with $u\in \retsf{\obj}$.  
We let $\actsf{\obj}$ denote the set of all object actions of $\obj$. 
\item [Object Events.]
Like memory events defined in \cref{sec:memory}, 
\emph{object events} are pairs $\evt=\ev{\proc}{\act}$ 
where $\proc\in\Proc$ and $\act\in\actsf{\obj}$.
We apply the same notations used for memory events to object events,
and let $\ObjEvs(\obj)$ denote the set of all object events.
By \emph{event} we collectively refer
to either a memory event or an object event. 
Given a sequence $\pi$ of events, we define the following notations:
\begin{itemize}
\item $\restrict{\pi}{\proc}$ denotes the restriction of $\pi$ \wrt the set
of events $\evt$ with $\procf{\evt}=\proc$. 
\item $\restrictmem{\pi}$ denotes the restriction of $\pi$ \wrt the set $\MemEvs$ of memory events.
\item $\restrictobj{\pi}$ denotes the restriction of $\pi$ \wrt the set $\ObjEvs(\obj)$ of object events.
\end{itemize}
\item [Histories.]
A \emph{history} of an object $\obj$ is a sequence of events in $\ObjEvs(\obj)$.
We denote by $\invres{\proc}{\op}{u}$ the history consisting of a single operation by 
process $\proc\in\Proc$ invoking $\op\in \opsf{\obj}$ with response value $u\in \retsf{\obj}$
(and omit the response value if it is $\ack$), \ie
$\invres{\proc}{\op}{u} \defeq \tup{\ev{\proc}{\invi{\op}},\ev{\proc}{\resi{u}}}$
and $\invres{\proc}{\op}{} \defeq \tup{\ev{\proc}{\invi{\op}},\ev{\proc}{\resi{\ack}}}$.
A history $\h$ is:
\begin{itemize}
\item  \emph{sequential} if it is a prefix of
a history of the form $\invres{\proc_1}{\op_1}{u_1} \cdot \invres{\proc_2}{\op_2}{u_2} \cdots \invres{\proc_n}{\op_n}{u_n}$;
\item \emph{well-formed} if $\restrict{\h}{\proc}$ is sequential for every $\proc\in \Proc$; and 
\item \emph{complete} if it is well-formed and each $\restrict{\h}{\proc}$ ends with a response event.
\end{itemize}
We let $\histories{\obj}$, $\chistories{\obj}$, and $\cshistories{\obj}$
denote the sets of all well-formed histories of $\obj$,
all complete histories of $\obj$,
and all complete sequential histories of $\obj$ (respectively).

\item [Specifications.]
We assume that every object $\obj$ is associated with a \emph{specification}, denoted $\spec{\obj}$, 
that is a subset of $\cshistories{\obj}$
that is prefix-closed (in the sense that $\h' \in \spec{\obj}$
for every $\h' \in \cshistories{\obj}$ that is a prefix of some $\h \in \spec{\obj}$).
An object $\obj$ is \emph{deterministic} if no 
two histories in $\spec{\obj}$
have longest common prefix that ends with an invocation.
\item [Implementations.]
An \emph{implementation $\implo$ of an operation $\op$ for a process $\proc$} is an
LTS whose set of transition labels are events with process identifier $\proc$.
We assume that a response event is always the last transition of executions of $\implo$
(\ie if $q \xrightarrow{\ev{\proc}{\resi{u}}}_{\implo} q'$, then no transition is enabled in $q'$).
An \emph{implementation $\impl$ of an object $\obj$} is a function
assigning an implementation $\impl(\op,\proc)$ of $\op$ for $\proc$ to every $\op\in \opsf{\obj}$ and $\proc\in\Proc$.

\begin{figure}[t]
  \begin{minipage}[b]{0.6\columnwidth}
\begin{mathpar}
    \centering
    \inferrule
    {\evt=\ev{\proc}{\invi{\op}} \\\\ q = \initial{\impl(\op,\proc)}}
    {\bot \xrightarrow{\evt}_\sysp{\impl}{\proc} \tup{\op,q}}
    \and
    \inferrule
    {\evt\in\MemEvs \\\\ q \xrightarrow{\evt}_{\impl(\op,\proc)} q'}
    {\tup{\op,q} \xrightarrow{\evt}_\sysp{\impl}{\proc} \tup{\op,q'}}
    \and
    \inferrule
    {\evt=\ev{\proc}{\resi{u}} \\\\ q \xrightarrow{\evt}_{\impl(\op,\proc)} \_}
    {\tup{\op,q} \xrightarrow{\evt}_\sysp{\impl}{\proc} \bot}
  \end{mathpar}
  \vspace{-15pt}
  \caption{Transitions of $\sysp{\impl}{\proc}$}
  \label{fig:lts-SIP}
\end{minipage}
\hfill 
  \begin{minipage}[b]{0.35\columnwidth}
\begin{mathpar}
    \centering
    \inferrule
    {\bar{q}(\proc) \xrightarrow{\evt}_{\sysp{\impl}{\proc}} q'}
    {\bar{q} \xrightarrow{\evt}_\sys{\impl} \bar{q}[\proc\mapsto q']}
  \end{mathpar}
  \vspace{-10pt}
  \caption{Transitions of $\sys{\impl}$}
  \label{fig:lts-SI}
  \end{minipage}
  \vspace{-10pt}
\end{figure}

An implementation $\impl$ of an object $\obj$ induces an LTS, denoted
$\sys{\impl}$, that repeatedly and concurrently executes the
operations of $\obj$ as $\impl$ prescribes.
To 
formally define
$\sys{\impl}$, we first define the ``per-process'' LTS induced by
$\impl$, denoted $\sysp{\impl}{\proc}$.  This LTS is given by:
$\states{\sysp{\impl}{\proc}} \defeq \set{\bot} \cup \set{\tup{\op,q}
  \st \op\in\opsf{\obj}, q\in\states{\impl(\op,\proc)}}$;
$\initial{\sysp{\impl}{\proc}} \defeq \bot$;
$\labels{\sysp{\impl}{\proc}}\defeq \ObjEvs(\obj) \cup \MemEvs$; and
the transitions are given in \cref{fig:lts-SIP}. The state $\bot$
means that the process is not currently executing any operation,
whereas $\tup{\op,q}$ means that process $\proc$ is currently executing
$\op$ and it is in state $q$ of the implementation of $\op$ for $\proc$.

In turn, 
$\sys{\impl}$ is given by: $\states{\sys{\impl}}$ is the set of all
mappings assigning a state in $\states{\sysp{\impl}{\proc}}$ to every
$\proc\in\Proc$;
$\initial{\sys{\impl}} \defeq \lambda \proc \ldotp \bot$;
$\labels{\sys{\impl}}\defeq \ObjEvs(\obj) \cup \MemEvs$; and the
transition relation in \cref{fig:lts-SI}.
This transition simply interleaves the transitions of the different processes.
In the sequel, we let $\traces{\impl} \defeq \traces{\sys{\impl}}$.

\item [Histories of Implementations.]
Let $\impl$ be an implementation of an object $\obj$, 
$\pi_0$ be a sequence of events,
and $\mm$ be a memory model.
A history $\h$ of $\obj$ is:
\begin{itemize}
\item \emph{generated by $\impl$ after $\pi_0$}
if $\h=\restrictobj{\pi}$ for some $\pi$ 
such that $\pi_0\cdot \pi \in\traces{\impl}$.
\item \emph{generated by $\impl$ after $\pi_0$ under $\mm$}
if $\h=\restrictobj{\pi}$ for some $\pi$ 
such that $\pi_0\cdot \pi \in\traces{\impl}$ and $\restrictmem{(\pi_0\cdot \pi)} \in \otraces{\mm}$.
\end{itemize}
We denote by $\histories{\pi_0,\impl}$ the set of all histories 
that are generated by $\impl$ after $\pi_0$,
and by $\histories{\pi_0,\impl,\mm}$ the set of all histories 
generated by $\impl$ after $\pi_0$ under $\mm$.
We also write $\histories{\impl}$ instead of $\histories{\emptyseq,\impl}$
and $\histories{\impl,\mm}$ instead of $\histories{\emptyseq,\impl,\mm}$.
\end{description}

\subsection{The Merge Theorem}
\label{sec:theorem}

Our main result relates mergeability properties of memory models and
objects implemented in those models, assuming that the implementation
provides minimal safety and liveness guarantees.
This result can be also seen as a \emph{CAP Theorem} for weak memory models~\cite{GilbertL02}, 
where partition tolerance of CAP corresponds to mergeability, as it allows two traces of distinct set of processes to run concurrently without interaction.
Our results are more fine grained, as we show the correspondence between mergeability
of certain traces in a memory model, and the (in)ability of these traces to implement non-mergeable object histories.

For the formal treatment, we first present the following lemma 
\iffull
(proven in~\cref{app:impos-obj}).
\fi
The lemma describes
the key shape of our results, namely that given two traces of an
implementation over a memory model, the merge property over these
traces carries over to a merge
property over the histories induced by the traces.
\begin{restatable}{lemma}{mainzero}
\label{lem:main0}
Let $\impl$ be an implementation of $\obj$.
Suppose that there exist sequences $\pi_0, \pi_1, \pi_2$ of events such that the following hold:
\begin{enumerate}[leftmargin=*,label=(\alph*)]
\item $\procsf{\pi_1} \cap \procsf{\pi_2} = \emptyset$;  $\pi_0 \cdot \pi_1, \pi_0 \cdot \pi_2 \in \traces{\impl}$;  $\restrictobj{\pi_0} \in \chistories{\obj}$;
and \item $\restrictmem{\pi_0} \cdot \sigma \in \otraces{\mm}$ 
for some (resp., every) $\sigma\in \restrictmem{\pi_1} \shuffle \restrictmem{\pi_2}$.
\end{enumerate}
Then, $\h \in \histories{\pi_0,\impl,\mm}$ 
for some (resp., every) $\h \in \restrictobj{\pi_1} \shuffle \restrictobj{\pi_2}$.
\end{restatable}

The Merge Theorem, which we obtain using this lemma,
makes several assumptions on implementations.
First, the safety condition, which we call \emph{consistency}, 
is restriction of linearizability to complete histories. For its definition,
we first define reorderings of sequences.

\begin{definition2}
Let $R \subseteq X \times X$.
A sequence $s' \in X^*$ is an \emph{$R$-reordering}
of a sequence $s \in X^*$ if there exists a bijection $f : \set{1 \til \size{s}} \to \set{1 \til \size{s'}}$
such that $s[i]=s'[f(i)]$ for every $1\leq i \leq \size{s}$,
and $f(i) < f(j)$ whenever $i < j$ and $\tup{s[i],s[j]}\in R$.
We denote by $\reorder{R}{s}$ the set of all $R$-reorderings of $s$,
and lift this notation to sets by letting
$\reorder{R}{S} \defeq \bigcup_{s\in S} \reorder{R}{s}$.
\end{definition2}

We define $\sproc$ and $\lin$ relations on events:
  \begin{align*}
    \sproc & \defeq \set{\tup{e_1,e_2} \st \procf{e_1} = \procf{e_2}}
    \\ \lin & \defeq \sproc \cup (\set{e \st \text{$e$ is a response event}} \times \set{e \st \text{$e$ is a invocation event}})
  \end{align*}

  \begin{definition2}
    \label{def:lin}
    A history $\h' \in \histories{\obj}$ \emph{linearizes}
    a history $\h\in \histories{\obj}$, denoted $\h \sqsubseteq \h'$,
    if $\h' \in \reorder{\lin}{\h}$.
    For a set $H' \subseteq \histories{\obj}$, 
    we write $\h \sqsubseteq H'$ if $\h \sqsubseteq \h'$ for some $\h'\in H'$.
  \end{definition2}

  \begin{definition2}
    \label{def:consistent}
    An implementation $\impl$ of an object $\obj$ is \emph{consistent} under a memory model $\mm$ if 
    $\h \sqsubseteq \spec{\obj}$ for every complete history $\h \in \histories{\impl,\mm}$.
  \end{definition2}

Consistency follows from linearizability~\cite{HW90}, and it is equivalent to linearizability
for implementations in which every history can be extended to a complete history.

Next, the liveness condition, which we call \emph{availability}, 
requires progress for the specific histories under consideration.

\begin{definition2}
  \label{def:available}
  An implementation $\impl$ of $\obj$ is \emph{available after a history $\h_0\in \cshistories{\obj}$}  \wrt a history $\h \in \histories{\obj}$
  if $\h \in \histories{\pi_0,\impl,\Atomic}$ for every $\pi_0 \in \traces{\impl}$ such that 
  $\restrictmem{\pi_0}\in\traces{\Atomic}$ and $\restrictobj{\pi_0}=\h_0$. 
  We say $\impl$ is \emph{available \wrt $\h$}, if it is available after $\emptyseq$ \wrt $\h$ (\ie $\h \in \histories{\impl,\Atomic}$).
  We call $\impl$ \emph{spec-available} if for every $\h_0,\h\in \cshistories{\obj}$ such that $\h_0 \cdot \h \in \spec{\obj}$,
  $\impl$ is available after $h_0$ \wrt $\h$.
\end{definition2}

Availability \wrt $\h$ after $\h_0$ only guarantees that the
implementation under \SCM is able to generate the history $\h$ when it
starts executing after generating $\h_0$. 
For deterministic implementations, availability \wrt $\h$ after $\h_0$
follows from availability \wrt $\h_0 \cdot \h$ (after $\epsilon$).
Note that availability 
considers \SCM rather than a general memory model $\mm$,
but when $\mm$ is well-behaved (\cref{def:well-behaved}),
$\h \in \histories{\pi_0,\impl,\Atomic}$ ensures that
$\h \in \histories{\pi_0,\impl,\mm}$.  
Spec-availability essentially means that the implementation can generate
all (sequential) specification histories and for
deterministic objects and implementations, it follows from obstruction-freedom~\cite{obstruction-free}.

The next lemma 
\iffull
(proven in~\cref{app:impos-obj})
\fi
is used in the sequel to derive availability \wrt a history $\h$ from the fact
that availability holds \wrt a sequential history that linearizes $\h$.

  \begin{restatable}{lemma}{availablelin}
    \label{lem:available_lin}
    Suppose that $\impl$ is available after $\h_0$ \wrt a history $\h' \in \histories{\obj}$.
    Then, $\impl$ is available after $\h_0$ \wrt every $\h\in \histories{\obj}$ such that $\h \sqsubseteq \h'$.
  \end{restatable}

Next, we define mergeability for objects, akin to mergeability for memory models (\cref{def:mergeable}):

\begin{definition2}
\label{def:mergeable_obj}
Two histories $\h_1, \h_2\in\chistories{\obj}$ with $\procsf{\h_1} \cap \procsf{\h_2} = \emptyset$
are \emph{weakly (resp., strongly) mergeable in $\spec{\obj}$ after a history 
$\h_0\in \cshistories{\obj}$}
if $\h_0 \cdot{} \h_1 \sqsubseteq \spec{\obj}$ and $\h_0 \cdot{} \h_2 \sqsubseteq \spec{\obj}$ 
imply that $\h_0 \cdot{} \h \sqsubseteq \spec{\obj}$ for some (resp., every) $\h \in \h_1 \shuffle \h_2$. 
\end{definition2}

We now have all prerequisites to state our Merge Theorem 
\iffull
(see~\cref{app:impos-obj} for the proof). 
\fi
\begin{restatable}{theorem}{main}
\label{thm:main}
Let $\impl$ be an implementation of an object $\obj$ that is consistent under a 
well-behaved memory model $\mm$.
Suppose that there exist $\pi_0 \in \traces{\impl}$, 
$\h_1,\h_2 \in \chistories{\obj}$ such that
the following hold, where $\h_0=\restrictobj{\pi_0} $ and $\sigma_0=\restrictmem{\pi_0}$:
\begin{enumerate}[leftmargin=0pt,align=left, label=(\roman*)]
\item $\h_0 \in \spec{\obj}$, $\sigma_0 \in\traces{\Atomic}$, and $\procsf{\h_1} \cap \procsf{\h_2} = \emptyset$,\label{item:1} 
\item $\impl$ is available after $\h_0$ \wrt some $\seqh^i \in \cshistories{\obj}$
  such that $\h_i \sqsubseteq \seqh^i$ for $i\in\set{1,2}$, \label{item:2} 
\item $\h_1$ and $\h_2$ are not weakly (resp., strongly) mergeable in $\spec{\obj}$ after $\h_0$.\label{item:3}
\end{enumerate}
Then, there exist $\pi_1$ and $\pi_2$ such that all of the following
hold:
\begin{enumerate}[leftmargin=0pt,align=left, label=(\alph*)]

\item For $i\in\set{1,2}$, we have $\pi_0 \cdot \pi_i \in \traces{\impl}$; \; $\restrictobj{\pi_i}=\h_i$;  \; $\sigma_0 \cdot \restrictmem{\pi_i} \in \traces{\Atomic}$; \; and $\procsf{\pi_i} = \procsf{\h_i}$. \label{item:main1}\label{item:main2}\label{item:main3} \label{item:main4}

\item For every $\pi_1' \in \reorder{\sproc}{\pi_1}$ and $\pi_2' \in \reorder{\sproc}{\pi_2}$
such that $\restrictobj{\pi_1'}=\h_1$, $\restrictobj{\pi_2'}=\h_2$, 
and $\sigma_0 \cdot \restrictmem{\pi_1'}, \sigma_0 \cdot \restrictmem{\pi_2'}\in \traces{\Atomic}$,
we have that $\restrictmem{\pi_1'}$ and $\restrictmem{\pi_2'}$ are not weakly (resp., strongly) mergeable in $\mm$.
In particular, $\restrictmem{\pi_1}$ and $\restrictmem{\pi_2}$ are not weakly (resp., strongly) mergeable in $\mm$.\label{item:main5}
\end{enumerate}
\end{restatable}

For simplicity, we explain \cref{thm:main} for $\pi_0 = \emptyseq$
(and hence $\h_0 = \sigma_0 = \emptyseq$).  The theorem assumes that
we start with an implementation $\impl$ that is consistent
under the memory model $M$ under consideration. Moreover, we assume
that we have two complete histories $\h_1$ and $\h_2$ of the object
such that the processes of $\h_1$ and $\h_2$ are distinct (condition
\ref{item:1}), $\impl$ is available \wrt some linearization of $\h_1$
and $\h_2$ (condition \ref{item:2}), and that $\h_1$ and $\h_2$ are
\emph{\textbf{not}} weakly (strongly) mergeable (condition
\ref{item:3}). Then, for $i\in \{1,2\}$ there must be a trace $\pi_i$
of $\impl$, 
corresponding to $\h_i$,
whose memory events are allowed by
$\SCM$,  
and processes are only included in
$\pi_i$ if they call some operation of the object (condition
\ref{item:main4}), such that $\pi_1$ and $\pi_2$ restricted to memory
events are \emph{\textbf{not}} weakly (strongly) mergeable in $M$
(second clause of condition \ref{item:main5}). In fact, weak (strong)
\emph{\textbf{non}}-mergeability extends to any process-preserving
reordering of $\pi_1$ and $\pi_2$ whose corresponding histories are
$\h_1$ and $\h_2$ and corresponding memory traces are \SCM traces (first clause of condition \ref{item:main5}). 


%% file: apps-summary.tex
\newcommand{\sumtable}{\begin{table}[t] 
  \centering \footnotesize
  \begin{center}\footnotesize
    \begin{tabular}[t]{l|l|l|l|l|l|l}
      &\makecell{History\\ Merge}  & Bound & Termination & $\Atomic$ & $\TSO$ & $\RA$\\
      \hline
      \makecell{One-sided \\non-comm (\cref{def-weakly-non-commutative})\\($\rego$, $\mrego$, $\mcounto$)} & \makecell{Weak,\\ but non-\\ strong} & Lower&\makecell{spec-\\available}&NA&\cref{thm:weak-non-comm} &\cref{thm:weak-non-comm}\\
      \cline{3-7}
      &&Upper & wait-free & well-known &  \textbf{\cref{thm:strong-impl}} & \textbf{\cref{thm:strong-impl}} \\
      \hline
	\makecell{Two-sided \\non-comm\ (\cref{def-strongly-non-commutative})\\($\conso$, set)} & \makecell{Non-\\weak} & Lower&\makecell{spec-\\available}&\cref{thm:two-side-lower}&\cref{thm:two-side-lower} &\cref{thm:two-side-lower}\\
      \cline{3-7}
      &&Upper & \makecell{obstruction-\\free} & \textbf{\cref{thm:two-sided-upper}}$^*$ & \textbf{\cref{thm:two-sided-upper}}$^*$ & FE \\
      \hline

    \makecell{Snapshot, counter}& \makecell{Non-\\weak} & Lower&\makecell{spec-\\available}&
    \cref{thm:snap-weak-tso-scm-ra},~\ref{thm:count-weak-tso-scm-ra}&
    \cref{thm:snap-weak-tso-scm-ra},~\ref{thm:count-weak-tso-scm-ra} & 
    \cref{thm:snap-weak-tso-scm-ra},~\ref{thm:count-weak-tso-scm-ra}\\
      \cline{3-7}
      &&Upper & \makecell{wait-free} & \textbf{\cref{thm:snap-upper},~\ref{thm:counter-upper}} & 
      \textbf{\cref{thm:snap-upper},~\ref{thm:counter-upper}}& FE \\
      \hline

     \makecell{Mutual Exclusion}& \makecell{Non-\\weak} & Lower&\makecell{solo-termination}&
    Cor~\ref{cor:mutex}&
    Cor~\ref{cor:mutex} & 
    Cor~\ref{cor:mutex}\\
      \cline{3-7}
      &&Upper & \makecell{starvation-\\free} & \textbf{\cref{sec:mutex}} & 
      \textbf{\cref{sec:mutex}} & FE \\
      \hline

    \end{tabular}
  \end{center}
  \caption{Implementability Results. Lower bounds are for consistent implementations and upper bounds
  are for linearizable ones.
  Tight upper bounds are \textbf{bold-faced}. The bounds which are tight in the absence of contention are marked with ($^*$).
  FE stands for ``fence-everywhere'' -- the strategy that obtains an equivalent $\RA$ implementation from
  an $\Atomic$ algorithm by inserting a fence after every non-RMW memory operation.}
  \label{tab:impl-results}
  \vspace{-20pt}
\end{table}
}

\section{Implementability of Objects on Weak Memory Models}
\label{sec-first-applications}

We demonstrate the power of the \mergethm by using it along with the
mergeability results in \cref{tab:results} to characterize
implementability of objects under weak memory models.

\subsection{One-Sided Non-Commutative Operations}
\label{sec:app_weak}

We start by analyzing implementability of pair of operations
$\op_1$ and $\op_2$ such that $\op_1$ is \emph{one-sided non-commutative} \wrt $\op_2$. 
Roughly, this means that the execution order of $\op_1$
and $\op_2$ affects the response of $\op_1$. 
Formally:

\begin{definition2}
\label{def-weakly-non-commutative}
An operation $\op_1  \in \opsf{\obj}$ is 
\emph{one-sided non-commutative \wrt an operation $\op_2 \in \opsf{\obj}$ in $\spec{\obj}$} 
if there exist $\h_0 \in \cshistories{\obj}$, processes $\proc_1 \neq \proc_2$, 
and response values $u_1,v_1,u_2\in\retsf{\obj}$ such that:
\begin{enumerate*}[label=(\roman*)]
\item $u_1 \neq v_1$;
\item $\h_0 \cdot \invres{\proc_1}{\op_1}{u_1} \in \spec{\obj}$; and 
\item $\h_0 \cdot \invres{\proc_2}{\op_2}{u_2} \cdot \invres{\proc_1}{\op_1}{v_1} \in \spec{\obj}$.
\end{enumerate*}
\end{definition2}

\begin{example}
Consider a standard register object $\rego$ with initial value $0$, and 
operations $\writeo{v}$, where $v\in V$ for some set of values $V$, and $\reado$.
Then, $\reado$ is one-sided non-commutative \wrt $\writeoo$ in $\spec{\rego}$. 
Indeed, for $\proc_1\neq\proc_2$ and $\h_0=\emptyseq$, we have both
$\invress{\proc_1}{\reado}{0}{30pt} \in \spec{\rego}$ and
$\invress{\proc_2}{\writeo{1}}{}{40pt} \cdot \invress{\proc_1}{\reado}{1}{30pt}\in \spec{\rego}$.
The same holds for \emph{max-register}~\cite{AAC12}, denoted $\mrego$,
that stores integers with the initial value $0$. 
We note that all pairs of specification histories of $\rego$ and $\mrego$ with disjoint sets of processes
are weakly mergeable.
\end{example}

\begin{example}
Consider a monotone counter object $\mcounto$ with initial value $0$, and 
operations $\inco$ and $\reado$.
Then, $\reado$ is one-sided non-commutative \wrt $\inco$ in $\spec{\mcounto}$ 
as for $\proc_1\neq\proc_2$ and $\h_0=\emptyseq$, we have 
$\invress{\proc_1}{\reado}{0}{30pt} \in \spec{\mcounto}$ and
$\invress{\proc_2}{\inco}{}{20pt} \cdot \invress{\proc_1}{\reado}{1}{30pt}\in \spec{\mcounto}$.
\end{example}

The next lemma 
\iffull
(proven in~\cref{app:app_weak})
\fi
shows that for deterministic objects, the existence 
of a pair of operations one of which is one-sided non-commutative \wrt to the other 
implies that their corresponding histories are  not strongly mergeable:

\begin{restatable}{lemma}{lemmaweaknoncommandmerge}
\label{lemma-weak-non-comm-and-merge}
Let $\obj$ be a deterministic object and suppose that 
$\op_1 \in \opsf{\obj}$ is one-sided non-commutative \wrt $\op_2\in \opsf{\obj}$ in $\spec{\obj}$. 
Then, there exist  $\h_0 \in \cshistories{\obj}$, processes $\proc_1\neq \proc_2$,  
and response values $u_1,u_2\in\retsf{\obj}$ 
such that $\invres{\proc_1}{\op_1}{u_1}$ and $\invres{\proc_2}{\op_2}{u_2}$ 
are not strongly mergeable in $\spec{\obj}$ after $\h_0$. 
\end{restatable}

Then, the following theorem
\iffull
(proven in~\cref{app:app_weak})
\fi
follows from \cref{thm:main} 
and properties $\TSO^\textsf{s}$ and $\RA^\textsf{s}_1$ in \cref{tab:results}. 

\begin{restatable}{theorem}{thmweaknoncomm}
\label{thm:weak-non-comm}
Let $\obj$ be a deterministic object and suppose that $\op_1 \in \opsf{\obj}$ 
is one-sided non-commutative \wrt $\op_2 \in \opsf{\obj}$ in $\spec{\obj}$. 
Let $\impl$ be a spec-available implementation of $\obj$ that is consistent under 
$\mm\in \set{\TSO, \RA}$.
Then, there exist $\proc_1, \proc_2\in\Proc$, $\pi_1 \in \traces{\impl(\op_1,\proc_1)}$, 
and $\pi_2\in \traces{\impl(\op_2,\proc_2)}$ such that the following hold 
for $\sigma_1=\restrictmem{\pi_1}$ and $\sigma_2=\restrictmem{\pi_2}$:
\begin{enumerate}[leftmargin=*,align=left, label=(\alph*)]
\item
if $\mm = \TSO$, then either $\sigma_1$ or $\sigma_2$ has a fence or a RMW event; and
\item
if $\mm=\RA$, then neither $\sigma_1$ nor $\sigma_2$ is RW, and 
one of the following holds:
\begin{enumerate*}[label=(\roman*)]
\item either $\sigma_1$ or $\sigma_2$ has a RMW event;
\item either $\sigma_1$ or $\sigma_2$ is not LTF (\ie has a fence in the middle); 
\item $\sigma_1$ is LF and $\sigma_2$ is TF; or
\item $\sigma_1$ is TF and $\sigma_2$ is LF.
\end{enumerate*}
\end{enumerate}
\end{restatable}

Since $\reado$ is one-sided non-commutative \wrt $\writeoo$ in both $\spec{\rego}$ and $\spec{\mrego}$, 
their respective implementations under $\TSO$ and $\RA$ are subject to the constraints
given in \cref{thm:weak-non-comm}.
The same holds for the implementations of 
the $\reado$ and $\inco$ operations of $\mcounto$. 

To establish the tightness of these lower bounds,
we present linearizable wait-free implementations of
$\rego$ and $\mrego$ that are optimal \wrt the above bounds:
for \TSO, it uses only reads, writes, and a single fence at the end of\hspace{2pt} $\writeoo$;
and for \RA, it uses only reads, writes, and a pair of fences at both the 
beginning and the end of\hspace{2pt} both $\writeoo$ and $\reado$.

A $\rego$ object is trivial to implement under $\Atomic$ and 
there are $\mrego$ implementations under $\Atomic$~\cite{AAC12}
with every operation being RBW.
We use these implementations as a basis for implementations under $\TSO$ and $\RA$ as follows:

\begin{description}[leftmargin=*,topsep=5pt,itemsep=2pt]
\item[TSO.] For $\TSO$, we utilize a \emph{fence-insertion strategy},
  which derives a linearizable $\TSO$ implementation of an object from its
  $\Atomic$ counterpart by inserting a fence in-between every consecutive 
  pair of write and read, as well as between a final write of an operation (if it exists) and the operation's response.
  We give full details, prove correctness, and present more examples of applications of this transformation
  \iffull
  in~\cref{sec-simulation}.
  \fi
  Using this strategy, we obtain
  a $\TSO$ implementation of $\rego$ as follows: $\writeoo$ first writes 
  to a memory location, and then executes a fence, and $\reado$
  reads the same memory location and returns the value read. 
  Likewise, to implement $\mrego$ under $\TSO$, we add 
  a fence at the end of the $\writeoo$ implementations of~\cite{AAC12}, and
  leave their $\reado$ implementation as is. 
\item[RA.] 

  We augment the $\TSO$ implementations above by adding another fence
  at the beginning of $\writeoo$ as well as fences at the beginning and the
  end of $\reado$. 
  \iffull
  The pseudocode of the $\mrego$ algorithm and its correctness proof 
  can be found in~\cref{app:max_reg}. 
  \fi
  Further details of the register implementation 
  and its correctness proof appear 
  \iffull
  in~\cref{app:reg_under_RA}.
  \fi
  For conciseness, our $\mrego$ implementation under
  $\RA$ is derived from a simplified version of the algorithm 
  in~\cite{AAC12} (with linear step complexity instead of logarithmic as in~\cite{AAC12}).
\end{description}

\subsection{Two-Sided Non-Commutative Operations and Mutual Exclusion}
\label{sec:app_two}

We next explore implementability of objects with non-weakly mergeable histories. 
We apply our framework to generalize the ``laws of order'' (LOO) results of~\cite{AGHKMV11}.
The next notion of two-sided non-commutativity is a strengthening
of one-sided non-commutativity defined above, and is identical to 
the notion of strong non-commutativity in~\cite{AGHKMV11}:

\begin{definition2}
\label{def-strongly-non-commutative}
Two operations $\op_1,\op_2 \in \opsf{\obj}$ are \emph{two-sided non-commutative in $\spec{\obj}$} 
if there exist history 
$\h_0 \in \cshistories{\obj}$, 
processes $\proc_1\neq \proc_2$, 
and response values $u_1 \neq v_1$ and $u_2 \neq v_2$ in $\retsf{\obj}$
such that:
\begin{enumerate*}[label=(\roman*)]
\item $\h_0 \cdot \invres{\proc_1}{\op_1}{u_1} \cdot \invres{\proc_2}{\op_2}{v_2} \in \spec{\obj}$; and
\item $\h_0 \cdot \invres{\proc_2}{\op_2}{u_2} \cdot \invres{\proc_1}{\op_1}{v_1} \in \spec{\obj}$.
\end{enumerate*}
\end{definition2}

\begin{example}
\label{ex:set}
Revisiting \cref{ex:LOO-intro}, in a standard 
set object $\seto$
the operations $\takeop(v)$ and $\takeop(v)$ (for any $v$) are strongly non-commutative.
Indeed, we can take
any $\proc_1\neq\proc_2$,
$\h_0=\invress{\proc}{\putop(v)}{}{30pt}$ (with any $\proc \in \Proc$),
$u_1=u_2=\True$, 
and $v_1=v_2=\False$,
and we have $\invress{\proc}{\putop(v)}{}{30pt} \cdot \invress{\proc_1}{\takeop(v)}{\True}{60pt} \cdot \invress{\proc_2}{\takeop(v)}{\False}{60pt} \in \spec{\seto}$
and $\invress{\proc}{\putop(v)}{}{30pt} \cdot \invress{\proc_2}{\takeop(v)}{\True}{60pt} \cdot \invress{\proc_1}{\takeop(v)}{\False}{60pt} \in \spec{\seto}$.
\end{example}

\begin{example}
\label{ex:cons}
Consider a consensus object $\conso$ with operations $\proposeop{0}$ and 
$\proposeop{1}$ and return values $\set{0,1}$.
Its specification $\spec{\conso}$ consists of all histories $\h \in \cshistories{\conso}$ such that 
every $\proposeop{v}$ invoked in $\h$ returns the same value, which is either $v$ or
the argument of one of the previously invoked $\proposeopp$ operations.
The operations $\proposeop{0}$ and $\proposeop{1}$ 
are two-sided non-commutative.
Indeed, for any $\proc_1\neq\proc_2$ and $\h_0=\emptyseq$, we have 
$\invress{\proc_1}{\proposeop{0}}{0}{50pt} \cdot \invress{\proc_2}{\proposeop{1}}{0}{50pt}\in \spec{\conso}$ and
$\invress{\proc_2}{\proposeop{1}}{1}{50pt} \cdot \invress{\proc_1}{\proposeop{0}}{1}{50pt}\in \spec{\conso}$.
\end{example}

Examples for other objects with consensus number $>1$, such as
swap, compare-and-swap, fetch-and-add, queues, stacks,
are constructed similarly.
\iffull
In~\cref{app:cn-two-sided}, 
\fi
we show that  deterministic objects with a pair of two-sided non-commutative 
operations must have consensus numbers $>1$.
(We conjecture that the converse also holds.)

We prove 
\iffull
in~\cref{sec:non-strong-impl}
\fi
that two-sided non-commutative operations imply non-weakly mergeability:

\begin{restatable}{lemma}{lemmastrongnoncommandmerge}
\label{lemma-strong-non-comm-and-merge}
Let $\obj$ be a deterministic object and $\op_1,\op_2\in \opsf{\obj}$ be 
two-sided non-commutative operations in $\spec{\obj}$. 
Then, there exist $\h_0 \in \cshistories{\obj}$,
processes $\proc_1 \neq \proc_2$ and response values $u_1,u_2\in\retsf{\obj}$
such that $\invres{\proc_1}{\op_1}{u_1}$ and $\invres{\proc_2}{\op_2}{u_2}$ 
are not weakly mergeable in $\spec{\obj}$ after $\h_0$. 
\end{restatable}

We now apply the merge theorem and the properties 
$\SCM^\textsf{w},\TSO^\textsf{w},\RA^\textsf{w}$ from \cref{tab:results} to obtain the lower bounds
of LOO under $\SCM$ along with impossibilities for $\TSO$ and $\RA$
\iffull
(see~\cref{sec:non-strong-impl} for the proof):
\fi
\begin{restatable}{theorem}{thmtwosidelower}
\label{thm:two-side-lower}
Let $\obj$ be a deterministic object with a pair of strongly non-commutative operations $\op_1,\op_2\in \opsf{\obj}$ 
in $\spec{\obj}$.
Let $\impl$ be a spec-available implementation of $\obj$ that is consistent under a memory model $\mm$.
Then, there exist $\proc_1 \in\Proc$ and $\pi_1\in \traces{\impl(\op_1, \proc_1)}$  
such that the following hold
for $\sigma_1=\restrictmem{\pi_1}$:
\begin{enumerate}[label=(\alph*)]
\item
if $\mm = \Atomic$, then 
$\sigma_1$ either has an RMW or is not RBW; and
\item
if $\mm \in \{\TSO,\RA\}$, then 
$\sigma_1$ either has  an RMW  
or is not LTF (\ie has a fence in the middle).
\end{enumerate}
\end{restatable}

Since deterministic objects with a pair of two-sided non-commutative 
operations have consensus numbers $>1$, their wait-free implementations must rely on RMWs~\cite{wait-free}.
We therefore consider their obstruction-free implementations to obtain upper bounds
in the absence of RMWs.\footnote{It is known that 
every deterministic object has read/write obstruction-free linearizable implementations in \SCM~\cite{obstruction-free}.}
\iffull
In~\cref{sec:two-sided},
\fi
we show that every object in this class has an 
obstruction-free implementation under $\TSO$ with a fence pattern optimal
\wrt our lower bounds in contention-free executions, \ie when a process
runs solo for long enough to complete its operation. For that,
we use a variant of a universal construction from~\cite{raynal17} instantiated on top of a
$\TSO$-based obstruction-free consensus algorithm. The latter is obtained
from shared memory Paxos~\cite{disk-paxos} using our fence-insertion strategy.

Finally, in~\cref{sec:mutex},
we derive lower bounds for mutual exclusion.
We define an object $\locko$ that can be implemented by means of an 
entry section of a mutual exclusion algorithm. We show that $\locko$
has a pair of non-weakly mergeable histories, and apply
the merge theorem to obtain the lower bounds of LOO for $\Atomic$
and their counterparts for $\TSO$ and $\RA$. 
A matching upper bound for $\TSO$ is obtained by adding a single
fence to the entry section of the Bakery algorithm~\cite{bakery}.

\subsection{Snapshot and Counter}
\label{sec:snapshot_and_counter}

We next explore implementability of snapshot~\cite{AADGMS93} ($\snapshoto$) and (non-monotone) counter ($\counto$). 
The former is known to be universal \wrt a large class of objects
implementable in read/write $\Atomic$~\cite{AH90}, and the latter has been studied
extensively as a building block for randomized 
consensus (\eg~\cite{AspnesH1990consensus,aspness-random2}). 

In \cref{sec:snap-count}, we revisit and formalize \cref{ex:snapshot-intro},
and obtain lower bounds on memory events and fence structure that must be exhibited
by any consistent and spec-available implementation of snapshot and counter
under the memory models we consider.
Specifically, we obtain the following for snapshot:

\begin{restatable}{theorem}{thmsnapweaktsoscmra}
\label{thm:snap-weak-tso-scm-ra}
Let $\impl$ be a spec-available implementation of $\snapshoto$ that is consistent under 
a memory model $\mm$.
Then, there exist $\proc,\proc'\in\Proc$,
$\pi_1\in \traces{\impl(\updateopp{w}, \proc)}$ for some
$i\in \{1..m\}$ and $w\in W$, 
and $\pi_2 \in \traces{\impl(\snapop, \proc')}$ such that
the following hold for $\sigma_1=\restrictmem{\pi_1}$ and
$\sigma_2=\restrictmem{\pi_2}$:
\begin{enumerate}[leftmargin=*,align=left, label=(\alph*)]
\item \label{thm:snap-weak-scm}
if $\mm = \Atomic$, then $\sigma_1 \cdot \sigma_2$ either has an RMW 
or is not RBW; 
\item \label{thm:snap-weak-tso}
if $\mm = \TSO$, then $\sigma_1 \cdot \sigma_2$ either has  an RMW  
or is not LTF (\ie has a fence in the middle); and
\item \label{thm:snap-weak-ra}
if $\mm = \RA$, then 
\begin{enumerate*}[label=(\roman*)]
\item 
either $\sigma_1$ or $\sigma_2$ has an RMW, or
\item 
either $\sigma_1$ or $\sigma_2$ is not LTF.
\end{enumerate*}
\end{enumerate}
\end{restatable}

We show that a similar lower bounds holds for $\counto$ for
$\pi_1\in \traces{\impl(\op, \proc)}$ with $\op \in \{\inco, \deco\}$ and
$\pi_2 \in \traces{\impl(\reado, \proc')}$. 

The wait-free linearizable implementations of both $\snapshoto$ and $\counto$ under 
$\Atomic$ are well-known~\cite{AADGMS93,AH90}. The implementation of 
$\updateop$ operation uses collect followed by a write, and the 
implementation of $\snapop$ uses a sequence of three collects. 
Counter can be implemented on top of a snapshot using a single call to
$\updateop$ to implement increment and decrement, and a single call
to $\snapop$ to implement read.
Both implementations exhibit a single read-after-write
across a consecutive pair of update and read, and are therefore
optimal \wrt to the above lower bounds.

To obtain optimal upper bounds for $\TSO$, the above algorithms
are modified using the fence-insertion strategy discussed above 
that inserts a fence at the end of $\updateop$ for snapshot,
and at the end of the increment and decrement for counter.
Optimal implementations under $\RA$ are left for future work.

\paragraph{Max-register vs.~snapshot and counter.}
Our analysis yields the first sharp separation between max-register
on the one hand and snapshot and counter on the other in terms of their 
implementability under $\RA$ using only reads, writes, and fences. 
Specifically, as we show above, 
max-register can be implemented under $\RA$ using fences only at the 
beginning and the end of $\reado$ and $\writeoo$. 
On the other hand, our lower bounds for snapshot and counter 
show that this fence placement 
is insufficient to correctly
implement these objects under $\RA$. We are unaware of prior 
results separating these objects. In particular,
all of them are equivalent \wrt their power to solve consensus 
under $\Atomic$~\cite{wait-free}.


%% file: related.tex
\section{Related Work}
\label{sec:related}

Our mergeability approach is inspired by the work of Kawash~\cite{KawashThesis},
who showed that, without fences and RMWs, the critical section problem, as well as certain producer/consumer
  coordination problems, cannot be solved in a variety of weak memory
  models that were studied at that time (including TSO).
However, while Kawash considers specific tasks, 
we derive a general result
by relating mergeability of traces in the underlying memory model to mergeability at the level of the implemented
object histories.\footnote{We also note that Kawash's merge strategy for TSO traces is unnecessarily complex, while our proofs
directly exploit the local store buffers for avoiding inter-thread communication.} 
Moreover, we also use different mergeability properties to differentiate between weak memory models.

We have already discussed how the results of \cite{AGHKMV11},
which were based on a covering technique~\cite{BurnsL93}, 
are obtained by a simpler merge-based argument.
The main advantage of our approach is its applicability beyond
``strongly non-commutative operations'' (see \cref{lemma-strong-non-comm-and-merge}),
as well as the fact that we directly handle weak memory models, which is only implicit in \cite{AGHKMV11}.
In addition, \cite{AGHKMV11} is restricted to deterministic objects and implementations,
while our merge theorem avoids these assumptions by stating more precise availability requirements.

Through consensus numbers, Herlihy~\cite{wait-free} already showed that for some of the objects we consider here, such as
sets, queues and stacks, RMW operations are required in any lock-free linearizable implementation.
This result does not have any implication for obstruction-freedom.
In fact, for every object,  there is a read/write obstruction-free linearizable implementation under \Atomic
(as consensus is universal and read/write obstruction-free solvable~\cite{wait-free, obstruction-free}).
Due to our results, if the object has non-weakly mergeable operations,
any such implementation cannot be RBW.

For several objects such as snapshot, counter, max register, work stealing and even relaxations of queues, stacks, and data sketches,
there have been proposed lock-free or wait-free read/write linearizable (or variants of it) RBW implementations 
under \Atomic~\cite{AH90, AADGMS93, AAC12, CP21, CRR23, RK20}. 
None of these works relate the possibility of such implementations with mergeability properties of the objects implemented.
Morrison and Afek~\cite{MA14} show how memory fences can be eliminated on \TSO
in the implementation of work stealing by assuming that the store buffers are 
bounded in size, and using this bound in the thief implementation to guarantee that a write is propagated to main memory 
after a number of subsequent writes.
In contrast, the store buffers in the \TSO model we study are unbounded, 
and hence their implementation is not considered linearizable.

In the weak memory literature, some works studied \emph{robustness} of concurrent implementations
under \TSO and \RA, where a robust implementation cannot have any non-\SCM 
behaviors~\cite{tso-robustness,tso-robustness2,Lahav:pldi19,Bouajjani18,Margalit21}.
We note, however, that robustness does not entail that linearizability under \SCM is transferred
to linearizability under \TSO or \RA (a register implementation that uses one shared variable
is robust, but fences are needed to ensure linearizability under \TSO and \RA).
This is different from the fence-insertion strategy 
\iffull
in~\cref{coro-from-scm-to-tso} 
\fi
that transfers linearizability under \Atomic to linearizability under \TSO.
Other works studied alternatives to linearizability for \TSO and \RA~\cite{Burckhardt_2012_library_tso,SinghL23,Raad2019:libraries},
whereas we take standard linearizability as a correctness criterion.


%% file: summary_app.tex
\newcommand{\rafig}{
\begin{figure*} 
\begin{mathparpagebreakable}
  \inferrule[write]
  {\evt=\ev{\proc}{\Write(\var,v)} \\
    \pv(\proc)(\var) < t \\ t \nin \timef{\restrict{\mra}{\var}} \\\\
    \view'=\pv(\proc)[\var \mapsto t] \\
    \msg = \tup{\var,v,t,\view'}
  }
  {\tup{\mra, \pv, \fv} 
    \xrightarrow{\evt}
    \tup{\mra \cup \set{\msg}, \pv[\proc \mapsto \view'], \fv}
  }
    \and
      \inferrule[read]
      {\evt=\ev{\proc}{\Read(\var,v)} \\\\
    \tup{\var,v,\pv(\proc)(\var),\_} \in \mra
  }
  {\tup{\mra, \pv, \fv}
    \xrightarrow{\evt}
    \tup{\mra, \pv, \fv}
  }
    \and
    \inferrule[fence]
  {\evt = \ev{\proc}{\Fence} \\\\ 
    \view' = \pv(\proc) \sqcup \fv
  }
  {\tup{\mra, \pv, \fv}
    \xrightarrow{\evt} \tup{\mra, \pv[\proc \mapsto \view'], \view'}
  }
\and
    \inferrule[rmw]
  {\evt=\ev{\proc}{\RMW(\var,\vexp,\vnew)} 
    \\\\
        \tup{\var,\vexp,\pv(\proc)(\var),\_} \in \mra \\\\ t = \pv(\proc)(\var) + 1 \\
    t \nin \timef{\restrict{\mra}{\var}} \\\\
    \view'=\pv(\proc)[\var \mapsto t] \\
    \msg = \tup{\var,\vnew,t,\view'}
  }
  {\tup{\mra, \pv, \fv}
    \xrightarrow{\evt}
    \tup{\mra \cup \set{\msg}, \pv[\proc \mapsto \view'],  \fv}
  }
  \and
    \inferrule[propagate]
{   \tup{\var,\_,t,\view} \in \mra \\
	  \pv(\proc)(\var) < t  \\
    \view' = \pv(\proc) \sqcup \view
  }
  {\tup{\mra, \pv, \fv}
    \xrightarrow{\tau} 
    \tup{\mra, \pv[\proc \mapsto \view'], \fv}
  }
\end{mathparpagebreakable}
\vspace{-10pt}
\caption{Transitions of the $\RA$ model}
\label{fig:ra_steps}
\end{figure*}
}

  \section{Memory Models}
\label{sec:memory-models}

\subsection{Strongly Consistent Memory Model}

In the \emph{strongly consistent} memory model, denoted $\Atomic$,
every read returns the value written by the last previously executed 
write to the same variable, or $v_\init$ if no such write exists;
RMW is an atomic sequencing of a read and a write;
and fence is a no-op.

Formally, $\Atomic$ is the LTS whose states are mappings
$m \in\Var \to \Val$; initial state is
$\init(\Atomic)\defeq\lambda \var \ldotp v_\init$; and transitions are
given as follows, where $m[\var\mapsto v]$ denotes functional override
(returning $v$ for $\var$ and $m(y)$ for every $y\neq x$).
\begin{mathpar}
  \inferrule[write]{\actionf{\evt}= \Write(\var,v)}
  {m \xrightarrow{\evt} m[\var\mapsto v]} 
\and
\inferrule[read]{\actionf{\evt}= 
  \Read(\var,m(\var))
}{m \xrightarrow{\evt} m} 
\and
\inferrule[rmw]{\actionf{\evt}= 
  \RMW(\var,\vexp,\vnew)
  \\\\ m(\var)=\vexp}{m \xrightarrow{\evt} m[\var\mapsto \vnew]} 
\and
\inferrule[fence]{\actionf{\evt}=\Fence}{m \xrightarrow{\evt} m} 
\end{mathpar}
\noindent Note that $\Atomic$ has no internal memory steps, and hence $\traces{\Atomic} = \otraces{\Atomic}$.

\subsection{Release-Acquire Memory Model}
\label{sec:ra}

The \emph{release-acquire} memory model, denoted $\RA$~\cite{sra}, is
a fragment of the C/C++11 memory model~\cite{Batty:2011}.  
We present an operational model~\cite{promising} augmented with
explicit non-deterministic propagation steps
following~Lahav et al~\cite{Lahav21_liveness}.

In this model the memory $\mra$ collects \emph{all} previously executed
writes stored as a set of \emph{messages}. 
Each message is a tuple of the form
$\msg =\tup{\var,v,t,\view}$ where:
$\var\in\Var$ is the variable of the message,
$v\in\Val$ is the written value,
$t\in\N$ is the \emph{timestamp} of the message,
and $\view \in \View \defeq \Var \to \N$ is the \emph{view} of the message.
Timestamps are used to order messages of the same variable. 
This order does not necessarily match the execution order. 
That is, when a process
writes to memory, a message may be added with a timestamp earlier than some
existing messages.  
Views are used to ensure causal consistency:
the machine states maintain a mapping $\pv : \Proc \to \View$ that assigns for 
every process $\proc$ a view recording the largest timestamp of a message 
that process $\proc$ has observed for every shared variable.
In turn, the view in every message $\msg$ is the view 
of the process that added the message to the memory at the time it added it.

When process $\proc$ writes to a variable $\var$ it adds a corresponding message $\msg$ to $\mra$ 
with some available timestamp $t$ that is larger than the timestamp of the latest message to $\var$ 
that was observed by $\proc$, as recorded in $\pv(\proc)(\var)$.
Then, the process $\proc$ updates its view to include the timestamp $t$
and records the updated view in the added message $\msg$.

To model inter-process communication, processes may
non-deterministically advance their views in silent internal memory
steps.  To do so, a process $\proc$ picks from the memory $\mra$ some
message $\msg$ of variable $\var$ whose timestamp is larger than the
the latest timestamp of messages to $\var$ that was observed by
$\proc$, as recorded in $\pv(\proc)(\var)$.  Then, the process $\proc$
inherits into its own view the view of $\msg$.  For the formal
definition, we let
$\view_1 \sqcup \view_2 \defeq \lambda \var \ldotp
\max\set{\view_1(\var),\view_2(\var)}$.  With this definition, the
process view after propagating its view using message $\msg$ with view
$\view$ is set to $\pv(\proc) \sqcup \view$.

Read steps are simple: when a process $\proc$ reads a variable $\var$ it retrieves the value
stored in the message of $\var$ currently viewed by $\proc$, \ie the message of $\var$ whose timestamp is $\pv(\proc)(\var)$.

Compare-and-swap operations can be thought of as an atomic combination
of a read and a write.  To ensure their atomicity, the message $\msg$
added by a successful RMW must be the immediate successor (in
timestamp order) of the message from which the RMW read its value.
This is enforced by using integer timestamps, and requiring that the
timestamp of the added message is the successor of the timestamp of
the read message in successful RMWs.  Namely, no message later added
to $\mra$ can have timestamp between the read message and the message
added by a successful RMW.

Finally, fences are supported using an additional view $\fv \in \View$
in the machine state.  When a fence is executed by process $\proc$ the
process inherits $\fv$ into its own view (setting its new view to
$\view' = \pv(\proc) \sqcup \fv$), and the new view is also stored in
$\fv$.

$\RA$'s states are tuples $\tup{\mra,\pv,\fv}$, with the initial state
$\init(\RA) \defeq \tup{\mra_\init,\pv_\init,\fv_\init}$,
where
$\mra_\init \defeq \set{\tup{\var,v_\init,0,\lambda \var \ldotp 0} \st
  \var\in\Var}$,
$\pv_\init \defeq \lambda \proc \ldotp \lambda \var \ldotp 0$, and
$\fv_\init \defeq \lambda \var \ldotp 0$.  The formal transitions are
given in \cref{fig:ra_steps}.  We use the functions $\varf{\msg}$,
$\valmf{\msg}$, $\timef{\msg}$, $\viewf{\msg}$ to retrieve the
variable ($\var$), value ($v$), timestamp ($t$), and view ($\view$) of
a message $\msg$.  We lift $\timef{\msg}$ to sets of messages
pointwise (\ie $\timef{S} \defeq \set{\timef{\msg} \st \msg \in S}$),
and
write $\restrict{\mra}{\var}$ for the set of all messages $\msg\in\mra$ with $\varf{\msg}=\var$.
\rafig

\subsection{Litmus Tests Revisited}
\label{sec:relat-strengths-memo}

It is straightforward to see how our operational models exhibits the
intended outcomes. First consider SB. Under \TSO the write to $x$
may be cached in the local buffers of $\cproc_1$. 
When $\cproc_2$ reads from $x$, it
reads from shared memory, which still contains the original value $0$.
This behavior is impossible under \Atomic, since
the writes occur immediately in the shared memory. Under \RA, 
it is possible for process $\cproc_1$ (resp.,
$\cproc_2$) to read a ``stale'' value for $y$ (resp., $x$) since the
read may occur before the view $\pv(\cproc_1)$ (resp.,
$\pv(\cproc_2)$) is updated.

Now consider IRIW. Under both \Atomic and \TSO, when updates to $x$
and $y$ occur in shared memory, they are seen by both $\cproc_3$ and
$\cproc_4$, meaning that if $a = c = 1$, then either $b = 1$ or $d=1$
(or both). However, under \RA, the following execution is
possible:
\begin{enumerate*}[label=\bfseries(\roman*)]
\item Both writes to $x$ and $y$ are executed.
\item Process $\proc_2$ updates its view of $x$ (via a propagation
  step), then performs its two reads, resulting in $a=1$ and $b=0$.
\item Process $\proc_3$ updates its view of $y$ (via a
  propagation step), then performs its two reads, resulting in $c=1$ and
  $d=0$.
\end{enumerate*}

Consider MP. Causal consistency for \Atomic is immediate, and
for \TSO is due to the FIFO ordering of propagation from local
buffers. Under \RA, the reasoning is a bit more subtle. If $a=1$, then
the following sequence of events must have occurred:
\begin{enumerate*}[label=\bfseries(\roman*)]
\item both writes in $\cproc_1$,
\item a propagation step that updates $\pv(\cproc_2)(y)$, \label{RA:step2}
\item the first read in $\cproc_2$.
\end{enumerate*}
Notice that after the write to $x$, a new message
$m_x \defeq \tup{x, 1, t_x, \_}$ becomes available, and
$\pv(\cproc_1)(x) = t_x$. After the write to $y$, a second message
$m_y \defeq \tup{y, 1, t_y,V}$ becomes available, where $V(x) = t_x$,
i.e., the view of $x$ for $m_y$ corresponds to $m_x$. When the
propagation in $\cproc_2$ occurs (i.e., step \labelcref{RA:step2}
above), $\pv(\cproc_2)(y)$ is updated to $t_y$ and $\pv(\cproc_2)(x)$ is  updated to $t_x$. 
Thus, when the second read in $\cproc_2$
occurs, it must read from $m_x$, setting $b = 1$.

Finally, \Atomic behavior can be recovered in both \TSO and \RA by using
\emph{fence} instructions to enforce propagation of writes. Let
SB$_{F}$ be SB but with a fence introduced between the write and the
read in both processes. Execution of SB$_{F}$ does not end with
$a = b = 0$ under both \TSO and \RA. However, note that the semantics
of fences are different for \TSO and \RA. In \TSO, executing a fence
ensures that all writes in the write buffer of the executing process
are drained. In \RA, fences work in pairs: the second fence that is
executed receives all ``happens-before'' information known by the
first fence. Thus, in SB$_{F}$, each fence has knowledge of the
preceding write that ``happened-before'' in the same
process. Regardless of how the instructions of SB$_{F}$ are
interleaved, the second fence to be executed receives information
about the write in the other process, making it impossible for the
read following the fence to read from the initial value of the shared
variable.
  
To recover \Atomic behavior for IRIW under \RA memory, one must
introduce a fence between the two reads in both $\cproc_2$ and
$\cproc_3$. In this version of IRIW, in any execution, where
$\cproc_2$ reads $x$ from the write in $\cproc_1$, and $\cproc_3$
reads $y$ from the write in $\cproc_4$ (thus setting $a = c = 1$), the
second fence to be executed (in $\cproc_2$ or $\cproc_3$) receives
``happens-before'' information from the first fence so that
$b = d = 0$ is impossible.

\wellbehaved*
\begin{proof}
For $\mm=\Atomic$ the claim is obvious:
$\sigma \in \traces{\Atomic}$ entails $\initial{\Atomic} \xrightarrow{\sigma}_\mm q$ for some $q\in\states{\Atomic}$;
every state of $\Atomic$ is stable; and since $\Atomic$ is deterministic, 
$\initial{\Atomic} \xrightarrow{\sigma}_\Atomic q$ and $\sigma \cdot \sigma' \in \traces{\Atomic}$ 
imply $\sigma' \in \otraces{\Atomic,q}$.

For $\TSO$, we can simulate the execution generating $\sigma$ in
$\Atomic$ by propagating in a $\tau$ step of $\TSO$ every write
immediately after the write is performed.  This way we reach a state
$q$ in which all store buffers are empty, which is stable.
Furthermore, for every $\sigma'$ such that
$\sigma \cdot \sigma' \in \traces{\Atomic}$, we can continue the
simulation of $\Atomic$ in $\TSO$ by propagating writes immediately
after they are performed, and obtain the
$\sigma' \in \otraces{\TSO,q}$.

The argument for $\RA$ is similar, except that instead of one propagation step after every write,
we use $\size{\Proc}-1$ $\tau$ steps of $\RA$ in which all processes, except for the writer itself,
acquire the view stored in the added message. We maintain the invariant that after every simulation step 
all thread views are equal and they all point to the most recently added message.
When this holds, we must be at a stable state.
\end{proof}

%% file: merge_app.tex
\section{Proofs for Section \ref{sec:merge}}
\label{app:proofs_merge}

We formally define the shuffle operation.
\begin{definition2}
The \emph{shuffle product} (or, \emph{shuffle}) of sequences $s_1,s_2 \in X^*$,
denoted $s_1 \shuffle s_2$, is inductively defined as follows:
\begin{mathpar}
\inferrule{~}{\emptyseq\in\emptyseq \shuffle \emptyseq} 
\and
\inferrule{s\in s_1 \shuffle s_2}{s \cdot \tup{x} \in (s_1 \cdot \tup{x}) \shuffle s_2}
\and
\inferrule{s\in s_1 \shuffle s_2}{s \cdot \tup{x} \in  s_1 \shuffle (s_2 \cdot \tup{x})}
\end{mathpar}
\end{definition2}

\begin{lemma}
\label{lem:shuffle_restrict}
Let $s_1,s_2\in X^*$ and let $Y\subseteq X$.
Then, $(\restrict{s_1}{Y}) \shuffle (\restrict{s_2}{Y}) = \restrict{(s_1 \shuffle s_2)}{Y}$.
\end{lemma}

\subsection{Mergeability in \Atomic}

\begin{theorem}
  \label{thm:atomic-merge}
  Every RW-RBW observable memory sequence $\sigma_1$
  and arbitrary observable memory sequence $\sigma_2$ with $\procsf{\sigma_1} \cap \procsf{\sigma_2} = \emptyset$
  are weakly mergeable in $\Atomic$.
\end{theorem}
\begin{proof}
Let $m_0 \in \states{\Atomic}=\Var \to \Val$ such that $\sigma_1,\sigma_2 \in \otraces{\Atomic,m_0}$.
We show that there exists $\sigma \in \sigma_1 \shuffle \sigma_2$ such that $\sigma \in \otraces{\Atomic,m_0}$.
Let $\sigma_1^{\textsf{r}} \in \Reads^*$ be the longest prefix of $\sigma_1$ consisting solely of read events,
and let $\sigma_1^{\textsf{w}}$ such that $\sigma_1 = \sigma_1^{\textsf{r}} \cdot \sigma_1^{\textsf{w}}$.
Define $\sigma \defeq \sigma_1^{\textsf{r}} \cdot \sigma_2 \cdot \sigma_1^{\textsf{w}}$.
Then, $\sigma\in \sigma_1 \shuffle \sigma_2$.
It is also easy to verify that $\sigma \in \otraces{\Atomic,m_0}$.
Indeed, the reads in $\sigma_1^{\textsf{r}}$ all read the value in $m_0$ and do not affect the memory state,
and the reads in $\sigma_1^{\textsf{w}}$ read from $\sigma_1^{\textsf{w}}$ itself.
Specifically, a read in $\sigma_1^{\textsf{w}}$ reading from variable $\var$
must appear after a write in $\sigma_1^{\textsf{w}}$ writing to $\var$,
or else $\sigma_1^{\textsf{w}}$ is not RBW (with $k_1=1$).
\end{proof}

\subsection{Mergeability in \TSO}

To prove mergeability properties in \TSO, we use the following simple observations about the $\TSO$ system.
First, stable states in $\TSO$ are states where all store buffers are empty:

\begin{proposition}
\label{prop:tso0}
If $\tup{m, b} \in \states{\TSO}$ is stable, then $b = b_\init$.
\end{proposition}

Second, fences have no effect when buffers are empty:

\begin{proposition}
\label{prop:tso01}
Suppose that $\sigma \in \otraces{\TSO,\tup{m, b_\init}}$,
and let $\sigma^{\textsf{lf}}$ be a sequence of fence events.
Then, $\sigma^{\textsf{lf}} \cdot \sigma \in \otraces{\TSO,\tup{m, b_\init}}$.
\end{proposition}

Third, in the end of the trace, buffers can drained and then fences can be executed:

\begin{proposition}
\label{prop:tso02}
Suppose that $\sigma \in \otraces{\TSO,\tup{m, b}}$,
and let $\sigma^{\textsf{tf}}$ be a sequence of fence events.
Then, $\sigma \cdot \sigma^{\textsf{tf}} \in \otraces{\TSO,\tup{m, b}}$.
\end{proposition}

Forth, any solo execution in $\TSO$ that only performs reads and writes 
(no RMWs and no fences)
can be simulated by an execution that only uses the store buffer, 
and never propagates its writes to the main memory:

\begin{proposition}
\label{prop:tso1}
Let $\sigma \in \otraces{\TSO,\tup{m, b}}$ be a non-empty solo 
observable memory sequence that consists solely of read and write events,
and let $\proc=\procf{\sigma[1]}$.
Then, $\tup{m, b} \xrightarrow{\sigma}_\TSO \tup{m, b[\proc \mapsto b(\proc) \cdot \beta]}$
for some $\beta \in (\Var \times \Val)^*$.
\end{proposition}

Fifth, if there is an execution not involving process $\proc$ and the initial store buffer of $\proc$ is empty,
then there is a similar execution from any initial store buffer of $\proc$:

\begin{proposition}
\label{prop:tso2}
Let $\sigma \in \otraces{\TSO,\tup{m, b}}$ for some $b$ 
such that $b(\proc)=\emptyseq$ for every $\proc\in\Proc \setminus \procsf{\sigma}$.
Then, $\sigma \in \otraces{\TSO,\tup{m, b'}}$
for every $b'$ such that $b'(\proc)=b(\proc)$ for every $\proc\in \procsf{\sigma}$.
\end{proposition}
\begin{proof}
We can consider an execution fragment inducing $\sigma$ that never propagates
from buffers of processes that are not in $\procsf{\sigma}$.
\end{proof}

Sixth, solo executions of different processes that never propagate from buffers can be arbitrarily interleaved:

\begin{proposition}
\label{prop:tso3}
Let $\sigma_1$ and $\sigma_2$ be non-empty solo 
observable memory sequences that consist solely of read and write events.
Let $\proc_1=\procf{\sigma[1]}$ and $\proc_2=\procf{\sigma[2]}$,
and suppose that $\proc_1 \neq \proc_2$.
Let $m\in\Var\to\Loc$ and $\beta_1,\beta_2  \in (\Var \times \Val)^*$
such that $\tup{m, b_\init} \xrightarrow{\sigma_1}_\TSO \tup{m, b_\init[\proc_1 \mapsto \beta_1]}$
and $\tup{m, b_\init} \xrightarrow{\sigma_2}_\TSO \tup{m, b_\init[\proc_2 \mapsto \beta_2]}$.
Then, $\tup{m, b_\init} \xrightarrow{\sigma}_\TSO \tup{m, b_\init[\proc_1 \mapsto \beta_1,\proc_2 \mapsto \beta_2]}$
for every $\sigma \in \sigma_1 \shuffle \sigma_2$.
\end{proposition}
\begin{proof}
By induction on the shuffle formation $\sigma \in \sigma_1 \shuffle \sigma_2$.
\end{proof}

We can now prove our two mergeability properties for $\TSO$:

\begin{theorem}
\label{thm:tso-merge1}
Every solo-RWF-LTF observable memory sequence $\sigma_1$
and arbitrary observable memory sequence $\sigma_2$ with $\procsf{\sigma_1} \cap \procsf{\sigma_2} = \emptyset$
are weakly mergeable in $\TSO$.
\end{theorem}

\begin{proof}
Let $\tup{m_0, b_0}$ be a stable state
such that $\sigma_1,\sigma_2 \in \otraces{\TSO,\tup{m_0, b_0}}$.
By \cref{prop:tso0}, we have $b_0 = b_\init$.
Let $\sigma_1^{\textsf{lf}} \in \Fences^*$ be the longest prefix of $\sigma_1$ consisting solely of fence events,
$\sigma_1^{\textsf{tf}} \in \Fences^*$ be the longest suffix of $\sigma_1$ consisting solely of fence events,
and let $\sigma_1'$ such that $\sigma_1 = \sigma_1^{\textsf{lf}} \cdot \sigma_1' \cdot \sigma_1^{\textsf{tf}}$.
Define $\sigma \defeq \sigma_1^{\textsf{lf}} \cdot \sigma_1' \cdot \sigma_2  \cdot \sigma_1^{\textsf{tf}}$.
Then, $\sigma\in \sigma_1 \shuffle \sigma_2$.
We show that $\sigma \in \otraces{\TSO,\tup{m_0, b_\init}}$.
By \cref{prop:tso01,prop:tso02}, it suffices to show that $\sigma_1' \cdot \sigma_2 \in \otraces{\TSO,\tup{m_0, b_\init}}$.
If $\sigma_1$ is empty, then $\sigma=\sigma_2 \in \otraces{\TSO,\tup{m_0, b_\init}}$.
Otherwise, let $\proc=\procf{\sigma[1]}$.
By \cref{prop:tso1}, we have 
$\tup{m_0, b_\init} \xrightarrow{\sigma_1}_\TSO \tup{m_0, b_\init[\proc \mapsto \beta]}$
for some $\beta \in (\Var \times \Val)^*$,
and by \cref{prop:tso2}, we have 
$\sigma_2 \in \otraces{\TSO,\tup{m_0, b_\init[\proc \mapsto \beta]}}$.
It follows that $\sigma \in \otraces{\TSO,\tup{m_0, b_\init}}$.
\end{proof}

\begin{theorem}
\label{thm:tso-merge2}
Every two solo-RW observable memory sequences, $\sigma_1$ and $\sigma_2$,
with $\procsf{\sigma_1} \cap \procsf{\sigma_2} = \emptyset$
are strongly mergeable in $\TSO$.
\end{theorem}

\begin{proof}
Let $\tup{m_0, b_0}$ be a stable state
such that $\sigma_1,\sigma_2 \in \otraces{\TSO,\tup{m_0, b_0}}$.
Let $\sigma \in \sigma_1 \shuffle \sigma_2$.
By \cref{prop:tso0}, we have $b_0 = b_\init$.
First, if $\sigma_1$ is empty, then $\sigma=\sigma_2 \in \otraces{\TSO,\tup{m_0, b_\init}}$.
Similarly, if $\sigma_2$ is empty, then $\sigma=\sigma_1 \in \otraces{\TSO,\tup{m_0, b_\init}}$.
Otherwise, let $\proc_1=\procf{\sigma[1]}$ and $\proc_2=\procf{\sigma[2]}$.
By \cref{prop:tso1}, we have $\tup{m_0, b_\init} \xrightarrow{\sigma_1}_\TSO \tup{m_0, b_\init[\proc_1 \mapsto \beta_1]}$
and $\tup{m_0, b_\init} \xrightarrow{\sigma_2}_\TSO \tup{m_0, b_\init[\proc_2 \mapsto \beta_2]}$
for some $\beta_1,\beta_2 \in (\Var \times \Val)^*$.
By \cref{prop:tso3}, we obtain that 
$\tup{m_0, b_\init} \xrightarrow{\sigma}_\TSO \tup{m_0, b_\init[\proc_1 \mapsto \beta_1,\proc_2 \mapsto \beta_2]}$
and so $\sigma \in \otraces{\TSO,\tup{m_0, b_\init}}$.
\end{proof}

\subsection{Mergeability in \RA}

To prove mergeability properties in \RA, we use the following observations about the $\RA$ system.

The memory and view in \RA configurations are only increasing:

\begin{proposition}
\label{prop:ra_mon}
Suppose that $\tup{\mra, \pv, \fv} \xrightarrow{\rho}_\RA \tup{\mra', \pv', \fv'}$.
Then:
\begin{itemize}
\item $\mra' = \mra \uplus \mra_1$ for some finite set of messages $\mra_1$.
\item $\pv'(\proc)(\var) \geq \pv(\proc)(\var)$ for every $\proc\in\Proc$ and $\var\in\Var$.
\item $\fv'(\var) \geq \fv(\var)$ for every $\var\in\Var$.
\end{itemize}
\end{proposition}
\begin{proof}
By inspecting all steps of $\RA$.
\end{proof}

In every step of \RA only a single process is involved,
and the view of the other processes do not matter and are not changed.
For the formal statement, we define
the \emph{participating process} in a transition
$\mathit{tr}=\tup{\tup{\mra, \pv, \fv}, l, \tup{\mra', \pv', \fv'}} \in \trans{\RA}$,
denoted $\procf{\mathit{tr}}$,
to be the (unique) process $\proc\in\Proc$ such that 
\begin{enumerate*}[label=(\roman*)]
\item $\procf{l}=\proc$ if $l \in \MemEvs$;
and \item $\pv(\proc) \neq \pv'(\proc)$ if $l = \tau$.
\end{enumerate*}

\begin{proposition}
\label{prop:ra_proc_step}
Let $\mathit{tr}=\tup{\tup{\mra, \pv, \fv}, l, \tup{\mra', \pv', \fv'}} \in \trans{\RA}$
and let $\proc=\procf{\mathit{tr}}$.
Then, for every $\pv^* \in \Proc \to \View$,
we have $\tup{\mra, \pv^*[\proc \mapsto \pv(\proc)], \fv} \xrightarrow{l}_\RA \tup{\mra', \pv^*[\proc \mapsto \pv'(\proc)], \fv'}$.
\end{proposition}
\begin{proof}
By inspecting all steps of $\RA$.
\end{proof}

There is never need to propagate messages to threads that do not participate in a certain execution.
For the formal statement, 
for an execution fragment $\alpha=\tup{q_0, l_1, q_1, l_2 \til l_n, q_n}$ of $\RA$,
we define the \emph{participating processes} in $\alpha$,
denoted $\procsf{\alpha}$, by $\procsf{\alpha} \defeq \set{\procf{q_{i-1}, l_i, q_i} \st 1 \leq i \leq n}$.

\begin{proposition}
\label{prop:ra_proc}
Suppose that $\tup{\mra_0, \pv_0, \fv_0} \xrightarrow{\rho}_\RA \tup{\mra, \pv, \fv}$.
Then, some execution fragment $\alpha$ of $\RA$ satisfies:
\begin{enumerate*}[label=(\roman*)]
\item $\alpha$ starts in $\tup{\mra_0, \pv_0, \fv_0}$;
\item $\alpha$ ends in a state $\tup{\mra, \pv', \fv}$ where 
$\pv'(\proc)=\pv(\proc)$ for every $\proc\in\procsf{\alpha}$;
\item $\restrictmem{\rho} = \restrictmem{\tracef{\alpha}}$;
and \item $\procsf{\alpha} = \procsf{\restrictmem{\rho}}$.
\end{enumerate*}
\end{proposition}

Without CAS, only the relative order of the timestamps matters.
This means that if we start from a stable state where all messages are propagated to all processes,
then nothing can force us to use particular numbers as timestamps. 

\begin{proposition}
\label{prop:ra_T}
Let $\tup{\mra_0, \pv_0, \fv_0}$ be a stable state.
Suppose that  $\tup{\mra_0, \pv_0, \fv_0} \xrightarrow{\rho}_\RA \tup{\mra_0 \uplus \mra, \pv, \fv}$
where $\rho$ does not have CAS events.
Then, for every finite set $T \subseteq \N \setminus \timef{\mra_0}$, we have 
$\tup{\mra_0, \pv_0, \fv_0} \xrightarrow{\rho}_\RA \tup{\mra_0 \uplus \mra', \pv', \fv'}$
for some $\mra'$, $\pv'$, and  $\fv$ such that $T \cap \timef{\mra'} = \emptyset$.
\end{proposition}

Two executions with disjoint sets of participating processes
that use disjoint sets of timestamps for their added messages,
and such that only one of them may use fences, cannot observe one another,
and thus can be strongly merged.

\begin{proposition}
\label{prop:ra_main}
Let $q_0=\tup{\mra_0, \pv_0, \fv_0} \in \states{\RA}$,
and let
\begin{align*}
\alpha_1 &=\tup{q_0, l_1, \tup{\mra_1, \pv_1, \fv_1}, l_2 \til l_{n_1}, \tup{\mra_{n_1}, \pv_{n_1}, \fv_{n_1}}} \\
\alpha_2 &=\tup{q_0, l_1', \tup{\mra_1', \pv_1', \fv_1'}, l_2' \til l_{n_2}', \tup{\mra_{n_2}', \pv_{n_2}', \fv_{n_2}'}} 
\end{align*}
be two execution fragments of $\RA$.
Suppose that 
\begin{enumerate*}[label=(\roman*)]
\item $\procsf{\alpha_1} \cap \procsf{\alpha_2} = \emptyset$;
\item $l_1 \cdots l_{n_1}$ has no fence events;
and \item $\timef{\mra_{n_1} \setminus \mra_0} \cap \timef{\mra_{n_2}' \setminus \mra_0} = \emptyset$.
\end{enumerate*}
Let $\rho \in (l_1 \cdots l_{n_1}) \shuffle (l_1' \cdots l_{n_2}')$.
Then, $q_0 \xrightarrow{\rho}_\RA \tup{\mra_0 \cup (\mra_{n_1} \setminus \mra_0) \cup (\mra_{n_2}' \setminus \mra_0), \pv, \fv_{n_2}'}$,
where
$$\pv \defeq \lambda \proc \ldotp 
\begin{cases} 
\pv_{n_1}(\proc) & \proc \in \procf{\alpha_1} \\ 
\pv_{n_2'}(\proc) & \proc \nin \procf{\alpha_1}
\end{cases}$$
\end{proposition}
\begin{proof}
By induction on the shuffle formation $\rho \in (l_1 \cdots l_{n_1}) \shuffle (l_1' \cdots l_{n_2}')$.
\end{proof}

We now have all the ingredients to prove the first mergeability
property for $\RA$. 

\begin{theorem}
\label{thm:ra-merge1}
\label{thm:ra-merge3}
\label{thm:ra-merge4}
Observable memory sequences $\sigma_1$ and $\sigma_2$ with $\procsf{\sigma_1} \cap \procsf{\sigma_2} = \emptyset$
are strongly mergeable in $\RA$ provided that either
\begin{enumerate}
\item $\sigma_1$ is RW; or \label{thm:RAM-Strong1}
\item $\sigma_1$ is RWF-PPTF and $\sigma_2$ is PPTF; or \label{thm:RAM-Strong2}
\item $\sigma_1$ is RWF-PPLF and $\sigma_2$ is PPLF. \label{thm:RAM-Strong3}
\end{enumerate}
\end{theorem}
\begin{proof}[Proof of \cref{thm:ra-merge1}~\eqref{thm:RAM-Strong1}]
Let $q_0 = \tup{\mra_0, \pv_0, \fv_0}$ be a stable state
such that $\sigma_1,\sigma_2 \in \otraces{\RA,q_0}$. 
Let $\sigma \in \sigma_1 \shuffle \sigma_2$.
We show that $\sigma \in \otraces{\RA,q_0}$.
Let $\rho_1, \rho_2 \in \traces{\RA,q_0}$ such that $\sigma_i = \restrictmem{\rho_i}$ ($i\in\set{1,2}$).

By \cref{prop:ra_mon}, there exist $\mra_1^0,\pv_1^0,\fv_1^0$ and $\mra_2^0,\pv_2^0,\fv_2^0$ such that \smallskip 

\hfill{} $q_0  \xrightarrow{\rho_1}_\RA \tup{\mra_0 \uplus \mra_1^0, \pv_1^0, \fv_1^0}$
and 
$q_0  \xrightarrow{\rho_2}_\RA \tup{\mra_0 \uplus \mra_2^0, \pv_2^0, \fv_2^0}$.\hfill{}

\smallskip\noindent
By \cref{prop:ra_T}, 
we have $q_0 \xrightarrow{\rho_1}_\RA \tup{\mra_0 \uplus \mra_1^*, \pv_1^*, \fv_1^*}$
for $\mra_1^*$, $\pv_1^*$, and $\fv_1^*\in\View$ such that
$\timef{\mra_2^0} \cap \timef{\mra_1^*} = \emptyset$.
Moreover, by \cref{prop:ra_proc}, there exists an execution fragment 
$\alpha_1=\tup{q_0, l_1, \tup{\mra_1, \pv_1, \fv_1}, l_2 \til l_{n_1}, \tup{\mra_{n_1}, \pv_{n_1}, \fv_{n_1}}}$ of $\RA$
such that $\sigma_1 = \restrictmem{\tracef{\alpha_1}}$, $\procsf{\alpha_1} = \procsf{\sigma_1}$,
and $\mra_{n_1} =  \mra_0 \uplus \mra_1^*$.
Since 
$q_0  \xrightarrow{\rho_2}_\RA \tup{\mra_0 \uplus \mra_2, \pv_2, \fv_2}$,
by \cref{prop:ra_proc}, there exists an execution fragment 
\[
  \alpha_2=\tup{q_0, l_1', \tup{\mra_1', \pv_1', \fv_1'}, l_2' \til l_{n_2}', \tup{\mra_{n_2}', \pv_{n_2}', \fv_{n_2}'}}
\]
of $\RA$
such that $\sigma_2 = \restrictmem{\tracef{\alpha_2}}$, $\procsf{\alpha_2} = \procsf{\sigma_2}$,
and $\mra_{n_2}' =  \mra_0 \uplus \mra_2^0$.
By \cref{lem:shuffle_restrict},
there exists $\rho \in (l_1 \cdots l_{n_1}) \shuffle (l_1' \cdots l_{n_2}')$
such that $\sigma = \restrictmem{\rho}$.
Finally, by \cref{prop:ra_main}, we have $\rho \in \traces{\RA,q_0}$,
and so  $\sigma \in \otraces{\RA,q_0}$.
\end{proof}
\begin{proof}[Proof of \cref{thm:ra-merge1}~\eqref{thm:RAM-Strong2}]
Let $q_0 = \tup{\mra_0, \pv_0, \fv_0}$ be a stable state
such that $\sigma_1,\sigma_2 \in \otraces{\RA,q_0}$.
Consider the traces $\sigma_1'$ and $\sigma_2'$ obtained by removing from 
$\sigma_1$ and $\sigma_2$ (respectively) the trailing fence of each process (if such fence exists).
Since we only removed fences, is easy to see that $\sigma_1',\sigma_2' \in \otraces{\RA,q_0}$.
Then, by \cref{thm:ra-merge1}, $\sigma_1'$ and $\sigma_2'$ are strongly mergeable in $\RA$.
Moreover, since the missing fence steps are the last steps performed by each process,
we can add them in any arbitrary position after the last event of the corresponding process.
Every sequence generated this way is still an observable trace is in $\otraces{\RA,q_0}$.
\end{proof}
\begin{proof}[Proof of \cref{thm:ra-merge1}~\eqref{thm:RAM-Strong3}]
The proof is symmetric to the proof of \cref{thm:ra-merge1}~\eqref{thm:RAM-Strong2}.
\end{proof}

\begin{theorem}
\label{thm:ra-merge2}
Every RWF-LTF observable memory sequence $\sigma_1$ 
and arbitrary observable memory sequence $\sigma_2$ with $\procsf{\sigma_1} \cap \procsf{\sigma_2} = \emptyset$
are weakly mergeable in $\RA$. 
\end{theorem}
\begin{proof}
Let $q_0 = \tup{\mra_0, \pv_0, \fv_0}$ be a stable state
such that $\sigma_1,\sigma_2 \in \otraces{\RA,q_0}$.
Consider the trace $\sigma_1'$ obtained by removing from $\sigma_1$ the leading and trailing fences.
Since we only removed fences, is easy to see that $\sigma_1' \in \otraces{\RA,q_0}$.
Then, by \cref{thm:ra-merge1}, $\sigma_1'$ and $\sigma_2$ are strongly mergeable in $\RA$.
In particular, $\sigma_2 \cdot \sigma_1' \in \otraces{\RA,q_0}$.
Next, we add back the fences we removed from $\sigma_1$
in the very beginning and very end of $\sigma_2 \cdot \sigma_1'$.
Since these fence steps are the first/last steps performed by each process,
the generated observable trace is in $\otraces{\RA,q_0}$.
In particular, since $q_0$ is stable, the execution of a fence from $q_0$ has no effect
(the process view is not changed).
\end{proof}


%% file: impos_app.tex
\section{Proofs for Section \ref{sec:impos-obj}}
\label{app:impos-obj}

\begin{proposition}
\label{lem:lin_shuffle}
Let $\h \in \h_1 \shuffle \h_2$ for some $\h_1,\h_2 \in \cshistories{\obj}$.
Suppose that $\h \sqsubseteq \seqh$ for some $\seqh \in \cshistories{\obj}$.
Then, $\seqh \in \h_1 \shuffle \h_2$.
\end{proposition}

\begin{lemma}
\label{lem:impl_sproc}
$\traces{\impl}$ is closed under $\reorder{\sproc}{\cdot}$,
\ie $\reorder{\sproc}{\pi} \subseteq \traces{\impl}$ whenever $\pi \in \traces{\impl}$.
\end{lemma}

\begin{proof}
The proof is by induction on the length of $\pi \in \traces{\impl}$. First, if $\pi=\emptyseq$,
then\\ $\reorder{\sproc}{\pi} = \emptyseq \in \traces{\impl}$. Hence, the base holds.

Suppose that $\reorder{\sproc}{\pi} \subseteq \traces{\impl}$ for all $\pi \in \traces{\impl}$
such that $\size{\pi}=k$, and consider $\pi' = \pi \cdot \tup{x} \in \traces{\impl}$. 
Then,
\begin{equation}
\reorder{\sproc}{\pi'} = \{u \cdot \tup{x} \cdot v \mid u \cdot v \in \reorder{\sproc}{\pi} \wedge 
\restrict{u}{\procf{x}} = \restrict{\pi}{\procf{x}} \}
\label{eq:reorder}
\end{equation}
We prove that $\reorder{\sproc}{\pi'} \subseteq \traces{\impl}$.

Let $\pi'' \in \reorder{\sproc}{\pi'}$. Then, $\pi'' = u \cdot \tup{x} \cdot v$ for some $u$, $v$
such that $u \cdot v \in \reorder{\sproc}{\pi}$ and $\restrict{u}{\procf{x}} = \restrict{\pi}{\procf{x}}$.
Since by the induction hypothesis, $u\cdot v \in \traces{\impl}$, we have  $u \in \traces{\impl}$. 
Let $\bar{q}$ be the state of $\sys{\impl}$ at the end of $u$.
By~(\ref{eq:reorder}), $\bar{q}(\procf{x})$ is also the state of $\procf{x}$ at the end
of $\pi$. Since $\pi \cdot \tup{x} \in \traces{\impl}$, by the definition of 
$\sys{\impl}$, there exists a state $q'$ of $\sysp{\impl}{\procf{x}}$ such that
$\tup{\bar{q}(\procf{x}), x, q'}$ is a transition of $\sysp{\impl}{\procf{x}}$.
Thus, the definition of $\sys{\impl}$ implies that
$\tup{\bar{q}, x, \bar{q}[\procf{x} \mapsto q']}$ is a transition
of $\sys{\impl}$, and therefore, $u \cdot \tup{x} \in \traces{\impl}$.

It remains to show that there exists an execution fragment $\alpha$ of $\sys{\impl}$
starting from $\bar{q}[\procf{x} \mapsto q']$ such that $\tracef{\alpha}=v$.
Since $\reorder{\sproc}{\pi}$ preserves the process order
in $\pi$, for all $p\in \procsf{u}$, $\restrict{u}{p}$ is a prefix of 
$\restrict{\pi}{p}$. Since by~(\ref{eq:reorder}), 
$\procf{x} \not\in \procsf{v}$, for all $p\in \procsf{v}$,
$\bar{q}[\procf{x} \mapsto q'](p)$ is equal to the state of $p$ at the end
of some prefix of $\restrict{\pi}{p}$. Since $\pi \in \traces{\impl}$,
for each $p \in \procsf{v}$, there exists an execution fragment $\alpha_p$
of $\sysp{\impl}{p}$ starting from $\bar{q}[\procf{x} \mapsto q'](p)$
such that $\tracef{\alpha_p} = \restrict{v}{p}$. Thus, by the definition
of $\sys{\impl}$, the interleaving $\alpha$ of 
$\{\alpha_p \mid p \in \procsf{v}\}$ such that
$\tracef{\alpha} = v$ is an execution fragment of $\sys{\impl}$, as needed.
\end{proof}

\begin{lemma}
\label{lem:impl_shuffle}
Suppose that $\pi_0 \cdot \pi_1 \in \traces{\impl}$
and $\pi_0 \cdot \pi_2 \in \traces{\impl}$
for some sequences $\pi_0, \pi_1,\pi_2$
such that $\restrictobj{\pi_0} \in \chistories{\obj}$
and $\procsf{\pi_1} \cap \procsf{\pi_2} = \emptyset$.
Then, $\pi_0 \cdot \pi \in \traces{\impl}$ for every $\pi \in \pi_1 \shuffle \pi_2$.
\end{lemma}

\begin{proof}
Since $\pi_0 \in \traces{\impl}$, there exists an execution 
fragment $\alpha_0$ of $\sys{\impl}$ such that $\tracef{\alpha_0}=\pi_0$.
Let $\bar{q_0}\in \states{\sys{\impl}}$ be the last state of $\alpha_0$.
Since $\restrictobj{\pi_0} \in \chistories{\obj}$, all $\sysp{\impl}{p}$
are in the state $\bot$ at the end of $\pi_0$.
Thus, for all $i\in \{1,2\}$ and $0 \le n \le \size{\pi_i}$, 
there exists a sequence 
$\tup{\bar{q}_0,\bar{q}_{i,1} \til \bar{q}_{i,\size{\pi_i}}}$ of the $\sys{\impl}$'s states
such that $\tup{\bar{q}_0,\pi_i[1],\bar{q}_{i,1} \til \pi_i[\size{\pi_i}], \bar{q}_{i,\size{\pi_i}}}$
is an execution fragment of $\sys{\impl}$.

Let $\pi \in \pi_1 \shuffle \pi_2$.
We prove by induction that for all $0 \le l \le \size{\pi}$, the following
holds:
\begin{description}
	\item
$P(l)$: There exists an execution fragment $\alpha$ of $\sys{\impl}$ such that 
$\tracef{\alpha}=\tup{\pi[1] \til \pi[l]}$, $\alpha[1]=\bar{q}_0$, and the following
holds for some integers $m$, $n$:
\begin{enumerate}[label=(\roman*)]
\item \label{Pi} for all $p \in \Proc$:
\[
\alpha[\size{\alpha}](p)= 
\begin{cases}
    \bar{q}_{1,m}(p),& \text{if } p\in \procsf{\pi_1}\\
    \bar{q}_{2,n}(p),              & \text{otherwise}
\end{cases}
\]
\item \label{Pii} $\tup{\pi[1] \til \pi[l]} \in 
\tup{\pi_1[1] \til \pi_1[m]} \shuffle \tup{\pi_2[1] \til \pi_2[n]}$.
\end{enumerate}
\end{description}
\paragraph{Base:} Since $\alpha=\tup{\bar{q}_0}$ is an execution fragment of $\sys{\impl}$,
it remains to show that~\ref{Pi}--\ref{Pii} for some $m$ and $n$.
Let $m=n=0$. Since $\bar{q}_{i,0}=\bar{q}_0$ for $i \in \{1,2\}$,
$\bar{q}_0(p)=\bar{q}_{1,0}(p)$, if $p\in \procsf{\pi_1}$, and
$\bar{q}_0(p)=\bar{q}_{2,0}(p)$, otherwise. 
Since $\emptyseq = \emptyseq \shuffle \emptyseq$,
both~\ref{Pi} and~\ref{Pii} hold for $m=n=0$, as needed.

\paragraph{Inductive step:} Assume $P(k)$ is true for $k<\size{\pi}$, meaning that 
there exists an execution fragment 
$\tup{\bar{q_0},\pi[1],\bar{q_1} \til \pi[k], \bar{q_k}}$ of $\sys{\impl}$ 
and integers $\hat{m}$, $\hat{n}$ such that~\ref{Pi}--\ref{Pii} hold 
for $m=\hat{m}$ and $n=\hat{n}$. We will prove that $P(k+1)$ holds. 
We first show that there exists $\bar{q}_{k+1} \in \states{\sys{\impl}}$
such that $\bar{q_k} \xrightarrow{\pi[k+1]}_\sys{\impl} \bar{q}_{k+1}$.
Since~\ref{Pii} holds for $l=k$, $m=\hat{m}$, $n=\hat{n}$,
$\tup{\pi[1] \til \pi[k]} \in  \tup{\pi_1[1] \til \pi_1[\hat{m}]} \shuffle 
\tup{\pi_2[1] \til \pi_2[\hat{n}]}$, 
which, given that $\pi \in \pi_1 \shuffle \pi_2$, implies that
$\pi[k+1] \in \{\pi_1[\hat{m}+1], \pi_2[\hat{n}+1]\}$. Suppose first that
$\pi[k+1]=\pi_1[\hat{m}+1]$. Thus, we need to show that 
there exists $\bar{q}_{k+1} \in \states{\sys{\impl}}$ such that
\begin{equation}
\label{eq:case1}
\bar{q_k} \xrightarrow{\pi_1[\hat{m}+1]}_\sys{\impl} \bar{q}_{k+1}
\end{equation}
Let $p=\procf{\pi_1[\hat{m}+1]}$. 
Since $\bar{q}_{1,\hat{m}} \xrightarrow{\pi_1[\hat{m}+1]}_{\sys{\impl}} \bar{q}_{1,{\hat{m}+1}}$, 
by the transition rule in \cref{fig:lts-SI}, there exists $q' \in \states{\sysp{\impl}{p}}$ such that
\begin{equation}
\label{eq:trans-m}
\bar{q}_{1,{\hat{m}+1}}=\bar{q}_{1,\hat{m}}[p \mapsto q'],
\end{equation}
and $\bar{q}_{1,\hat{m}}(p) \xrightarrow{\pi_1[\hat{m}+1]}_{\sysp{\impl}{p}} q'$.
Since~\ref{Pi} holds for $l=k$, $m=\hat{m}$, $n=\hat{n}$, we have
$\bar{q_k}(p)=\bar{q}_{1,\hat{m}}(p)$. Thus,
$\bar{q_k}(p) \xrightarrow{\pi_1[\hat{m}+1]}_{\sysp{\impl}{p}} q'$. 
Applying transition rule in \cref{fig:lts-SI} again, we get
$$
\bar{q_k}(p) \xrightarrow{\pi_1[\hat{m}+1]}_{\sysp{\impl}{p}} \bar{q_k}[p \mapsto q'].
$$
Thus,~(\ref{eq:case1}) holds for $\bar{q}_{k+1}=\bar{q_k}[p \mapsto q']$, and therefore, 
$\tup{\bar{q_0},\pi[1],\bar{q_1} \til \pi[k], \bar{q_k}, \pi_{k+1}, \bar{q}_{k+1}}$
is an execution fragment of $\sys{\impl}$. 

Furthermore, since $\procsf{\pi_1} \cap \procsf{\pi_2} = \emptyset$, and 
$p\in \procsf{\pi_1}$, we get $p\not\in \procsf{\pi_2}$. Thus, 
for all $p' \in \procsf{\pi_2}$, $\bar{q}_{k+1}(p')=\bar{q}_k(p')$. 
Since~\ref{Pi} holds for $l=k$, $m=\hat{m}$, $n=\hat{n}$, we have
for all $p'\in \procsf{\pi_2}$,  $\bar{q}_k(p') = \bar{q}_{2,\hat{n}}(p')$.

Next, consider a process $p'\in \procsf{\pi_1}$. If $p'=p$, then
by~(\ref{eq:trans-m}), $\bar{q}_{k+1}(p')=q'=\bar{q}_{1,\hat{m}+1}$. 
Otherwise, $\bar{q}_{k+1}(p')=\bar{q}_k(p')$. 
Since~\ref{Pi} holds for $l=k$, $m=\hat{m}$, $n=\hat{n}$,
$\bar{q}_k(p')=\bar{q}_{1,\hat{m}}(p')$. Since by~(\ref{eq:trans-m}),
$\bar{q}_{1,\hat{m}}(p')=\bar{q}_{1,\hat{m}+1}(p')$, we have
$\bar{q}_{k+1}(p')=\bar{q}_{1,\hat{m}+1}(p')$.
Thus,~\ref{Pi} holds for $l=k+1$, $m=\hat{m}+1$, and
$n=\hat{n}$.

Next, since~\ref{Pii} holds for $l=k$, $m=\hat{m}$, $n=\hat{n}$, 
$\tup{\pi[1] \til \pi[k]}=\tup{\pi_1[1] \til \pi_1[m]} \shuffle \tup{\pi_2[1] \til \pi_2[n]}$,
which implies that
$\tup{\pi[1] \til \pi[k], \pi[k+1]}=
\tup{\pi[1] \til \pi[k], \pi_1[m+1]}=\\ 
\tup{\pi_1[1] \til \pi_1[m], \pi_1[m+1]} \shuffle \tup{\pi_2[1] \til \pi_2[n]}$.
Hence,~\ref{Pii} holds for $l=k+1$, $m=\hat{m}+1$, and
$n=\hat{n}$. Thus, we conclude that both~\ref{Pi} and~\ref{Pii} hold
for $l=k+1$, $m=\hat{m}+1$, and $n=\hat{n}$.

Finally, instantiating the above argument for $\pi[k+1]=\pi_1[\hat{n}+1]$, we obtain
that\\
$\tup{\bar{q_0},\pi[1],\bar{q_1} \til \pi[k], \bar{q_k}, \pi_{k+1}, \bar{q}_{k+1}}$
is an execution fragment of $\sys{\impl}$ such that 
both~\ref{Pi} and~\ref{Pii} hold for $l=k+1$, $m=\hat{m}$, and $n=\hat{n}+1$.
Thus, $P(k+1)$ is true, and therefore, $P(l)$ holds for all $0 \le l \le \size{\pi}$.

Since $P(\size{\pi})$ holds, there exists an execution fragment $\alpha$
of $\sys{\impl}$ such that $\tracef{\alpha}=\pi$ and $\alpha[1]=\bar{q_0}$.
Since $\alpha_0$ is an execution of $\sys{\impl}$, and
$\alpha_0[\size{\alpha_0}]=\bar{q_0}=\alpha[1]$,
$\alpha_0 \cdot \alpha$ is also an execution of $\sys{\impl}$.
Thus, $\tracef{\alpha_0 \cdot \alpha}=\pi_0 \cdot \pi$
is a trace of $\sys{\impl}$, as needed.
\end{proof}

\mainzero*
\begin{proof}
Let $\h_i = \restrictobj{\pi_i}$ and $\sigma_i = \restrictmem{\pi_i}$ ($i\in\set{0,1,2}$). 
We give the argument for the ``some'' case. (The proof for ``every'' is similar.)
Let $\sigma\in \sigma_1 \shuffle \sigma_2$ such that $\sigma_0 \cdot \sigma \in \otraces{\mm}$.
By \cref{lem:shuffle_restrict}, there exists $\pi \in \pi_1 \shuffle \pi_2$ such that $\restrictmem{\pi} = \sigma$.
By \cref{lem:impl_shuffle}, we have $\pi_0 \cdot \pi \in \traces{\impl}$.
Let $\h=\restrictobj{\pi}$. By \cref{lem:shuffle_restrict}, we have $\h \in \h_1 \shuffle \h_2$.
Then, since $\h_0 \cdot \h=\restrictobj{(\pi_0 \cdot \pi)}$ and
$\sigma_0 \cdot \sigma = \restrictmem{(\pi_0 \cdot \pi)}$,
it follows that $\h \in \histories{\pi_0,\impl,\mm}$.
\end{proof}

\begin{proposition}
\label{lem:shuffle_complete}
Let $\h_1,\h_2 \in \chistories{\obj}$.
Then, $\h_1 \shuffle \h_2 \subseteq \chistories{\obj}$.
\end{proposition}

\begin{proposition}
\label{lem:same_procs}
Suppose that $\pi_0 \cdot \pi \in \traces{\impl}$
for some sequences $\pi_0, \pi$ such that $\restrictobj{\pi_0} \in \chistories{\obj}$.
Then, $\procsf{\restrictobj{\pi}}=\procsf{\pi}$.
\end{proposition}

\begin{proof}
If $p\in \procsf{\restrictobj{\pi}}$, then there exists $e\in \restrictobj{\pi}$ such that
$\procf{e}=p$. Since $e\in \pi$, $p\in \procsf{\pi}$, and therefore
$\procsf{\restrictobj{\pi}} \subseteq \procsf{\pi}$.
Suppose that $p\in \procsf{\pi}$. Then, there exists $e\in \pi$ such that
$\procf{e}=p$. If $\actionf{e} \in \actsf{\obj}$, then $p\in \procsf{\restrictobj{\pi}}$, 
which means that $\procsf{\pi} \subseteq \procsf{\restrictobj{\pi}}$.

Suppose that $\actionf{e} \nin \actsf{\obj}$. Then, $\actionf{e} \in \actsf{\memacts}$. 
Let $i$ be the index of $e$ in $\pi_0 \cdot \pi$. Since $\pi_0 \cdot \pi \in \traces{\impl}$, 
there exists $j < i$ such that $(\pi_0 \cdot \pi)[j] = \ev{p}{\invi{o}}$ 
for some $o \in \opsf{\obj}$, and there does not exist $k$ such that
$j < k < i$ and $(\pi_0 \cdot \pi)[k] = \ev{p}{a}$ for some $a \in \actsf{\obj}$.
If $j \le \size{\pi_0}$, then since $\pi_0$ is complete, there exists $j < k' < i$
such that $(\pi_0 \cdot \pi)[k'] = \ev{p}{\resi{o}}$, which is a contradiction. 
Hence, $j > \size{\pi_0}$, which implies that $\ev{p}{\invi{o}} \in \pi$.
Thus, $p \in \procsf{\restrictobj{\pi}}$, and therefore, 
$\procsf{\pi} \subseteq \procsf{\restrictobj{\pi}}$.

We get that in all cases, $\procsf{\restrictobj{\pi}} \subseteq \procsf{\pi}$ and
$\procsf{\pi} \subseteq \procsf{\restrictobj{\pi}}$, and therefore,
$\procsf{\restrictobj{\pi}} = \procsf{\pi}$, as needed.
\end{proof}

\availablelin*
\begin{proof}
For ease the presentation, we will assume that all invocations and
responses in $h'$ are unique. We can easily achieve such property
by adding a fictitious timestamp to each invocation or response
of a process.

Let $\h\in \histories{\obj}$ such that $\h \sqsubseteq \h'$.
Consider a $\pi_0 \in \traces{\impl}$ such that 
$\restrictmem{\pi_0}\in\traces{\Atomic}$ and $\restrictobj{\pi_0}=\h_0$.
The assumption of the lemma implies that $\h' \in \histories{\pi_0,\impl,\Atomic}$.
Thus there is a $\pi'$ such that $\pi_0\cdot \pi' \in\traces{\impl}$,
$\restrictmem{(\pi_0\cdot \pi')} \in \otraces{\Atomic}$ and $\h'=\restrictobj{\pi'}$.
We have that $\restrictobj{\pi'} = \h' \in \reorder{\lin}{\h}$.
We will show that there exists a $\pi$ such that $\pi_0\cdot \pi \in\traces{\impl}$,
$\restrictmem{(\pi_0\cdot \pi)} \in \otraces{\Atomic}$
and $\h=\restrictobj{\pi}$, which proves the lemma.
Such $\pi$ will be constructed by gradually modifying invocation and responses of $\pi'$
until the history of $\obj$ in the final modified trace is precisely $\h$.

For any sequence $s$ and $1 \leq i,j \leq |s|$, we let $s[i,j]$ denote the subsequence of $s$ with
all its elements from position $i$ to position $j$. 
By induction on $\ell$, $0 \leq \ell \leq |\h| (= |\h'|)$, we show that there exists a $\pi^\ell$
such that:
\begin{itemize}
\item $\pi_0 \cdot \pi^\ell \in \traces{\impl}$,
\item $\restrictmem{(\pi_0 \cdot \pi^\ell)} \in \otraces{\Atomic}$, 
\item $\restrictobj{\pi^\ell} \in \reorder{\lin}{\h}$ and
\item $\restrictobj{\pi^\ell}[1, \ell] = h[1, \ell]$.
\end{itemize}
Notice $\pi^{|\h|}$ is a sequence we are looking for. 

For the base case, we let $\pi^0 = \pi'$. 
We already saw that $\pi_0 \cdot \pi^0 \in \traces{\impl}$,
$\restrictmem{(\pi_0 \cdot \pi^0)} \in \otraces{\Atomic}$, 
$\restrictobj{\pi^0} \in \reorder{\lin}{\h}$,
and $\restrictobj{\pi^0}[1, 0] = h[1, 0]$ is true by vacuity.
 
For the inductive step $\ell+1$, $0 \leq \ell < |h|$, we have two cases.
In the fist case, $\restrictobj{\pi^\ell}[\ell+1] = \h[\ell+1]$, 
and then $\restrictobj{\pi^\ell}[1, \ell+1] = h[1, \ell+1]$.
We set $\pi^{\ell+1} = \pi^\ell$, and we are done, by induction hypothesis.

In the second case, $\restrictobj{\pi^\ell}[\ell+1] \neq \h[\ell+1]$,
we modify $\pi^\ell$ to obtain~$\pi^{\ell+1}$.
Let $1 \leq i \leq |\pi^\ell|$ such that $\pi^\ell[i] = \restrictobj{\pi^\ell}[\ell+1]$.
By induction hypothesis, $\restrictobj{\pi^\ell}$ is a reorder of $\h$,
and hence $\h[\ell+1]$ appears somewhere in $\pi^\ell$.
Let $1 \leq j \leq |\pi^\ell|$ such that $\pi^\ell[j] = \h[\ell+1]$.
 Notice that~$i < j$.
 For the rest of the proof, let $e$ and $f$ denote $\pi^\ell[i]$ and $\pi^\ell[j]$, respectively.
We have the following subcases, considering that $e$ and $f$ can be invocation or responses:
\begin{itemize}

\item \underline{$e$ is an invocation and $f$ is a response.}
Observe that $e$ appears somewhere in $h[\ell+2, |\h|]$.
Thus, the response $f$ appears before the invocation $e$ in $\h$,
but they appear in opposite order in $\restrictobj{\pi^\ell}$,
which implies that $\restrictobj{\pi^\ell} \nin \reorder{\lin}{\h}$, contradicting the induction hypothesis.
Therefore, this subcase cannot happen.\\

\item \underline{$f$ is an invocation.}
Let $$\pi^{\ell+1} = \pi^\ell[1,i-1] \cdot f \cdot e \cdot \pi^\ell[i+1,j-1] \cdot \pi^\ell[j+1,|\pi^\ell|].$$

We argue that $\pi^{\ell+1}$ has the desired properties. 
\begin{itemize}
\item By construction and since $\restrictobj{\pi^\ell}[1, \ell] = h[1, \ell]$, by induction hypothesis,
we have $\restrictobj{\pi^{\ell+1}}[1, \ell+1] = h[1, \ell+1]$.

\item We have $\restrictmem{(\pi_0 \cdot \pi^{\ell+1})} \in \otraces{\Atomic}$
because $\restrictmem{(\pi_0 \cdot \pi^{\ell+1})} = \restrictmem{(\pi_0 \cdot \pi^\ell)}$, by construction,
and $\restrictmem{(\pi_0 \cdot \pi^\ell)} \in \otraces{\Atomic}$, by induction hypothesis. 

\item We now prove that $\pi_0 \cdot \pi^{\ell+1} \in \traces{\impl}$. 
Let $\proc = \procf{f}$.
Observe that there is no $i \leq k < j$ such that $\pi^\ell[k]$ is an invocation or response of $\proc$. 
If so, then observe that $\pi^\ell[k]$ appears somewhere in $h[\ell+2, |\h|]$,
and hence, in $\h$, two events of $\proc$, $f$ and $\pi^\ell[k]$, appear in an order, and they appear in the opposite
order in $\restrictobj{\pi^\ell}$, which implies that  $\restrictobj{\pi^\ell}$ does not follows the process
order in $h$; hence $\restrictobj{\pi^\ell} \nin \reorder{\lin}{\h}$, contradicting the induction hypothesis. 
Thus, $\procf{e} \neq p$. 
We conclude that there are no events of $\proc$ in $\pi^\ell[i,j-1]$.
Note that $\pi_0 \cdot \pi^{\ell+1} \in \reorder{\sproc}{\pi_0 \cdot \pi^\ell}$,
and thus Lemma~\ref{lem:impl_sproc} gives that 
$\pi_0 \cdot \pi^{\ell+1} \in \traces{\impl}$, as 
$\pi_0 \cdot \pi^\ell \in \traces{\impl}$, by induction hypothesis.

\item To prove $\restrictobj{\pi^{\ell+1}} \in \reorder{\lin}{\h}$, we use the conclusion above that 
there are no events of $\proc$ in $\pi^\ell[i,j-1]$. 
By induction hypothesis, $\restrictobj{\pi^\ell} \in \reorder{\lin}{\h}$, and thus
$\restrictobj{\pi^\ell}$ follows the process order of $\h$, and the response-invocation order of $\h$ as well.
To obtain $\pi^{\ell+1}$, only the invocation $f$ of $\proc$ at position $j$ is moved backward to position $i$.
Thus, $\restrictobj{\pi^{\ell+1}}$ follows the process order of $\h$. 
Furthermore, it follows the response-invocation order of $\h$ too 
becase any response $\pi^\ell[k]$, $i \leq k \leq j$, with $\procf{\pi^\ell[k]} \neq \proc$, appears in $h[\ell+2, |\h|]$,
hence the invocation $f$ appears before the response $\pi^\ell[k]$ in $h$ (so it is OK if the order of them is changed in $\pi^{\ell+1}$).
Therefore,  $\restrictobj{\pi^{\ell+1}} \in \reorder{\lin}{\h}$.\\
\end{itemize}

\item \underline{$e$ and $f$ are responses.}
This subcase is more tricky, it might require several modifications to~$\pi^\ell$.
Let $$\bar \pi = \pi^\ell[1,i-1] \cdot \pi^\ell[i+1,j-1] \cdot f \cdot e \cdot \pi^\ell[j+1,|\pi^\ell|].$$

The sequence $\bar \pi$ satisfies the following properties:
\begin{itemize}
\item We have $\restrictmem{(\pi_0 \cdot \bar \pi)} \in \otraces{\Atomic}$
because $\restrictmem{(\pi_0 \cdot \bar \pi)} = \restrictmem{(\pi_0 \cdot \pi^\ell)}$, by construction,
and $\restrictmem{(\pi_0 \cdot \pi^\ell)} \in \otraces{\Atomic}$, by induction hypothesis. 

\item We also have $\pi_0 \cdot \bar \pi \in \traces{\impl}$. 
Let $\proc = \procf{e}$.
Observe that there is no $i < k < j$ such that $\pi^\ell[k]$ is an invocation of $\proc$. 
If so, then observe that $\pi^\ell[k]$ appears somewhere in $h[\ell+2, |\h|]$,
and hence, in $\h$, a response, $f$, appears before an invocation, $\pi^\ell[k]$, 
and they appear in the opposite
order in $\restrictobj{\pi^\ell}$, which implies that  $\restrictobj{\pi^\ell}$ does not follows the response-invocation
order in $h$; hence $\restrictobj{\pi^\ell} \nin \reorder{\lin}{\h}$, contradicting the induction hypothesis. 
Thus, we conclude that $\procf{f} \neq p$ and there are no events of $\proc$ in $\pi^\ell[i+1,j]$.
Observe $\pi_0 \cdot \bar \pi \in \reorder{\sproc}{\pi_0 \cdot \pi^\ell}$,
and thus Lemma~\ref{lem:impl_sproc} gives that 
$\pi_0 \cdot \bar \pi \in \traces{\impl}$, as 
$\pi_0 \cdot \pi^\ell \in \traces{\impl}$, by induction hypothesis.

\item It holds $\restrictobj{\bar \pi} \in \reorder{\lin}{\h}$ too.
To prove it, we use the conclusion above that  there are no events of $\proc$ in $\pi^\ell[i+1,j]$.
By induction hypothesis, $\restrictobj{\pi^\ell} \in \reorder{\lin}{\h}$, and thus
$\restrictobj{\pi^\ell}$ follows the process order of $\h$, and the response-invocation order of $\h$ as well.
To obtain $\pi^{\ell+1}$, only the response $e$ of $\proc$ at position $i$ is moved forward to position~$j$.
Thus, $\restrictobj{\bar \pi}$ follows the process order of $\h$. 
Moreover, it follows the response-invocation order of $\h$ too because 
$e$ is moved backward in $\pi^\ell$ and hence any response-invocation pair in $\h$
that is followed in $\restrictobj{\pi^\ell}$ is still followed in  $\restrictobj{\bar \pi}$.
Therefore,  $\restrictobj{\pi^{\ell+1}} \in \reorder{\lin}{\h}$.\\
\end{itemize}

To conclude this subcase, we consider whether $\restrictobj{\bar \pi}[1, \ell+1]$ is equal to $h[1, \ell+1]$ or not.
Let $1 \leq k \leq |\bar \pi|$ such that $\bar \pi[k] = \restrictobj{\bar \pi}[\ell+1]$.
\begin{itemize}

\item If $\bar \pi[k] = f (= \h[\ell+1])$, then, by construction and since $\restrictobj{\pi^\ell}[1, \ell] = h[1, \ell]$, by induction hypothesis,
we have $\restrictobj{\bar \pi}[1, \ell+1] = h[1, \ell+1]$. Thus, we set $\pi^{\ell+1} = \bar \pi$ and we are done.

\item Otherwise, consider  $e' = \bar \pi[i]$, $1 \leq i \leq |\bar \pi|$, such that $\bar \pi[i] = \restrictobj{\bar \pi}[\ell+1]$.
Note that, by construction and induction hypothesis, $\restrictobj{\bar \pi}[1, \ell] = h[1, \ell]$.
By the first subcase, $e'$ has to be a response, so we are in the same third subcase again, now with $\bar \pi$ 
(which satisfies all properties as $\pi^\ell$), and responses $e'$ and $f$.
Hence we can repeat the same construction.
This sequence of constructions has to end, due to finiteness. Therefore, eventually, we end up with a
$\bar \pi$ satisfying $\restrictobj{\bar \pi}[1, \ell+1] = h[1, \ell+1]$ too,
and then we set $\pi^{\ell+1} = \bar \pi$.
This concludes the third subcase.
\end{itemize}
\end{itemize}
The induction step holds. This completes the proof.
\end{proof}

\main*

\begin{proof}
For each $i \in \set{1,2}$,
$\impl$is available after $\h_0$ \wrt $\seqh^i$ and $\h_i \sqsubseteq \seqh^i$.
By \cref{lem:available_lin},  we have that $\impl$ is also available after $\h_0$ \wrt $\h_i$.
Thus, there exists $\pi_i$ such that $\pi_0 \cdot \pi_i \in \traces{\impl}$,
$\restrictobj{\pi_i}=\h_i$ and $\sigma_0 \cdot \restrictmem{\pi_i}\in\traces{\Atomic}$.
By \cref{lem:same_procs}, we have that $\procsf{\h_i}=\procsf{\pi_i}$.
Moreover,  $h_0 \cdot \h_i \sqsubseteq \spec{\obj}$, as $\impl$ is consistent under $\mm$.

We now show that for every pair  $\pi_1' \in \reorder{\sproc}{\pi_1}$ and $\pi_2' \in \reorder{\sproc}{\pi_2}$
such that 
$\restrictobj{\pi'_1}=\h_1$, $\restrictobj{\pi'_2}=\h_2$, and $\sigma_0 \cdot \restrictmem{\pi'_1}, \sigma_0 \cdot \restrictmem{\pi'_2} \in \traces{\Atomic}$,
the traces $\restrictmem{\pi'_1}$ and $\restrictmem{\pi'_2}$ are not weakly (resp., strongly) mergeable in $\mm$.
Let $\sigma_1' = \restrictmem{\pi_1'}$ and $\sigma_2' = \restrictmem{\pi_2'}$.
Since $\mm$ is well-behaved, there exists a stable state $q_0 \in \states{\mm}$  such that:
\begin{enumerate}
\item $\initial{\mm} \xrightarrow{\rho_0}_\mm q_0$ for some memory sequence $\rho_0$ such that $\restrict{\rho_0}{\MemEvs}=\sigma_0$.
\item For every $\sigma'$ such that $\sigma_0 \cdot \sigma' \in \traces{\Atomic}$, $\sigma' \in \otraces{\mm,q_0}$.
\end{enumerate}
Note that, in particular, we have $\sigma_1',\sigma_2' \in \otraces{\mm,q_0}$.

Now, by contradiction, suppose that  $\sigma_1'$ and $\sigma_2'$ are weakly (resp., strongly) mergeable in $\mm$.
Thus, for some (resp., for every)  $\sigma' \in \sigma_1' \shuffle \sigma_2'$, $\sigma' \in \otraces{\mm,q_0}$.
Hence $\sigma_0 \cdot \sigma' \in \otraces{\mm}$, for any such~$\sigma'$.
For each $i \in \set{1,2}$, we have $\pi_0 \cdot \pi'_i  \in \reorder{\sproc}{\pi_0 \cdot \pi_1}$,
since $\pi'_i \in \reorder{\sproc}{\pi_i}$.
Thus, \cref{lem:impl_sproc} gives that $\pi_0 \cdot \pi'_i \in \traces{\impl}$, as we already saw that $\pi_0 \cdot \pi_i \in \traces{\impl}$.
By \cref{lem:main0}, $\h \in \histories{\pi_0,\impl,\mm}$, for some (resp., for every) $\h\in \h_1 \shuffle \h_2$.
Then, $\h_0 \cdot \h \in \histories{\emptyseq,\impl,\mm}$, for any such $\h$.
Since $\h_1$ and $\h_2$ are complete, by \cref{lem:shuffle_complete}, $\h$ is complete too,
and hence $\h_0 \cdot \h$ is also complete.
Moreover, since $\impl$ is consistent under $\mm$,
we have that $\h_0 \cdot \h \sqsubseteq \spec{\obj}$.
This contradicts that $\h_1$ and $\h_2$ are not weakly (resp., strongly) mergeable in $\spec{\obj}$ after $\h_0$,
becase we already saw that $\h_0 \cdot \h_1 \sqsubseteq \spec{\obj}$ and $\h_0 \cdot \h_2 \sqsubseteq \spec{\obj}$.
\end{proof}


%% file: apps-strong.tex
\section{Proofs for Section \ref{sec:app_weak}}
\label{app:app_weak}

\lemmaweaknoncommandmerge*
\begin{proof}
Since $\op_1$ is one-sided non-commutative  \wrt $\op_2$ in $\spec{\obj}$, there exist
$\h_0\in \cshistories{\obj}$, processes $\proc_1 \neq \proc_2$, and response values $u_1,v_1,u_2\in\retsf{\obj}$
such that $u_1 \neq v_1$,
$\h_1=\h_0 \cdot \invres{\proc_1}{\op_1}{u_1} \in \spec{\obj}$, and
$\h_2=\h_0 \cdot \invres{\proc_2}{\op_2}{u_2} \cdot \invres{\proc_1}{\op_1}{v_1} \in \spec{\obj}$.
Since $\spec{\obj}$ is prefix-closed, we have $\h_2'=\h_0 \cdot \invres{\proc_2}{\op_2}{u_2} \in \spec{\obj}$.
Since $\h_1, \h_2'\in \spec{\obj}$, we have to show that there exists 
$\h \in \invres{\proc_1}{\op_1}{u_1} \shuffle \invres{\proc_2}{\op_2}{u_2}$ such that
$\h_0 \cdot \h \nin \spec{\obj}$. Let $\h=\invres{\proc_2}{\op_2}{u_2} \cdot \invres{\proc_1}{\op_1}{u_1}$,
and assume by contradiction that 
$\h_0 \cdot \h \in \spec{\obj}$.
Then, since $u_1 \neq v_1$, $\h_0 \cdot \invres{\proc_2}{\op_2}{u_2} \cdot \tup{\proc_1, \invi{\op_1}}$, 
the longest common prefix of $\h_2$ and $\h_0 \cdot \h$,
ends in an invocation. Since both $\h_2, \h_0 \cdot \h \in \spec{\obj}$,
this contradicts the fact that $\obj$ is deterministic.
\end{proof}

\thmweaknoncomm*
\begin{proof}
Since $\op_1$ is one-sided non-commutative \wrt $\op_2$ in $\spec{\obj}$,
by \cref{lemma-weak-non-comm-and-merge}, 
there exist processes $\proc_1 \neq \proc_2$,  
response values $u_1,u_2\in\retsf{\obj}$,
and history $\h_0 \in \cshistories{\obj}$
such that $\h_1=\invres{\proc_1}{\op_1}{u_1}$ and $\h_2=\invres{\proc_2}{\op_2}{u_2}$ 
are not strongly mergeable in $\spec{\obj}$ after $\h_0$.

Since $\spec{\obj}$ is prefix-closed, $\h_0\in \spec{\obj}$. 
Since $\impl$ is spec-available, there exists $\pi_0\in \traces{\impl}$
such that $\restrictobj{\pi_0} = \h_0$ 
and $\sigma_0=\restrictmem{\pi_0} \in \otraces{\Atomic}$.
Moreover, since $\impl$ is spec-available, we also 
we get that $\impl$ is available after $\h_0$ \wrt both $\h_1$ and $\h_2$.

Thus, by \cref{thm:main}, there exist 
$\pi_1 \in \traces{\impl(\op_1,\proc_1)}$ and $\pi_2\in \traces{\impl(\op_2,\proc_2)}$
such that $\sigma_1=\restrictmem{\pi_1}$ and $\sigma_2=\restrictmem{\pi_2}$
are not strongly mergeable in $\mm$. 
Then, the required follows
from property $\TSO^\textsf{s}$ in \cref{tab:results} for $\mm=\TSO$
or from properties $\RA^\textsf{s}_1,\RA^\textsf{s}_2,\RA^\textsf{s}_3$ in \cref{tab:results} for $\mm=\RA$.
\end{proof}

\begin{restatable}{theorem}{thmstrongimpl}
\label{thm:strong-impl}
For $\obj \in \{\rego, \mrego\}$, we have:
\begin{enumerate*}[label=(\roman*)]
\item there exists a linearizable wait-free implementation of $\obj$ under $\TSO$ that
uses only reads, writes, and a single fence at the end of\hspace{2pt} $\writeoo$, and
\item there exists a linearizable wait-free implementation of $\obj$ under $\RA$ that
uses only reads, writes, and a pair of fences at both the 
beginning and the end of\hspace{2pt} both $\writeoo$ and $\reado$.
\end{enumerate*}
\end{restatable}
\begin{proof}
\cref{coro-from-scm-to-tso} directly proves (a).
Under \Atomic, $\rego$ can trivially be wait-free linearizable implemented using a single register.
\cref{coro-from-scm-to-tso} implies that if a fence at the end of $\writeoo$ is added,
the implementations remains wait-free and linearizable under \TSO. 
For $\mrego$, we consider the wait-free linearizable implementation under \Atomic of~Aspnes et al~\cite{AAC12},
whose $\reado$ operation executes a sequence of writes and 
$\writeoo$ operation executes a sequence of reads followed by a sequence of writes. 
Again, \cref{coro-from-scm-to-tso} implies that if a fence at the end of $\writeoo$ is added,
the implementations remains wait-free and linearizable in \TSO. 

For (b), by~\cref{theo-max-reg},
the wait-free collect based implementation with fences at the beginning and end of each operation in~\cref{alg:max-reg},
is linearizable for $\mrego$ under \RA. 
For $\rego$, the obvious wait-free implementation that uses a single register and fences as the theorem specifies, is linearizable
(the linearizability uses a similar and simpler idea than the linearizability proof of~\cref{alg:max-reg}).
\end{proof}


%% file: simulation.tex
\section{Simulating $\Atomic$ in $\TSO$ with Few Fences}
\label{sec-simulation}

In this section we show a simple systematic transformation that transfers linearizability under \Atomic
to linearizability under \TSO. The transformation takes an implementation and adds a fence
in between any pair of consecutive write  and read, and any pair of consecutive write and response of an operation.

\begin{lemma}
\label{lemma-from-scm-to-tso}
Consider an implementation $\impl$ of an object $\obj$.
Let $\impl'$ be the implementation obtained from $\impl$ by adding a fence in between 
every pair of consecutive write and read or response.
Then, $\histories{\impl,\Atomic} = \histories{\impl',\TSO}$. 
\end{lemma}

\begin{proof}
We first observe that $\histories{\impl,\Atomic} \subseteq \histories{\impl',\TSO}$.
Since fences in $\Atomic$ are no-ops, we have $\histories{\impl,\Atomic} = \histories{\impl',\Atomic}$,
and since $\TSO$ is weaker than $\SCM$, \ie $\otraces(\SCM) \subset \otraces{\TSO}$,
we have $\histories{\impl',\SCM} \subseteq \histories{\impl',\TSO}$.
Thus, $\histories{\impl,\Atomic} \subseteq \histories{\impl',\TSO}$.

We now prove that $\histories{\impl',\TSO} \subseteq \histories{\impl,\Atomic}$.
Since fences are no-op in $\SCM$, it suffices to show $\histories{\impl',\TSO} \subseteq \histories{\impl',\Atomic}$.
Let $\pi'$ be any trace of $\impl'$ such that $\restrictmem{\pi'} \in \otraces{\TSO}$. 
We argue that $\restrictobj{\pi'} \in \histories{\impl',\Atomic}$.
Let $\rho$ be a sequence of memory and object events, and propagation events  
such that (1) it is $\pi'$ when propagation events are removed from it,
and (2) it belongs to $\traces{\TSO}$ when restricted to memory and propagation events.
Such sequence exists since $\restrictmem{\pi'} \in \otraces{\TSO}$.
Since $\impl'$ only adds to $\impl$ fences  in between pairs of consecutive write and read or response,
for every operation of~$\pi'$:
\begin{enumerate} 
\item every read reads its value from main memory, and
\item writes are propagated in FIFO order to main memory before the end of the operation.
\end{enumerate}
We now modify $\rho$ as follows: first, every write is moved forward, right before 
the silent memory event that propagates it to main memory, and second
all propagation events are removed. Let us call $\pi$ the resulting sequence.
As store buffers in $\TSO$ are propagated in FIFO order,
we have $\pi \in \reorder{\sproc}{\pi'}$ and then 
$\pi \in \traces{\impl'}$, by \cref{lem:lin_shuffle}.
Also, we have $\restrictmem{\pi} \in \otraces{\TSO}$.
Observe that, for every process $\proc$, $\restrict{\pi}{\proc} = \restrict{\pi'}{\proc}$,
since the writes of each process are moved forward in the order the were issued. 
Moreover, the relative order of invocations and responses in $\pi$ and $\pi'$ is the same, 
as writes are moved within the invocation and response of the operation it belongs to,
\ie $\restrictobj{\pi} = \restrictobj{\pi'}$. 
We have that in $\pi$ every write is visible to every process at the moment the
write is issued. This is precisely how writes behave in $\Atomic$.
Then, $\restrictmem{\pi} \in \otraces{\Atomic}$, and hence $\restrictobj{\pi'} \in \histories{\impl',\Atomic}$.
\end{proof}

The following corollary considers the usual definition of linearizability, 
and not the restricted notion of consistency in \cref{thm:main}.

\begin{corollary}
\label{coro-from-scm-to-tso}
Consider any implementation $\impl$ of an object $\obj$. 
Let $\impl'$ be the implementation obtained from $\impl$ by adding a fence in between 
every pair of consecutive write and read or response.
Then, $\impl$ is obstruction-free/lock-free/wait-free linearizable under \Atomic if and only if 
$\impl'$ is obstruction-free/lock-free/wait-free linearizable under \TSO.
\end{corollary}

\begin{proof}
Clearly, the added fences in $\impl'$ do not change the progress properties of $\impl$.
Thus, $\impl$ and $\impl'$ have the same progress properties. 
By \cref{lemma-from-scm-to-tso}, $\histories{\impl,\Atomic} = \histories{\impl',\TSO}$,
and hence either both $\impl$ and $\impl'$ are linearizable, or none of them is linearizable.
\end{proof}

We now consider object linearizable implementations under $\Atomic$ 
in which each operation performs a sequence of reads followed by a sequence of writes.
We call such implementations \emph{read-then-write}.
Formally, a memory sequence $\sigma$ is read-then-write (RTW) if for every $k_1 < k_2$,
if $(1)$ $\typf{\sigma[k_1]}=\Write$ and 
$(2)$ $\procf{\sigma[k_1]} = \procf{\sigma[k_2]}$,
then $\typf{\sigma[k_2]}\neq\Write$.
Note that RTW implies RBW but not vice versa.  
Then, an implementation is RTW if 
for every $\op\in\opsf{\obj}$ and $\proc\in\Proc$, all memory traces in $\impl(\op,\proc)$ are RW and RTW.
By \cref{coro-from-scm-to-tso}, these implementations under \SCM implements the same object under $\TSO$
by adding a fence at the end of each operation (note that the modified implementation is not RTW).

In $\Atomic$, there have been proposed several lock-free or wait-free (both strictly stronger than obstruction-free) 
RTW linearizable implementations of objects with weakly mergeable pairs of histories.
The following are examples we are aware of:

\begin{itemize}
\item The classical wait-free snapshot implementation of Afek et al~\cite{AADGMS93}. 
In this implementation, $\mathtt{Scan}$ repeatedly performs sets of reads 
until the process is able to identify a consistent snapshot; and
$\mathtt{Update}$ first invokes $\mathtt{Scan}$ to take a snapshot, and then executes a single write.
Thus, the memory traces generated by both operations are RTW.

\item The universal construction of the object types in which every pair of operations either commute
or one overwrites the other~\cite{AH90}. (\emph{Counter} is an example of an object type
satisfying this condition.) 
The construction of~\cite{AH90} relies on a single wait-free snapshot object. To execute an operation,
a process invokes $\mathtt{Scan}$, performs some local computation, and then executes $\mathtt{Update}$.
Thus, when instantiated with the snapshot implementation of~\cite{AADGMS93}, 
this construction guarantees that the memory traces produced by every operation
of the original object are RTW.

\item The wait-free $m$-bounded max-register implementation of Aspnes et al~\cite{AAC12}. $\mathtt{Read}$ performs $O(\log m)$ reads, and $\mathtt{MaxWrite}$ performs a sequence of $O(\log m)$ reads followed by $O(\log m)$ writes.

\item The wait-free relaxed FIFO work-stealing of \cite{CP21}.  $\mathtt{Read}$ performs two writes, 
and $\mathtt{Take}$ and $\mathtt{Steal}$ perform two reads followed by a single write.
\end{itemize}

By Corollary~\ref{coro-from-scm-to-tso}, 
these implementations can easily be adapted to correctly implement the same objects in $\TSO$.
Finally, it is easy to prove that these objects have pairs of operations that are one-sided non-commutative,
and hence Theorem~\ref{thm:weak-non-comm} shows that any lock-free read/write implementation
of them in \TSO must use fences.


%% file: max-RA.tex
\iffull
\section{Fence-Optimal Implementations under \RA}
\label{app:impl_under_RA}

\subsection{Fence-Optimal Max Register Under $\RA$}
\label{app:max_reg}
\fi

\newcommand{\opf}[1]{{\mathsf{op}({#1})}}

The pseudocode of a linearizable wait-free implementation of $\mrego$ 
under \RA is given in~\cref{alg:max-reg}. The function 
$collect(M)$ reads one by one, in an arbitrary order, the entries of $M$, and
returns an array with the read values. The algorithm 
is fence-optimal. It uses one fence at the beginning and one fence
at the end of every operation, thus matching the lower bounds of~\cref{thm:weak-non-comm}.

\begin{algorithm}[h!]
\caption{\small $\mrego$ implementation in \RA. Algorithm for proces $p_i$.}
\label{alg:max-reg}
\begin{algorithmic}[1]\small
	\begin{multicols}{2}
\Statex Shared variables:
\Statex  \hspace{0.2cm} $int[n] \, \, \, M = [0, \hdots, 0]$ 
\Statex
\Procedure{read}{$\,$}
	\State {\sf fence()}
	\State $m[] = collect(M)$
	\State {\sf fence()}
	\State \Return $max(m[])$
\EndProcedure
\Procedure{write}{$v$}
	\State {\sf fence()}
	\State $m[] = collect(M)$
	\If{$max(m[]) < v$}
		\State $M[i] = v$
	\EndIf
	\State {\sf fence()}
	\State \Return $\ack$
\EndProcedure
\end{multicols}

\end{algorithmic}
\end{algorithm}

\iffull
\begin{theorem}
\label{theo-max-reg}
\cref{alg:max-reg} is a wait-free linearizable implementation of $\mrego$ under \RA.
\end{theorem}

\begin{proof}
Clearly, the implementation is wait-free. For linearizability, consider 
any trace $\pi$ of the implementation such that $\restrictmem{\pi}\in\otraces{\RA}$. 
Consider any continuation $\pi'$ of
$\pi$ that completes all pending operation of $\pi$ (and starts no new operation) and
$\restrictmem{\pi'}\in\otraces{\RA}$.
Such continuation exists due to wait-freedom.
Observe that any linearization of $\pi'$ is a linearization of $\pi$.
Thus, we can assume that  $\pi$ has no pending operations.

To linearize $\pi$, we focus on write events (of $\writeo$ operations) and
read events that read the maximal value of the collect it belongs to (of $\writeo$ and $\reado$ operations).
These are the events in $\pi$ that determine the changes in the state of the object implemented.
We define a strict partial order $T$ on this set of events such that 
every total order on the operations that extends $T$ induces a valid sequential history of $\mrego$.
Any such valid sequential history is a linearization of $\pi$ because $T$ respects real-time
order of operations, \ie the relative order of operations in $\pi$.

Let $J \defeq \set{1 \til \size{\pi}}$.
For every $j \in J$, let $\valwf{j}$ denote the value written/read in event $\pi[j]$ (if it is a read or write event).
For every $j \in J$, let $\opf{j}$ denote the set of indices in $J$ that belong to the 
operation (either $\writeo{v}$ or $\reado$) that contains the event $\pi[j]$.
Let:
\begin{itemize}
\item $W \defeq \set{ j \in J \st \typf{\pi[j]}=\Write }$.
\item $R \defeq \set{ j \in J \st \typf{\pi[j]}=\Read \text{ and $\valwf{j}$ is the maximal value 
among all read events in $\opf{j}$} }$.
\end{itemize}

The relation $T$ relates $j_1, j_2 \in W \cup R$, if one of the following holds:
\begin{itemize}
\item $\valwf{j_1} < \valwf{j_2}$; 
\item $j_1 \in W$, $j_2 \in R$, and $\valwf{j_1} = \valwf{j_2}$;
or \item$j_1 \in R$, $j_2 \in R$, $\valwf{j_1} = \valwf{j_2}$, and $j_1 < j_2$.
\end{itemize}

It is not difficult to check that $T$ is transitive and irreflexive.
Also, observe that a non-$0$ read in the trace always obtains its value from a previous write.
Formally, for every $r \in R$ with $\valwf{r}\neq 0$, there exists some $w\in W$ such that 
$\valwf{w}=\valwf{r}$ and $wTr$.

Now, we claim that every total order on the set of operations, $\set{\opf{j} \st j \in J}$, extending $T$ induces
a valid sequential history of $\mrego$.
Indeed, in every such total order $S \supseteq T$, every $\reado$ operation that returns a non-$0$ value,
obtains its value from a previous $\writeo$ operation,
and if $\opf{j_1}$ is a $\writeo$ operation, $\opf{j_2}$ is a $\writeo$ operation,
and $j_1 S j_2$, we must have $\valwf{j_1} \leq \valwf{j_2}$
(or else, $j_2 T j_1$).

To conclude that every total order $S$ is a linearization of $\pi$, we argue that 
the real-time order of operations in $\pi$ is a subset of $T$. This implies
that if the response of $\opf{j_1}$ precedes the invocation of $\opf{j_2}$,
then $\opf{j_1}$ appears before $\opf{j_2}$ in $S$.
First, events of the same thread are increasing in value order.
Second, to see that that order between response events and subsequent invocation events is included in $T$,
suppose that $j_1$ is a response event and $j_2>j_1$ is an invocation event.
Then, due to the placement of fences in the algorithm, every access in $\opf{j_2}$ is aware
(\ie later in the ``happens-before'' order) than every access in $\opf{j_1}$.
Then, if $j_2\in W$, then $\opf{j_2}$ only writes a larger value, so $\valwf{j_1} < \valwf{j_2}$.
Otherwise, if $j_2\in R$, then $\opf{j_2}$ reads a greater of equal value than $\opf{j_1}$, so $\valwf{j_1} \leq \valwf{j_2}$.
In both case, we have $j_1 T j_2$.
This concludes the proof.
\end{proof}

\subsection{Fence-Optimal Register Under $\RA$}
\label{app:reg_under_RA}

\begin{theorem}
\label{thm:reg_RA}
Consider the implementation of a register object that uses a single memory location
for the register contents; implements write by storing to memory, read by loading from memory;
and has fences at the beginning and at the end of both write and read.
Then, this implementation is a wait-free linearizable implementation of $\rego$ under \RA.
\end{theorem}

\begin{proof}
Clearly, the implementation is wait-free. For linearizability, consider 
any trace $\pi$ of the implementation such that $\restrictmem{\pi}\in\otraces{\RA}$. 
Consider any continuation $\pi'$ of
$\pi$ that completes all pending operation of $\pi$ (and starts no new operation) and
$\restrictmem{\pi'}\in\otraces{\RA}$.
Such continuation exists due to wait-freedom.
Observe that any linearization of $\pi'$ is a linearization of $\pi$.
Thus, we can assume that  $\pi$ has no pending operations.

Let $J \defeq \set{1 \til \size{\pi}}$.
For every $j \in J$, let $\opf{j}$ denote the set of indices in $J$ that belong to the 
operation (either $\writeo{v}$ or $\reado$) that contains the event $\pi[j]$.
Let:
\begin{itemize}
\item $W \defeq \set{ j \in J \st \typf{\pi[j]}=\Write }$.
\item $R \defeq \set{ j \in J \st \typf{\pi[j]}=\Read}$.
\end{itemize}

The relation $T$ relates $j_1, j_2 \in W \cup R$, if one of the following holds:
\begin{itemize}
\item $j_1 \in W$, $j_2 \in R$, and $\pi[j_2]$ reads its value from $\pi[j_1]$.
\item $\pi[j_2]$ has timestamp larger than $\pi[j_1]$
(either by picking a timestamp when writing, or by obtaining the message of the write when reading).
\end{itemize}

It is not difficult to check that $T$ is transitive and irreflexive.
Also, observe that a non-$0$ read in the trace always obtains its value from a previous write.
Formally, for every $r \in R$ with $\valwf{r}\neq 0$, there exists some $w\in W$ such that 
$\valwf{w}=\valwf{r}$ and $wTr$.

Moreover, by inspecting the \RA semantics, it is straightforward to see that 
every total order on the set of operations, $\set{\opf{j} \st j \in J}$, extending $T$
 induces a valid sequential history of $\rego$;
 and that due to the placement of fences in the implementation 
 the real-time order of operations in $\pi$ is a subset of $T$.
Indeed, the fences induce synchronization between the events, 
which means that only a larger timestamp may be picked in the second operation
if it is a write, and only a larger or equal timestamp may be observed in the second
operation if it is a read.
\end{proof}
\fi

%% file: apps-weak.tex
\iffull
\section{Proofs and More Details for Section \ref{sec:app_two}}
\label{sec:non-strong-impl}

\subsection{Two-sided non-commutativity}
\label{sec:two-sided}

\lemmastrongnoncommandmerge*
\begin{proof}
Let $\h_0 \in \cshistories{\obj}$, $\proc_1, \proc_2\in\Proc$, and $u_1,v_1,u_2,v_2\in\retsf{\obj}$
that satisfy the conditions of \cref{def-strongly-non-commutative}.
We argue that $\invres{\proc_1}{\op_1}{u_1}$ and $\invres{\proc_2}{\op_2}{u_2}$ 
are not weakly mergeable in $\spec{\obj}$ after $\h_0$.
First, we have that $\h_0 \cdot \invres{\proc_1}{\op_1}{u_1} \in \spec{\obj}$ 
and $\h_0 \cdot \invres{\proc_2}{\op_2}{u_2} \in \spec{\obj}$ (since $\spec{\obj}$ is prefix-closed).
Second, there cannot be $\h \in \invres{\proc_1}{\op_1}{u_1} \shuffle  \invres{\proc_2}{\op_2}{u_2}$ 
with $\h_0 \cdot \h \sqsubseteq \spec{\obj}$.
Indeed, otherwise, either $\h_0 \cdot \invres{\proc_1}{\op_1}{u_1} \cdot  \invres{\proc_2}{\op_2}{u_2}$ 
or  $\h_0 \cdot  \invres{\proc_2}{\op_2}{u_2} \cdot \invres{\proc_1}{\op_1}{u_1}$ are in $\spec{\obj}$,
but this contradicts the fact that $\obj$ is deterministic, as 
the longest common prefix of 
$\h_0 \cdot \invres{\proc_1}{\op_1}{u_1} \cdot \invres{\proc_2}{\op_2}{u_2}$ 
and $\h_0 \cdot \invres{\proc_1}{\op_1}{u_1} \cdot \invres{\proc_2}{\op_2}{v_2}$ 
ends with an invocation to $\op_2$, and similarly, the longest common prefix of 
$\h_0 \cdot \invres{\proc_2}{\op_2}{u_2} \cdot \invres{\proc_1}{\op_1}{u_1}$ 
and $\h_0 \cdot \invres{\proc_2}{\op_2}{u_2} \cdot \invres{\proc_1}{\op_1}{v_1}$
ends with an invocation of $\op_1$.
\end{proof}

\thmtwosidelower*
\begin{proof}
By \cref{lemma-strong-non-comm-and-merge},
there exist $\h_0 \in \cshistories{\obj}$, processes $\proc_1 \neq \proc_2$, and response values $u_1,u_2\in\retsf{\obj}$
such that $\invres{\proc_1}{\op_1}{u_1}$ and $\invres{\proc_2}{\op_2}{u_2}$ 
are not weakly mergeable in $\spec{\obj}$ after $\h_0$.
By \cref{thm:main}, it follows that there exist 
$\pi_1 \in \traces{\impl(\op_1,\proc_1)}$ and $\pi_2 \in \traces{\impl(\op_2,\proc_2)}$
such that $\restrictmem{\pi_1}$ and $\restrictmem{\pi_2}$ are not weakly mergeable in $\mm$.
Then, the required follows from properties $\SCM^\textsf{w}$, $\TSO^\textsf{w}$, and $\RA^\textsf{w}$ in \cref{tab:results}).
\end{proof}

\paragraph{Upper bounds.}
It is well-known that a linearizable wait-free implementation of an arbitrary
(deterministic) data type can be obtained in $\Atomic$ using an array of wait-free
consensus objects $\mathit{CONS}[]$ such that the $k$th operation to execute  
is determined by the outcome of $\mathit{CONS}[k]$.

On the other hand, an obstruction-free consensus can 
be implemented from registers in $\Atomic$, e.g., using a shared memory
implementation of the Paxos algorithm~\cite{disk-paxos}. 
The latter simulates each of the prepare and accept phases of the Paxos ballot
using a write followed by a collect. Since the first ballot can have an a priori
assigned leader, it only requires a single accept phase. Thus, in the absence
of contention, shared memory Paxos can reach a decision using a single write-collect
block. Since this algorithm can be converted into 
a $\TSO$-based implementation by applying the fence insertion strategy
described in~\cref{{sec:app_weak}}, we have the following:

\begin{theorem}
\label{thm:two-sided-upper}
Let $\obj$ be an object with a pair of two-sided non-commutative operations $\op_1$ and $\op_2$.
Then, there exists a linearizable obstruction-free implementation $\impl_\mm$ under a memory
model $\mm$, where:
\begin{enumerate}[leftmargin=*,label=(\alph*)]
\item \label{thm:two-side-upper:atomic}
$\impl_{\Atomic}$ uses a sequence of repeated write-collect blocks 
to implement either $\op_1$ or $\op_2$, and only executes a single
write-collect in the absence of contention.

\item \label{thm:two-side-upper:tso}
$\impl_{\TSO}$ uses a sequence of repeated write-fence-collect blocks 
to implement either $\op_1$ or $\op_2$, and only executes a single
write-fence-collect in the absence of contention.
\end{enumerate}
\end{theorem}

\noindent Since consensus has a pair of two-sided non-commutative operations, it is subject
to the lower bound of \cref{thm:two-side-lower}, which implies that in the absence of contention, 
the upper bound of \cref{thm:two-sided-upper} is tight. 
An obstruction-free universal 
construction under $\RA$ with a fence pattern optimal in our lower bounds 
is left for  future work.
\fi

\iffull
\subsection{Mutual Exclusion}
\label{sec:mutex}
\fi
We use the merge theorem (\cref{thm:main}) for the case of non-weakly mergeable histories
and the mergeability results for the memory models to establish minimum synchronization requirements 
for mutual exclusion. Our result for $\Atomic$ reproves the corresponding lower bound of~\cite{AGHKMV11}. 

Consider a (non-standard) lock object $\locko$ with 
$\opsf{\locko} \defeq \set{\acquireo}$ and 
$\retsf{\locko} \defeq \set{\ack}$.
Its specification is given by
\begin{center}
$\spec{\locko} \defeq \set{\emptyseq} \cup \{\invress{\proc}{\acquireo}{}{40pt} \st \proc\in\Proc\}$.
\end{center}
The histories $\invress{\proc}{\acquireo}{}{40pt}$ and 
$\invress{\proc'}{\acquireo}{}{40pt}$ where $\proc\neq \proc'$ are not weakly mergeable.
Thus, by the merge theorem and 
properties $\SCM^\textsf{w}$, $\TSO^\textsf{w}$, and $\RA^\textsf{w}$ in \cref{tab:results}, we have:
\begin{theorem}
Let $\impl$ be a spec-available implementation of $\locko$ that is consistent under a memory model $\mm$.
Then, there exist $\proc \in\Proc$ and $\pi\in \traces{\impl(\acquireo, \proc)}$  
such that the following hold for $\sigma=\restrictmem{\pi}$: 
\begin{enumerate}[leftmargin=*,label=(\alph*)]
\item
if $\mm = \Atomic$, then 
$\sigma$ either has an RMW event or is not RBW; and
\item
if $\mm \in \set{TSO,\RA}$, then 
$\sigma$ either has an RMW event or a fence.
\end{enumerate}
\end{theorem}

\iffull
The proof of this theorem is identical to that of~\cref{thm:two-side-lower}.
\fi
Since the implementation of the entry section of a mutual exclusion algorithm 
can be used to implement $\acquireo$, we obtain that  entry section of 
a solo-terminating mutual exclusion algorithm on \SCM has to
use a RAW pattern or an RMW;
and on \TSO or \RA, it must use an RMW or a fence.

There exist many algorithms implementing starvation-free mutual exclusion under $\SCM$, which 
use the RAW pattern to implement the entry section. As before, their 
counterparts under $\TSO$ can be obtained by adding a fence between every pair of 
consecutive write and read (\cref{sec:app_weak}). For example, 
the transformation of Bakery algorithm~\cite{bakery} only requires a single
fence to separate a write-only block at the beginning of the entry section from 
the read-only block right afterwards. The resulting implementation
is therefore tight. 
Mutual exclusion under $\RA$ with an RMW or a fence has several 
verified implementations~\cite{Lahav:pldi19}.


%% file: snap-counter-app.tex
\section{Lower and Upper Bounds for Snapshot and Counter}
\label{sec:snap-count}

\paragraph{Lower bounds for snapshot.}
Consider a (single-writer) snapshot object $\snapshoto$ storing a
vector of a length $\size{\Proc}$ over a set of values $W$
(also represented as function in $\Proc \to W$)
 with the initial vector of $\tup{\bot \til \bot}$. 
The operations are
$\set{\updateopp{w} \st w\in V} \cup \set{\scanop}$, 
and its return values are $\set{\ack} \cup (\Proc \to W)$. 
The specification $\spec{\snapshoto}$ 
consists of all complete sequential histories where each $\scanop$ event returns $v$
such that $v(\proc)$ is the value written by the last preceding $\updateop$ by process $\proc$, 
or $\bot$ if no such $\updateop$ exists. 

\begin{proposition}
\label{prop:snap-merge}
Let $w, w' \in W$, $\proc_1,\proc_2,\proc_3\in\Proc$, and $\h_1, \h_2 \in \chistories{\snapshoto}$,
such that $w \neq w'$, $i \neq j$, $\procsf{\h_1} \cap \procsf{\h_2} = \emptyset$, and the following hold:
\begin{itemize}
\item
$\h_1 \sqsubseteq \invress{\proc_1}{\updateopp{w}}{}{60pt} \cdot \invress{\proc_3}{\snapop}{v}{30pt}$, 
where $v=\lambda \proc \ldotp \ite{\proc=\proc_1}{w}{\bot}$; and
\item
$\h_2  \sqsubseteq \invress{\proc_2}{\updateopp{w'}}{}{60pt} \cdot \invress{\proc_2}{\snapop}{v'}{30pt}$,
where $v'=\lambda \proc \ldotp \ite{\proc=\proc_2}{w'}{\bot}$.
\end{itemize}
Then, $\h_1$ and $\h_2$ are not weakly mergeable in $\spec{\snapshoto}$ after $\emptyseq$.
\end{proposition}

Next, we use the merge theorem (instantiated for the case of non-weakly mergeable histories) 
together with \cref{prop:snap-merge} 
and the mergeability results 
$\SCM^\textsf{w}$, $\TSO^\textsf{w}$, and $\RA^\textsf{w}$
from \cref{tab:results}
to establish lower bounds on implementability of snapshot.

\thmsnapweaktsoscmra*
\begin{proof}
First, consider the case of $\mm \in \set{\Atomic,\TSO}$.
Let $\proc_1, \proc_2$ be distinct processes and consider the histories
\begin{center}
$\h_1=\invress{\proc_1}{\updateopp{1}}{}{60pt} \cdot \invress{\proc_1}{\snapop}{v}{30pt}$ \quad and \quad
$\h_2=\invress{\proc_2}{\updateopp{1}}{}{60pt} \cdot \invress{\proc_2}{\snapop}{v'}{30pt}$,
\end{center}
where $v=\lambda \proc \ldotp \ite{\proc=\proc_1}{w}{\bot}$ and $v'=\lambda \proc \ldotp \ite{\proc=\proc_2}{w'}{\bot}$.
Then, by \cref{prop:snap-merge}, $\h_1$ and $\h_2$ are not weakly mergeable in $\spec{\snapshoto}$
after $\h_0=\emptyseq$.
Clearly, we also have 
$\h_1, \h_2 \in \spec{\snapshoto}$, and since $\impl$ is spec-available,
it is available \wrt both $\h_1$ and $\h_2$.

Thus, by \cref{thm:main}, there exist $\pi'_1, \pi'_2 \in \traces{\impl}$ such that 
$\h_1=\restrict{\pi'_1}{\snapshoto}$ and $\h_2=\restrict{\pi'_2}{\snapshoto}$, and
$\sigma_1'=\restrictmem{\pi_1'}$ and $\sigma_2'=\restrictmem{\pi_2'}$ are not weakly
mergeable in $\mm$.
Observe that $\pi_1'=\pi_1 \cdot \pi_2$ where 
$\pi_1\in \traces{\impl(\updateopp{1}, \proc_1)}$ and
$\pi_2 \in \traces{\impl(\snapop, \proc_1)}$.
Let $\sigma_1=\restrictmem{\pi_1}$ and $\sigma_2=\restrictmem{\pi_2}$.
Then, $\sigma_1'=\sigma_1 \cdot \sigma_2$. 
Thus, the required follows properties
$\SCM^\textsf{w}$ and $\TSO^\textsf{w}$ in \cref{tab:results}.

\medskip
Next, we consider the case of $\mm=\RA$.
Let $\proc_1, \proc_2, \proc_3 \in \Proc$ be distinct processes, and consider the histories:

$\h_1 = \tup{\ev{\proc_1}{\invi{\updateopp{1}}},\ev{\proc_3}{\invi{\scanop}}, \ev{\proc_1}{\resi{\ack}},\ev{\proc_3}{\resi{v}}}$ and

$\h_2 =\invress{\proc_2}{\updateopp{2}}{}{60pt} \cdot \invress{\proc_2}{\snapop}{v'}{30pt}$,

\noindent
where $v=\lambda \proc \ldotp \ite{\proc=\proc_1}{w}{\bot}$ and $v'=\lambda \proc \ldotp \ite{\proc=\proc_2}{w'}{\bot}$.
Then, by \cref{prop:snap-merge}, $\h_1$ and $\h_2$ are not weakly mergeable in $\spec{\snapshoto}$
after $\h_0=\emptyseq$.
Note that $\h_2 \in \spec{\snapshoto}$. Consider the following sequential history of $\spec{\snapshoto}$: \medskip

\hfill{}$\seqh^1 = \invress{\proc_1}{\writesno{1}}{}{50pt} \cdot
\invress{\proc_3}{\scanop}{v}{30pt}
$\hfill{}

\medskip \noindent By assumption, $\impl$ is available \wrt $h_2$ and $\seqh^1$.

Then, by \cref{thm:main}, there exist $\pi_1$ and $\pi_2$ such that:
\begin{itemize}
\item $\restrictobj[\snapshoto]{\pi_i}=\h_i$ for $i\in\set{1,2}$.
\item $\pi_i \in \traces{\impl}$ for $i\in\set{1,2}$.
\item $\restrictmem{\pi_i} \in \traces{\Atomic}$ for $i\in\set{1,2}$.
\item $\procsf{\pi_i} = \procsf{\h_i}$ for $i\in\set{1,2}$.
\item For every $\pi_1' \in \reorder{\sproc}{\pi_1}$ such that $\restrictmem{\pi_1'}\in \traces{\Atomic}$ and $\restrictobj[\snapshoto]{\pi_1}=\h_1$ 
and $\pi_2' \in \reorder{\sproc}{\pi_2}$ such that $\restrictmem{\pi_2'}\in \traces{\Atomic}$ and $\restrictobj[\snapshoto]{\pi_2}=\h_2$,
$\restrictmem{\pi_1'}$ and $\restrictmem{\pi_2'}$ are not weakly mergeable in $\RA$.
\end{itemize}

Let $\pi_1'$ be the sequence obtained from $\pi_1$ by:
\begin{itemize}
\item  moving 
$\tup{\proc_1,\invi{\updateopp{1}}}$, $\tup{\proc_3,\invi{\scanop}}$ and all leading fences to the beginning of the sequence; and 
\item  moving 
$\tup{\proc_1,\resi{\ack}}$, $\tup{\proc_3,\resi{v}}$ and all trailing fences to the end of the sequence.
\end{itemize}
In this rearrangement we keep the internal order among moved events as it is in $\pi_1$.
Then, $\pi_1' \in \reorder{\sproc}{\pi_1}$ and $\restrictobj[\snapshoto]{\pi_1}=\h_1$.
Moreover, among memory events, we only moved fences which are no-ops under $\Atomic$.
Thus, $\restrictmem{\pi_1} \in \traces{\Atomic}$ implies $\restrictmem{\pi_1'}\in \traces{\Atomic}$.
By taking $\pi_2'=\pi_2$, we obtain that 
$\restrictmem{\pi_1'}$ and $\restrictmem{\pi'_2}$ are not weakly mergeable in $\RA$.
Finally, by property $\RA^\textsf{w}$ in \cref{tab:results}, we obtain that 
$\restrictmem{\pi'_1}$ is not LTF, or it contains some RMW event.
This implies that either $\restrict{\restrictmem{\pi'_1}}{\proc_1}$ or  $\restrict{\restrictmem{\pi'_1}}{\proc_3}$ are not LTF or contain some RMW event.
\end{proof}

\paragraph{Lower bounds for counter.}
Consider a counter object $\counto$ with the initial value of $0$, and
the increment ($\inco$), decrement ($\deco$), and read ($\reado$)
operations. Then, we have:

\begin{proposition}
\label{prop:count-merge}
Let $\proc_1,\proc_2,\proc_3\in\Proc$ and $\h_1, \h_2 \in \chistories{\counto}$
such that $\procsf{\h_1} \cap \procsf{\h_2} = \emptyset$ and the following hold:
\begin{itemize}
\item $\h_1 \sqsubseteq \invress{\proc_1}{\inco}{}{40pt} \cdot \invress{\proc_3}{\reado}{1}{40pt}$; and
\item $\h_2 \sqsubseteq \invress{\proc_2}{\deco}{}{40pt} \cdot \invress{\proc_2}{\reado}{-1}{40pt}$
\end{itemize}
Then, $\h_1$ and $\h_2$ are not weakly mergeable in $\spec{\counto}$ after $\emptyseq$.
\end{proposition}

Then, the following can be obtained by instantiating 
the proof of \cref{thm:snap-weak-tso-scm-ra} to use \cref{prop:count-merge}.

\begin{theorem}
\label{thm:count-weak-tso-scm-ra}
Let $\impl$ be a spec-available implementation of $\counto$ that is consistent under 
a memory model $\mm$. 
Then, there exist $\proc, \proc'\in\Proc$, 
$\pi_1\in \traces{\impl(\op, \proc)}$ where $\op \in \{\inco, \deco\}$ and
$\pi_2 \in \traces{\impl(\reado, \proc')}$ 
such that the following hold
for $\sigma_1=\restrictmem{\pi_1}$ and $\sigma_2=\restrictmem{\pi_2}\}$:
\begin{enumerate}[label=(\alph*)]
\item
if $\mm = \Atomic$, then $\sigma_1 \cdot \sigma_2$ either has a RMW event 
or is not RBW; and
\item
if $\mm = \TSO$, then $\sigma_1 \cdot \sigma_2$ has either a RMW event 
or is non-LTF (\ie has a fence in the middle).
\item \label{count:snap-weak-ra}
if $\mm = \RA$, then 
\begin{enumerate*}[label=(\roman*)]
\item 
either $\sigma_1$ or $\sigma_2$ has an RMW, or
\item 
either $\sigma_1$ or $\sigma_2$ is non-LTF.
\end{enumerate*}
\end{enumerate}
\end{theorem}

\paragraph{Upper bounds.}
There is a wait-free snapshot implementation~\cite{AADGMS93} that is linearizable under \SCM,
in which $\snapop$ performs a sequence of reads, and 
$\updateop$ performs a sequence of reads followed by a write,
Using the fence insertion strategy in~\cref{sec:app_weak},
a linearizable wait-free implementation of snapshot under $\TSO$ is obtained
from such implementations by adding a single fence at the end of $\updateop$.

\begin{theorem}
\label{thm:snap-upper}
For $\mm \in \{\Atomic, \TSO\}$, there exists a linearizable wait-free implementation of snapshot 
$\snapshoto_{\mm}$ under $\mm$ such that: 
\begin{enumerate}[leftmargin=*,label=(\alph*)]
\item \label{thm:snap-upper:atomic}
$\snapshoto_{\Atomic}$ uses only 
a sequence of reads followed by a write to implement $\updateop$ and only reads to implement $\snapop$, and

\item \label{thm:snap-upper:tso}
$\snapshoto_{\TSO}$ uses only a sequence of reads followed by a write and a fence at the end to implement
$\updateop$, and only reads to implement $\snapop$.
\end{enumerate}
\end{theorem}

Observe that~\cref{thm:snap-upper}~\ref{thm:snap-upper:atomic} implies 
that any pair of consecutive
$\updateop$ and $\snapop$ is RBW, which is tight in the lower bound of 
\cref{thm:snap-weak-tso-scm-ra}~\ref{thm:snap-weak-scm}. 
Likewise,~\cref{thm:snap-upper}~\ref{thm:snap-upper:tso} is tight in the lower bound
of \cref{thm:snap-weak-tso-scm-ra}~\ref{thm:snap-weak-tso}, which stipulates that a fence is needed somewhere within
consecutively executed $\updateop$ and $\snapop$.

A linearizable wait-free counter can be implemented on top of a snapshot instance as follows: each process
$\proc_i$ stores its contribution to the current counter value in a local variable $c_i$ initialized
to $0$. To increment (resp., decrement) the counter, $\proc_i$ increments (resp., decrements) 
$c_i$, and then invokes $\updateopp{c_i}$ to share its contribution with other processes.
To read the counter, a process calls $\snapop$ and returns the sum of the values stored
in the returned vector. 
\begin{theorem}
\label{thm:counter-upper}
For $\mm \in \{\Atomic, \TSO\}$, there exists a linearizable wait-free implementation of counter 
$\counto_{\mm}$ under $\mm$ such that: 
\begin{enumerate}[label=(\alph*)]
\item
$\counto_{\Atomic}$ uses only 
writes to implement $\inco$ and $\deco$ and only reads to implement $\reado$, and

\item $\counto_{\TSO}$ uses only writes and a fence at the end to implement
$\inco$ and $\deco$, and only reads to implement $\reado$.
\end{enumerate}
\end{theorem}
As in the case of snapshot, the synchronization strategy stipulated by this
result is optimal \wrt the lower bound of \cref{thm:count-weak-tso-scm-ra}.
The optimal implementations of  snapshot and counter  under $\RA$ 
are left for future work.


%% file: cons-num-two-sided.tex
\section{Consensus Number of Two-Sided Non-Commutative Operations}
\label{app:cn-two-sided}

We establish the following theorem:

\begin{theorem}
The consensus number of every deterministic object with a pair of two-sided non-commutative 
operations is at least 2.
\label{thm:consnu}
\end{theorem}
\begin{proof}
Let $\obj$ be a deterministic object with operations
$\op_1, \op_2 \in \opsf{\obj}$ that are two-sided non-commutative in $\spec{\obj}$.
Then, there exist
a history $\h_0 \in \cshistories{\obj}$, processes $\proc_1 \neq \proc_2$, 
and response values $u_1 \neq v_1$ and $u_2 \neq v_2$ in $\retsf{\obj}$
such that
$\h_1 = \h_0 \cdot \invres{\proc_1}{\op_1}{u_1} \cdot \invres{\proc_2}{\op_2}{v_2} \in \spec{\obj}$ and
$\h_2 = \h_0 \cdot \invres{\proc_2}{\op_2}{u_2} \cdot \invres{\proc_1}{\op_1}{v_1} \in \spec{\obj}$.

Let $\impl$ be a linearizable wait-free implementation of $\obj$ on \Atomic.
We give an algorithm $A$ on $\Atomic$ for two processes, $\proc_1$ and $\proc_2$, that solves consensus 
using $\impl$ and a two-element shared array $M$ of read/write variables. 
First, $\impl$ is initialized by executing $\h_0$,
which is possible since $\impl$ is wait-free and $\obj$ is deterministic.
To solve consensus, each process $p_i$, $i\in \{1, 2\}$, first writes its input $x_i$ in $M[i]$,
then performs operation $\op_i$ on $\impl$. If $\op_i$ returns
$u_i$, then $p_i$ decides $M[i]$, otherwise, $p_i$ decides
$M[j]$ where $j=2$ if $i=1$, and $j=1$, otherwise.

Since $\impl$ is wait-free, $A$ is wait-free. We now argue that $A$ 
satisfies the validity and agreement properties of consensus. 
We consider the possibility of process crashes (as usual in consensus number proofs). 

We first prove the following:
\begin{claim}
For all $i \in \{1, 2\}$, 
if a process decides the value stored in $M[i]$, then $M[i]=x_i$. 
\label{lemma:consnu-aux}
\end{claim}

\begin{claimproof}
If a process $p_i$ decides
the value stored in $M[i]$, then by the time this happens, $p_i$ has already
written its input value to $M[i]$, and therefore, $M[i]=x_i$, as needed. 
Let $\h$ be a history of $\impl$ in which $p_j$ decides the value stored in $M[i]$ where $j\neq i$,
and assume without loss of generality that $j=1$ and $i=2$. 
Then, the invocation of $\op_1$ by $\proc_1$ is 
complete and returns $u_1' \neq u_1$ in $\h$. Since $\impl$ is a linearizable
implementation of $\obj$, there exists a sequential and complete 
history $\h'$ such that $\h \sqsubseteq \h_0 \cdot \h' \sqsubseteq \spec{\obj}$. 
Since $\obj$ is deterministic, $\h_0 \cdot \h'\in \spec{\obj}$, 
and $u_1' \neq u_1$,
$\h'$ cannot be of the form $\invres{\proc_1}{\op_1}{u_1'} \cdot \h''$.
Hence, there exists $u_2'$ such that $\h' = \invres{\proc_2}{\op_2}{u_2'} \cdot \invres{\proc_1}{\op_1}{u_1'}$.
Since $\h \sqsubseteq \h_0 \cdot \h'$, $\op_2$ must be invoked before
$\op_1$ is completed in $\h$. Thus, by the time $\op_1$ returns, $p_2$ must have already completed
writing $M[2]$, and therefore, $M[2]=x_2$.  The result follows since $p_1$ can only
decide $M[2]$ after $\op_1$ returns.
\end{claimproof}

Now, since a process always decides a value stored in either $M[1]$ or $M[2]$, by 
\cref{lemma:consnu-aux}, validity holds. Moreover, if only one process decides and 
the other one crashes before deciding, agreement holds as well. Suppose that both processes
decide and assume by contradiction that agreement is violated. 
Then, there are two cases to consider:

\begin{enumerate}

\item $p_1$ decides $M[1]=x_1$ and $p_2$ decides $M[2]=x_2$. Then, both the 
invocation of $\op_1$ by $\proc_1$ and the invocation of $\op_2$ by $\proc_2$ 
are complete and return $u_1$ and $u_2$, respectively. Since $\impl$ is a linearizable
implementation of $\obj$, either
\emph{(i)} $\h_1'=\h_0 \cdot \invres{\proc_1}{\op_1}{u_1} \cdot \invres{\proc_2}{\op_2}{u_2} \in \spec{\obj}$ or
\emph{(ii)} $\h_2'=\h_0 \cdot \invres{\proc_2}{\op_2}{u_2} \cdot \invres{\proc_1}{\op_1}{u_1} \in \spec{\obj}$.
Suppose \emph{(i)} holds. Since $\h_1 \in \spec{\obj}$ and $u_2 \neq v_2$, $\h_1$ and $\h_1'$ have
the longest common prefix that ends in an invocation contradicting the fact that $\obj$
is deterministic. Likewise, if \emph{(ii)} holds, then since 
$\h_2 \in \spec{\obj}$ and $u_1 \neq v_1$, $\h_2$ and $\h_2'$ have
the longest common prefix that ends in an invocation contradicting the determinism of $\obj$.

\item $p_1$ decides $M[2]=x_2$ and $p_2$ decides $M[1]=x_1$. 
Then, both the 
invocation of $\op_1$ by $\proc_1$ and the invocation of $\op_2$ by $\proc_2$ 
are complete and return $u_1' \neq u_1$ and $u_2' \neq u_2$, respectively. 
Since $\impl$ is a linearizable implementation of $\obj$, either
\emph{(i)} $\h_1'=\h_0 \cdot \invres{\proc_1}{\op_1}{u_1'} \cdot \invres{\proc_2}{\op_2}{u_2'} \in \spec{\obj}$ or
\emph{(ii)} $\h_2'=\h_0 \cdot \invres{\proc_2}{\op_2}{u_2'} \cdot \invres{\proc_1}{\op_1}{u_1'} \in \spec{\obj}$.
Suppose \emph{(i)} holds. Since $\h_1 \in \spec{\obj}$ and $\obj$ is deterministic, we have 
$u_1'=u_1$, which is a contradiction. Likewise, if \emph{(ii)} holds, then $\h_2 \in \spec{\obj}$ and the fact that $\obj$ is 
deterministic imply that $u_2'=u_2$, which is a contradiction.
\qedhere
\end{enumerate}

\end{proof}